\documentclass{article}
\usepackage{amssymb}
\usepackage{mathtools}
\mathtoolsset{showonlyrefs}

\usepackage{amsthm} 

\usepackage{ifxetex,ifluatex}
\ifxetex
\usepackage{xltxtra}
\else
\ifluatex
\usepackage{luatextra}
\else
\usepackage[utf8]{inputenc}
\usepackage[T1]{fontenc}
\usepackage{newtxtext}
\usepackage[vvarbb]{newtxmath}
\csname else\expandafter\endcsname
\romannumeral -`0
\fi
\fi
\iftrue\relax%
\defaultfontfeatures{Ligatures=TeX}
\usepackage[math-style=TeX]{unicode-math}
\AtBeginDocument{%
}
\fi
\providecommand{\allOne}{\mathbb{1}}

\usepackage[nolists,nomarkers,figuresonly]{endfloat}
\extrafloats{100}

\usepackage{booktabs}       
\usepackage{microtype}      
\usepackage{natbib}         

\usepackage{authblk}

\usepackage{etoolbox}
\usepackage{color}
\usepackage{algorithm}
\usepackage{algorithmic}
\usepackage{enumitem}

\usepackage{geometry}
\usepackage[hyphens]{url}   
\usepackage{hyperref} 
\usepackage{color}
\usepackage{graphicx}

\hypersetup{
  pdftitle={Lazifying Conditional Gradient Algorithms},
  pdfauthor={Gábor Braun, Sebastian Pokutta, Daniel Zink},
  pdfkeywords={Frank–Wolfe algorithm, conditional gradient,
    caching, linear optimization oracle, convex optimization},
  pdfsubject={MSC2000 Primary 68Q32; 
    Secondary 68W27, 
    90C52 
  }}

\title{Lazifying Conditional Gradient Algorithms}

\author{Gábor Braun}
\author{Sebastian Pokutta}
\author{Daniel Zink}
\affil{ISyE, Georgia Institute of Technology\\
  Atlanta, GA\\
  \texttt{\{gabor.braun,sebastian.pokutta,daniel.zink\}@gatech.edu}}

\setlist[enumerate]{label=(\roman*)}
\newlist{enumerate*}{enumerate*}{1}
\setlist[enumerate*]{label=(\arabic*),
  after=.,
  itemjoin={{, }}, itemjoin*={{, or }}}

\floatstyle{ruled}
\newfloat{oracle}{htbp}{ora}
\floatname{oracle}{Oracle}

\newcommand{\AUG}[1]{\ensuremath{\operatorname{AUG}\sb{#1}}}
\newcommand{\LP}[1]{\ensuremath{\operatorname{LP}\sb{#1}}}
\newcommand{\LPsep}[1]{\ensuremath{\operatorname{LPsep}\sb{#1}}}
\newcommand{\LLPsep}[1]{\ensuremath{\operatorname{LLPsep}\sb{#1}}}

\newcommand{\argmin}[1]{\ensuremath{\operatorname{argmin}_{#1}}}
\newcommand{\argmax}[1]{\ensuremath{\operatorname{argmax}_{#1}}}

\newtheorem{theorem}{Theorem}[section]
\newtheorem{proposition}[theorem]{Proposition}
\newtheorem{corollary}[theorem]{Corollary}
\newtheorem{lemma}[theorem]{Lemma}
\theoremstyle{definition}

\theoremstyle{remark}
\newtheorem{remark}[theorem]{Remark}

\DeclarePairedDelimiterX{\normSimple}[1]{\lVert}{\rVert}
{\ifblank{#1}{\mathord{\cdot}}{#1}}

\DeclarePairedDelimiter{\abs}{\lvert}{\rvert}

\newcommand{\norm}[2][]{\normSimple{#2}\ifblank{#1}{}{\sb{#1}}}
\newcommand{\Norm}[2][]{\normSimple*{#2}\ifblank{#1}{}{\sb{#1}}}
\newcommand{\dualnorm}[2][]{\mathinner{%
    \normSimple{#2}\ifblank{#1}{}{\sb{#1}}\sp{*}}}

\newcommand{\supp}[1]{\operatorname{supp}(#1)}

\newcommand{\ball}[2]{\mathbb{B}_{#1}\left({#2}\right)}

\newcommand{\R}{\mathbb R}

\newcommand{\expectation}[1]{\mathbb{E}\left[#1\right]}

\DeclareMathOperator{\dom}{dom}

\begin{document}
\maketitle{}
\begin{abstract}
  Conditional gradient algorithms (also often called Frank-Wolfe
  algorithms) are popular due to their simplicity of only requiring a
  linear optimization oracle and more recently they
  also gained significant traction for online learning. While simple
  in principle, in many cases the actual implementation of the linear
  optimization oracle is costly.
  We show a general method to \emph{lazify}
  various conditional gradient algorithms,
  which in actual computations leads to several orders of
  magnitude of speedup in wall-clock time.
  This is achieved by using a faster separation oracle
  instead of a linear optimization oracle,
  relying only on \emph{few} linear optimization oracle calls.
\end{abstract}

\section{Introduction}
\label{sec:introduction}

Convex optimization is an important technique both from a theoretical
and an applications perspective. Gradient descent based
methods are widely used due to their simplicity and easy
applicability to many real-world problems.  We are interested in
solving constraint convex optimization problems of the form
\begin{equation}
\label{eq:mainProb}
  \min_{x \in P} f(x),
\end{equation}
where \(f\)
is a smooth convex function and \(P\)
is a polytope, with access to \(f\)
being limited to first-order information,
i.e., we can obtain \(\nabla f(x)\)
and \(f(x)\) for a given \(x \in P\) and access to \(P\)
via a linear minimization oracle, which returns
\(\LP{P}(c) = \operatorname{argmin}_{x \in P} c x\)
for a given linear objective \(c\).

\begin{algorithm}
  \caption{Frank-Wolfe Algorithm \citep{frank1956algorithm}}
  \label{alg:FW-basic}
  \begin{algorithmic}[1]
    \REQUIRE
      smooth convex function \(f\) with curvature \(C\),
      start vertex \(x_{1} \in P\),
      linear minimization oracle \(\LP{P}\)
    \ENSURE points \(x_{t}\) in \(P\)
    \FOR{\(t=1\) \TO \(T-1\)}
      \STATE
        \(v_t \leftarrow \LP{P}(\nabla f(x_{t}))\)
         \STATE \(x_{t+1} \leftarrow (1 - \gamma_{t}) x_{t} + \gamma_{t}
          v_{t}\) with \(\gamma_t \coloneqq \frac{2}{t+2}\)
    \ENDFOR
  \end{algorithmic}
\end{algorithm}

When solving Problem~\eqref{eq:mainProb} using gradient descent
approaches in order to maintain feasibility, typically a projection
step is required.  This projection back into the feasible region \(P\)
is potentially computationally expensive, especially for complex
feasible regions in very large dimensions.  As such, projection-free
methods gained a lot of attention recently, in particular the
Frank-Wolfe algorithm \citep{frank1956algorithm} (also known as
conditional gradient descent \citep{levitin1966constrained}; see also
\citep{jaggi2013revisiting} for an overview) and its online version
\citep{hazan2012projection} due to their simplicity. We recall the
basic Frank-Wolfe algorithm in 
Algorithm~\ref{alg:FW-basic}. These methods eschew the projection step and rather use a
linear optimization oracle to stay within the feasible region. While
convergence rates and regret bounds are often suboptimal, in many
cases the gain due to only having to solve a \emph{single} linear
optimization problem over the feasible region in every iteration still
leads to significant computational advantages (see e.g.,
\citep[Section~5]{hazan2012projection}). This led to conditional
gradient algorithms being used for e.g., online optimization
and more generally machine learning.
Also the property that these algorithms
naturally generate sparse distributions over the extreme points of the
feasible region is often helpful.
Further
increasing the relevance of these methods, it was shown recently that
conditional gradient methods can also achieve linear convergence (see
e.g., \cite{garber2013linearly,FW-converge2015,LDLCC2016}) as well as
that the number of total gradient evaluations can be reduced while
maintaining the optimal number of oracle calls as shown in
\cite{lan2014conditional}.

Unfortunately, for complex feasible regions even solving the linear
optimization problem might be time-consuming and as such the cost of
solving the LP might be non-negligible. This could be the case, e.g.,
when linear optimization over the feasible region is a hard
problem or when solving large-scale optimization or learning
problems. As such it is natural to ask the following questions:
\begin{enumerate}
\item\label{item:2}
  Does the linear optimization oracle have to be called in every
  iteration?
\item\label{item:3}
  Does one need approximately optimal solutions for convergence?
\item\label{item:4}
  Can one reuse information across iterations?
\end{enumerate}

We will answer these questions in this work, showing that
\ref{item:2} the LP oracle is not required to be called
in every iteration,
\ref{item:3} much weaker guarantees are sufficient,
and \ref{item:4} we can reuse information. To
significantly reduce the cost of oracle calls \emph{while} maintaining
identical convergence rates up to small constant factors, we replace
the linear optimization oracle by a \emph{(weak) separation oracle}
\begin{oracle}
  \caption{Weak Separation Oracle \(\LPsep{P}(c, x, \Phi, K)\)}
  \label{alg:weak-separate-oracle}
  \begin{algorithmic}
    \REQUIRE
    linear objective \(c \in \mathbb{R}^{n}\),
    point \(x \in P\),
    accuracy \(K \geq 1\),
    objective value \(\Phi > 0\);
    \ENSURE Either
    \begin{enumerate*}
    \item\label{item:positive}
      vertex \(y \in P\) with
      \(c (x - y) > \Phi / K\)
    \item\label{item:negative}
      \FALSE: \(c (x - z) \leq \Phi\) for all \(z \in P\)
    \end{enumerate*}
  \end{algorithmic}
\end{oracle}
(Oracle~\ref{alg:weak-separate-oracle}) which approximately solves
a \emph{separation problem} within a multiplicative factor and returns
improving vertices.  We stress
that the weak separation oracle is significantly weaker than
approximate minimization, which has been already considered in
\cite{jaggi2013revisiting}.  In fact, there is no guarantee that the
improving vertices returned by the oracle are near to the optimal
solution to the linear minimization problem.  It is this relaxation of
dual bounds and approximate optimality that will provide a significant
speedup as we will see later. However, if the oracle does not return
an improving vertex (returns \algorithmicfalse), then this fact can be
used to derive a reasonably small dual bound of the form:
\(f(x_{t}) - f(x^{*}) \leq \nabla f(x_{t}) (x_{t} - x^{*}) \leq
\Phi_{t}\) for some \(\Phi_t > 0\).  While the accuracy \(K\) is
presented here as a formal argument of the oracle, an oracle
implementation might restrict to a fixed value \(K > 1\), which often
makes implementation easier. We point out that the cases
\ref{item:positive} and \ref{item:negative} potentially overlap if
\(K > 1\). This is intentional and in this case it is unspecified which
of the cases the oracle should choose (and it does not matter for
the algorithms).

This new oracle encapsulates the smart use of
the \emph{original} linear optimization oracle,
even though for some problems
it could potentially be implemented directly without relying on
a linear programming oracle.
Concretely, a weak separation oracle can
be realized by a single call to a linear optimization oracle and as
such is no more complex than the original oracle. However it has two
important advantages: it allows for \emph{caching} and \emph{early
  termination}.  Caching refers to storing previous solutions, and
first searching among them to satisfy the oracle's
separation condition. The underlying linear optimization oracle is
called only, when none of the cached solutions satisfy the
condition.  Algorithm~\ref{alg:LPSepLPOracle} formalizes this process.
Early termination is the technique to stop the linear
optimization algorithm before it finishes at an appropriate stage,
when from its internal data a suitable oracle answer can be easily
recovered; this is clearly an implementation dependent technique.
The two techniques can be combined, e.g.,
Algorithm~\ref{alg:LPSepLPOracle} could use an early terminating
linear oracle or other implementation of the weak separation oracle in
line \ref{line:early-stop}.
\begin{algorithm}
  \caption{\(\LPsep{P}(c, x, \Phi, K)\) via LP oracle}
  \label{alg:LPSepLPOracle}
  \begin{algorithmic}[1]
    \REQUIRE
      linear objective \(c \in \mathbb{R}^{n}\),
      point \(x \in P\),
      accuracy \(K \geq 1\),
      objective value \(\Phi > 0\);
    \ENSURE Either
    \begin{enumerate*}
    \item vertex \(y \in P\) with
      \(c (x - y) > \Phi / K\)
    \item \FALSE: \(c (x - z) \leq \Phi\) for all \(z \in P\)
    \end{enumerate*}
    \IF{\(y \in P \) cached with \(c (x - y) > \Phi / K\) exists}
      \RETURN \(y\) \COMMENT{Cache call}
    \ELSE
    \STATE \(y \leftarrow \argmax{x \in P} c x\)
      \label{line:early-stop}
    \COMMENT{LP call}
    \IF{\(c (x - y) > \Phi / K\)}
      \STATE add \(y\) to cache
      \RETURN \(y\)
    \ELSE
      \RETURN \FALSE
    \ENDIF
    \ENDIF
  \end{algorithmic}
\end{algorithm}

We call \emph{lazification} the technique of replacing a linear
programming oracle with a much weaker one,
and we will demonstrate
significant speedups in wall-clock performance (see
e.g., Figure~\ref{fig:cacheEffect}), while maintaining identical theoretical
convergence rates.

To exemplify our approach we provide conditional gradient algorithms
employing the weak separation oracle for the standard Frank-Wolfe
algorithm as well as the variants in
\citep{hazan2012projection,LDLCC2016,garber2013linearly}, which have
been chosen due to requiring modified convergence arguments that go
beyond those required for the vanilla Frank-Wolfe algorithm.
Complementing the theoretical analysis we report computational results
demonstrating effectiveness of our approach via a significant
reduction in wall-clock time compared to their linear optimization
counterparts.

\subsection*{Related Work}
\label{sec:related-work}

There has been extensive work on Frank-Wolfe algorithms and conditional gradient
algorithms, so we will restrict to review work most
closely related to ours. The Frank-Wolfe algorithm was originally
introduced in \citep{frank1956algorithm} (also known as conditional
gradient descent \citep{levitin1966constrained} and has been intensely
studied in particular in terms of achieving stronger convergence
guarantees as well as affine-invariant versions. We demonstrate our
approach for the vanilla Frank-Wolfe algorithm
\citep{frank1956algorithm} (see also \citep{jaggi2013revisiting}) as
an introductory example. We then consider more complicated variants
that require non-trivial changes to the respective convergence proofs
to demonstrate the versatility of our approach. This includes the
linearly convergent variant via local linear optimization
\citep{garber2013linearly} as well as the pairwise conditional
gradient variant of \cite{LDLCC2016}, which is especially efficient in
terms of implementation. However, our technique also applies to the
\emph{Away-Step Frank-Wolfe} algorithm, the \emph{Fully-Corrective
  Frank-Wolfe} algorithm, the \emph{Pairwise Conditional Gradient}
algorithm, as well as the \emph{Block-Coordinate Frank-Wolfe}
algorithm. Recently, in \cite{Freund2016} guarantees for arbitrary
step-size rules were provided and an analogous analysis can be also
performed for our approach. On the other hand, the analysis of the
inexact variants, e.g., with approximate linear minimization does not
apply to our case as our oracle is significantly weaker than
approximate minimization as pointed out earlier. For more
information, we refer the interested reader to the excellent overview
in \citep{jaggi2013revisiting} for Frank-Wolfe methods in general as
well as \cite{FW-converge2015} for an overview with respect to global
linear convergence.

It was also recently shown in \cite{hazan2012projection} that the
Frank-Wolfe algorithm can be adjusted to the online learning setting
and in this work we provide a lazy version of this algorithm.
Combinatorial convex online optimization has been investigated in a
long line of work (see e.g.,
\citep{kalai2005efficient,audibert2013regret,neu2013efficient}). It is
important to note that our regret bounds hold in the \emph{structured online
learning} setting, i.e., our bounds depend on the \(\ell_1\)-diameter
or sparsity of the polytope, rather than its ambient dimension for
arbitrary convex functions (see e.g.,
\citep{cohen2015following,gupta2016solving}).  We refer the interested
reader to \citep{ocoBook} for an extensive overview.

A key component of the new oracle is the ability to cache and reuse
old solutions, which accounts for the majority of the observed speed
up. The idea of caching of oracle calls was already explored in
various other contexts such as cutting plane methods (see e.g.,
\cite{joachims2009cutting}) as well as the \emph{Block-Coordinate
  Frank-Wolfe} algorithm in
\cite{shah2015multi,osokin2016minding}. Our lazification approach
(which uses caching) is however much more lazy, requiring
no multiplicative approximation
guarantee; see \cite[Proof of Theorem 3. Appendix
F]{osokin2016minding} and \cite{lacoste2013block} for comparison to
our setup.

\subsection*{Contribution}
\label{sec:contribution}

The main technical contribution of this paper is a new approach,
whereby instead of finding the optimal solution, the oracle is used
only to find a \emph{good enough solution} or a \emph{certificate}
that such a solution does not exist, both ensuring the desired
convergence rate of the conditional gradient algorithms.

Our contribution can be summarized as follows:

\begin{enumerate}
\item \emph{Lazifying approach.} We provide a general method to lazify conditional gradient
algorithms. For this we replace the linear optimization oracle with
a weak separation oracle, which allows us to reuse feasible solutions
from previous oracle calls, so that in many cases the oracle call can
be skipped. In fact, once a simple representation of the underlying
feasible region is learned no further oracle calls are needed. We also
demonstrate how parameter-free variants can be obtained.

\item \emph{Lazified conditional gradient algorithms.} We exemplify our
approach by providing lazy versions of the vanilla Frank-Wolfe
algorithm as well as of the conditional gradient
methods in \citep{hazan2012projection,garber2013linearly,LDLCC2016}.

\item  \emph{Weak separation through augmentation.} We show in the case of
0/1 polytopes how to implement a weak separation oracle with at most
\(k\)
calls to an augmentation oracle that on input \(c \in \R^n\)
and \(x\in P\)
provides either an improving solution \(\overline{x} \in P\)
with \(c\overline{x} < cx\)
or ensures optimality, where \(k\) denotes the \(\ell_{1}\)-diameter
of \(P\).  This is useful when the solution space is sparse.

\item \emph{Computational experiments.} We demonstrate computational
superiority by extensive comparisons of the weak separation based
versions with their original versions. In all cases we report
significant speedups in wall-clock time often of several
orders of magnitude.
\end{enumerate}

It is important to note that in all cases, we inherit the same
requirements, assumptions, and properties of the baseline algorithm
that we lazify. This includes applicable function classes, norm
requirements, as well as smoothness and (strong) convexity
requirements. We also maintain identical convergence rates up to
(small) constant factors.

A previous version of this work appeared as extended abstract in
\cite{bpz2017conf}; this version has been significantly revised over
the conference version including a representative subset of
more extensive computational results,
full proofs for all described variants, as well as a variant that uses
an augmentation oracle instead of linear optimization oracle (see
Section~\ref{sec:weak-separate-augment}).

\subsection*{Outline}
\label{sec:outline}

We briefly recall notation and notions in Section~\ref{sec:preliminaries} and
consider conditional gradient algorithms in
Section~\ref{sec:offline-frank-wolfe}. In
Section~\ref{sec:frank-wolfe-with} we consider parameter-free
variants of the proposed algorithms, and in
Section~\ref{sec:online-fw-error} we examine online versions.
Finally, in
Section~\ref{sec:weak-separate-augment} we show a realization of
a weak separation oracle with an even weaker oracle in the case of
combinatorial problem and we provide extensive computational results
in Section~\ref{sec:experiments}. 

\section{Preliminaries}
\label{sec:preliminaries}

Let \(\norm{\cdot}\)
be an arbitrary norm on \(\R^{n}\),
and let \(\dualnorm{}\) denote the dual norm of \(\norm{}\).
A function \(f\)
is \emph{\(L\)-Lipschitz}
if \(\abs{f(y) - f(x)} \leq L \norm{y - x}\)
for all \(x,y \in \dom f\).
A convex function \(f\)
is \emph{smooth} with \emph{curvature} at most \(C\)
if
\[f(\gamma y + (1-\gamma) x) \leq f(x) + \gamma \nabla f(x) (y - x) + C
\gamma^{2} / 2\]
for all \(x,y \in \dom f\)
and \(0 \leq \gamma \leq 1\).
A function \(f\)
is \emph{\(S\)-strongly
  convex} if
\[f(y) - f(x) \geq \nabla f(x) (y - x) + \frac{S}{2} \norm{y -
  x}^{2}\]
for all \(x,y \in \dom f\).
Unless stated otherwise Lipschitz continuity and strong convexity will
be measured in the norm \(\norm{\cdot}\). Moreover, let \(\ball{r}{x} \coloneqq \{y
\mid \norm{x-y} \leq r\}\) be the ball around \(x\) with
radius \(r\) with respect to \(\norm{.}\). 
In the following, \(P\) will denote the feasible region, a polytope
and the vertices of \(P\) will be denoted by
\(v_{1}, \dots, v_{N}\).

\section{Lazy Conditional Gradient}
\label{sec:offline-frank-wolfe}

We start with the most basic Frank-Wolfe algorithm as a simple
example for lazifying by means
of a \emph{weak separation oracle}.  We then lazify
more complex Frank-Wolfe algorithms in \cite{garber2013linearly} and
\cite{LDLCC2016}.  Throughout this section
\(\norm{\cdot}\) denotes the \(\ell_{2}\)-norm.

\subsection{Lazy Conditional Gradient: a basic example}
\label{sec:lazy-c-basic}

We start with lazifying the original Frank-Wolfe algorithm (arguably
the simplest Conditional Gradient algorithm), adapting the baseline
argument from \cite[Theorem~1]{jaggi2013revisiting}.  While the
vanilla version has suboptimal convergence rate \(O(1/T)\), its
simplicity makes it an illustrative example of the main idea of
lazification.  The lazy algorithm (Algorithm~\ref{alg:FW-WSep})
maintains an upper bound \(\Phi_{t}\) on the convergence rate, guiding
its eagerness for progress when searching for an improving vertex
\(v_{t}\). If the weak separation oracle provides an improving vertex
\(v_t\) we refer to this as a \emph{positive call} and
if the oracle claims there are no improving vertices
we call it a \emph{negative call}.

\begin{algorithm}
  \caption{Lazy Conditional Gradient}
  \label{alg:FW-WSep}
  \begin{algorithmic}[1]
    \REQUIRE
      smooth convex function \(f\) with curvature \(C\),
      start vertex \(x_{1} \in P\),
      weak linear separation oracle \(\LPsep{P}\),
      accuracy \(K \geq 1\),
      step sizes \(\gamma_{t}\),
      initial upper bound \(\Phi_{0}\)
    \ENSURE points \(x_{t}\) in \(P\)
    \FOR{\(t=1\) \TO \(T-1\)}
      \STATE
        \label{line:FW-big-Phi-t}
        \(\Phi_{t} \leftarrow
        \frac{\Phi_{t-1} + \frac{C \gamma_{t}^{2}}{2}}
        {1 + \frac{\gamma_{t}}{K}}\)
      \STATE
        \label{line:FW-improving-v}
        \(v_t \leftarrow \LPsep{P}(\nabla f(x_{t}),
        x_{t}, \Phi_{t}, K)\)
      \IF{\(v_{t} = \FALSE\)}
        \STATE
        \label{line:FW-same}
        \(x_{t+1} \leftarrow x_{t}\)
      \ELSE
        \STATE
          \label{line:FW-update}
          \(x_{t+1} \leftarrow
          (1 - \gamma_{t}) x_{t} + \gamma_{t} v_{t}\)
      \ENDIF
    \ENDFOR
  \end{algorithmic}
\end{algorithm}

The step size \(\gamma_{t}\) is chosen to (approximately)
minimize \(\Phi_{t}\) in Line~\ref{line:FW-big-Phi-t};
roughly \(\Phi_{t-1} / K C\). 

\begin{theorem}
  \label{thm:offline-FW-error}
  Assume \(f\) is convex and smooth with curvature \(C\).
  Then Algorithm~\ref{alg:FW-WSep}
  with
  \(\gamma_{t} = \frac{2 (K^{2} + 1)}{K (t + K^{2} + 2)}\)
  and \(f(x_{1}) - f(x^{*}) \leq \Phi_{0}\)
  has convergence rate
  \begin{equation}
    \label{eq:offline-FW-error}
    f(x_{t}) - f(x^*)
    \leq
    \frac{2 \max\{ C, \Phi_{0} \} (K^{2} + 1)}{t + K^{2} + 2}, 
  \end{equation}
  where \(x^*\) is a minimum point of \(f\) over \(P\).
\begin{proof}
We prove by induction that
\[
f(x_{t}) - f(x^{*}) \leq \Phi_{t-1}
.
\]
The claim is clear for \(t=1\) by the choice of \(\Phi_{0}\).
Assuming the claim is true for \(t\), we prove it for \(t+1\).
We distinguish two cases depending on the return value
of the weak separation oracle in Line~\ref{line:FW-improving-v}.

In case of a positive call, i.e.,
when the oracle returns an improving solution \(v_{t}\),
then \(\nabla f(x_{t}) (x_{t} - v_{t}) \geq \Phi_{t} / K\),
which is used in the second inequality below.
The first inequality follows by smoothness of \(f\), and the 
second inequality by the induction hypothesis and the fact that
\(v_{t}\) is an improving solution:
\begin{align*}
  f(x_{t+1}) - f(x^*)
  &\leq
  \underbrace{f(x_t) - f(x^*)}_{\leq \Phi_{t-1}}
  + \gamma_{t}
  \underbrace{\nabla f(x_{t}) (v_{t} - x_{t})}_{\leq - \Phi_{t} / K}
  +
  \frac{C \gamma_{t}^{2}}{2} \\
  & \leq 
  \Phi_{t-1}
  -
  \gamma_{t}
  \frac{\Phi_{t}}{K}
  +
  \frac{C \gamma_{t}^{2}}{2} \\
  & =
  \Phi_{t}
  ,
\end{align*}

In case of a negative call, i.e.,
when the oracle returns no improving solution,
then in particular
\(\nabla f(x_{t}) (x_{t} - x^{*}) \leq \Phi_{t}\),
hence by Line~\ref{line:FW-same}
\begin{equation}
  f(x_{t+1}) - f(x^{*}) = f(x_{t}) - f(x^{*}) \leq
  \nabla f(x_{t}) (x_{t} - x^{*})
  \leq \Phi_{t}
  .
\end{equation}

Finally, using the specific values of \(\gamma_{t}\)
we prove the upper bound
\begin{equation}
  \Phi_{t-1}
  \leq
  \frac{2 \max\{ C, \Phi_{0} \} (K^{2} + 1)}{t + K^{2} + 2}
\end{equation}
by induction on \(t\).
The claim is obvious for \(t = 1\).
The induction step is an easy computation
relying on the definition of \(\Phi_{t}\)
on Line~\ref{line:FW-big-Phi-t}:
\begin{equation}
  \Phi_{t}
  =
  \frac{\Phi_{t-1} + \frac{C \gamma_{t}^{2}}{2}}
  {1 + \frac{\gamma_{t}}{K}}
  \leq
  \frac{\frac{2 \max\{ C, \Phi_{0} \} (K^{2} + 1)}{t + K^{2} + 2}
    + \frac{\max\{ C, \Phi_{0} \} \gamma_{t}^{2}}{2}}
  {1 + \frac{\gamma_{t}}{K}}
  \leq
  \frac{2 \max\{ C, \Phi_{0} \} (K^{2} + 1)}{t + K^{2} + 3}
  .
\end{equation}
Here the last inequality follows from
the concrete value of \(\gamma_{t}\).
\end{proof}
\end{theorem}

Note that by design, the algorithm converges at the worst-case
rate that we postulate due to the negative calls when it does
not move. Clearly, this is highly undesirable,
therefore the algorithm should be understood as
the \emph{textbook variant} of lazy conditional gradient.
We will present an improved, parameter-free variant of
Algorithm~\ref{alg:FW-WSep} in Section~\ref{sec:frank-wolfe-with} that
converges at the best possible rate that the non-lazy variant would
achieve (up to a small constant factor).

\subsection{Lazy Pairwise Conditional Gradient}

In this section
we provide a lazy variant
(Algorithm~\ref{alg:pairwise-cond-gradient})
of the Pairwise Conditional Gradient algorithm
from \cite{LDLCC2016},
using separation instead of linear optimization.
We make identical assumptions: the feasible region
is a \(0/1\) polytope,
i.e., all vertices of \(P\) have only \(0/1\) entries,
and moreover it is given in the
form \(P = \{x \in \mathbb{R}^n \mid 0 \le x \le \allOne, Ax = b\}\),
where \(\allOne\) denotes the all-one vector.

\begin{algorithm}
  \caption{Lazy Pairwise Conditional Gradient (LPCG)}
  \label{alg:pairwise-cond-gradient}
  \begin{algorithmic}[1]
    \REQUIRE
      polytope \(P\),
      smooth and \(S\)-strongly convex function \(f\)
      with curvature \(C\),
      accuracy \(K \geq 1\),
      non-increasing step-sizes \(\eta_t\),
      eagerness \(\Delta_{t}\)
    \ENSURE points \(x_t\)
    \STATE \(x_1 \in P\) arbitrary and \(\Phi_0 \ge f(x_1) - f(x^*)\)
    \FOR{\(t = 1, \dots, T\)}
      \STATE
      \(
       \widetilde{\nabla}f(x_t)_i \coloneqq \begin{cases}
                                \nabla f(x_t)_i & \text{if } x_{t,i} > 0\\
                                -\infty & \text{if } x_{t,i} = 0
                              \end{cases}
      \)
      \STATE \(\Phi_t \leftarrow
        \frac{2\Phi_{t-1} + \eta_t^2 C}{2 +
           \frac{\eta_t}{K \Delta_t}}\)
        \label{line:LPCG-Phi-t}
      \STATE
      \label{line:objective-function}
      \(c_t \leftarrow \left(\nabla f(x_t), -\widetilde{\nabla} f(x_t)\right)\)
      \STATE
      \label{line:pairwise-oracle-call}
      \((v_t^+, v_t^-) \leftarrow
	\LPsep{P \times P} \left(
          c_{t},\, (x_t, x_t),\, \frac{\Phi_t}{\Delta_t},\,
          K \right)\)
      \IF{\((v_t^+, v_t^-) = \FALSE\)}
        \STATE
        \label{line:pairwise-negative-case}
        \(x_{t+1} \leftarrow x_{t}\)
      \ELSE
	\STATE
          \label{line:approx-feasible}
          \(\tilde{\eta}_t \leftarrow \max \{2^{-\delta} \mid
          \delta \in \mathbb{Z}_{\ge 0},\ 2^{-\delta} \le \eta_t\}\)
	\STATE
	  \label{line:pairwise-pos-case}
	  \(x_{t+1} \leftarrow x_t + \tilde{\eta}_{t} (v_t^+ - v_t^-)\)
      \ENDIF
    \ENDFOR
  \end{algorithmic}
\end{algorithm}

Observe that Algorithm~\ref{alg:pairwise-cond-gradient} calls the
linear separation oracle \LPsep{} on the cartesian product of \(P\)
with itself. Choosing the objective function as in
Line~\ref{line:objective-function} allows us to simultaneously
find an improving direction and an away-step direction.

Let \(\operatorname{card}{x}\) denote the number of non-zero entries
of the vector \(x\).
\begin{theorem}
  \label{thm:pairwise-cond-gradient}
  Let \(x^{*}\) be a minimum point of \(f\) in \(P\),
  and \(\Phi_{0}\) an upper bound of \(f(x_{1}) - f(x^{*})\).
  Furthermore,
  let \(\operatorname{card}(x^*) \leq \alpha\),
  \(M_1 \coloneqq \sqrt{\frac{S}{8 \alpha}}\),
 \(\kappa \coloneqq
 \min \left\{ \frac{M_{1}}{K C}, 1 / \sqrt{\Phi_{0}} \right\}\),
 \(\eta_t \coloneqq \kappa \sqrt{\Phi_{t-1}}\) and
 \(\Delta_t \coloneqq \sqrt{\frac{2 \alpha \Phi_{t-1}}{S}}\), then
 Algorithm~\ref{alg:pairwise-cond-gradient}
 has convergence rate
 \[
  f(x_{t+1}) - f(x^*) \le \Phi_t \le \Phi_0 \left( \frac{1 + B}{1 + 2B} \right)^t,
 \]
 where \(B \coloneqq \kappa \cdot \frac{M_{1}}{2 K}\).
\end{theorem}

We recall a technical lemma for the proof.
\begin{lemma}[{\cite[Lemma 2]{LDLCC2016}}]
  \label{lem:decomposition}
  Let \(x,y \in P\).
  Then \(x\) is a liner combination
  \(x = \sum_{i=1}^k \lambda_{i} v_{i}\)
  of some vertices \(v_{i}\) of \(P\)
  (in particular, \(\sum_{i=1}^{k} \lambda_{i} = 1\))
  with
  \(x - y =  \sum_{i=1}^{k} \gamma_{i} (v_{i} - z)\)
  for some \(0 \leq \gamma_{i} \leq \lambda_{i}\) and
  \(z\in P\) such that
  \(\sum_{i=1}^{k} \gamma_{i} \leq
  \sqrt{\operatorname{card}(y)} \norm{x-y}\).
\end{lemma}
\begin{proof}[Proof of Theorem~\ref{thm:pairwise-cond-gradient}]
The feasibility of the iterates \(x_t\) is ensured by
Line~\ref{line:approx-feasible} and the monotonicity of
the sequence \(\{\eta_t\}_{t\ge 1}\)
with the same argument as in
\cite[Lemma~1 and Observation~2]{LDLCC2016}.

  We first show by induction that
  \[
   f(x_{t+1}) - f(x^*) \le \Phi_t.
  \]
  For \(t=0\) we have \(\Phi_0 \ge f(x_1) - f(x^*)\).
  Now assume the statement for some \(t\ge 0\).
  In case of a negative call
  (Line~\ref{line:pairwise-negative-case}),
  we use the guarantee of
  Oracle~\ref{alg:weak-separate-oracle} to get
  \[
   c_t [(x_t, x_t) - (z_1, z_2)] \leq \frac{\Phi_t}{\Delta_t}
  \]
  for all \(z_1, z_2 \in P\), which is equivalent to
  (as \(c_{t}(x_{t}, x_{t}) = 0\))
  \[
   \widetilde{\nabla} f(x_t) z_2 - \nabla f(x_t) z_1 \le \frac{\Phi_t}{\Delta_t}
  \]
  and therefore
  \begin{equation}
  \label{eq:gradient-upper-bound}
   \nabla f(x_t)(\tilde{z}_2 - z_1) \le \frac{\Phi_t}{\Delta_t},
  \end{equation}
  for all \(\tilde{z}_2, z_1 \in P\) with
  \(\supp{\tilde{z}_2} \subseteq \supp{x_t}\),
  where \(\supp{x}\) denotes the set of non-zero coordinates of \(x\).
  We use Lemma~\ref{lem:decomposition}
  for the decompositions \(x_t = \sum_{i=1}^k \lambda_i v_i\) and
  \(x_{t} - x^* = \sum_{i=1}^{k} \gamma_{i} (v_{i} - z)\) with
  \(0 \leq \gamma_{i} \leq \lambda_{i}\), \(z \in P\) and
  \[\sum_{i=1}^{k} \gamma_{i}
    \leq \sqrt{\operatorname{card}(x^{*})}\norm{x_{t} - x^{*}}
    \leq \sqrt{\frac{2 \operatorname{card}(x^*) \Phi_{t-1}}{S}}
    \leq \Delta_{t},\]
  using the induction hypothesis and strong convexity
  in the second inequality.
  Then
  \begin{equation*}
    f(x_{t+1}) - f(x^*)
    = f(x_{t}) - f(x^*) \le \nabla f(x_t) (x_t - x^*)
    = \sum_{i=1}^{k} \gamma_{i} \cdot
    \underbrace{\nabla f(x_t) (v_i - z)}_{\leq \Phi_{t} / \Delta_{t}}
    \leq \Phi_{t},
  \end{equation*}
  where we used Equation~\eqref{eq:gradient-upper-bound}
  for the last inequality.

  In case of a positive call
  (Lines~\ref{line:approx-feasible} and \ref{line:pairwise-pos-case})
  we get,
  using first smoothness of \(f\), then
  \(\eta_{t} / 2 < \tilde{\eta_{t}} \leq \eta_{t}\)
  and
  \(\nabla f(x_t)(v_t^+ - v_t^-) \leq - \Phi_{t} / (K \Delta_{t})\),
  and finally the definition of \(\Phi_{t}\):
  \begin{align*}
   f(x_{t+1}) - f(x^*)
   &= f(x_{t}) - f(x^*)
   +
   f(x_{t} + \tilde{\eta_{t}} (v_{t}^{+} - v_{t}^{-})) - f(x_{t})
   \\
   &
   \le \Phi_{t-1} + \tilde{\eta}_t \nabla f(x_t)(v_t^+ - v_t^-) +
      \frac{\tilde{\eta}_t^2 C}{2} \\
   & \le \Phi_{t-1} - \frac{\eta_t}{2} \cdot
   \frac{\Phi_t}{K \Delta_{t}}
   + \frac{\eta_t^2 C}{2}
   = \Phi_t
   .
  \end{align*}
  Plugging in the values of \(\eta_t\) and \(\Delta_t\)
  to the definition of \(\Phi_{t}\) gives the desired bound.
  \begin{equation*}
   \Phi_{t}
   = \frac{2\Phi_{t-1} + \eta_t^2 C}{2 +
     \frac{\eta_t}{K \Delta_t}}
   =
   \Phi_{t-1}
   \frac{1 + \kappa^{2} C / 2}{1 + \kappa M_{1} / K}
   \leq
   \Phi_{t-1}
   \frac{1 + B}{1 + 2 B}
   \le \Phi_0 \left( \frac{1 + B}{1 + 2B} \right)^t.
   \qedhere
  \end{equation*}
 \end{proof}

\subsection{Lazy Local Conditional Gradient}
\label{sec:lazy-local}

In this section we provide a lazy version
(Algorithm~\ref{alg:llao-FW})
of the conditional gradient
algorithm from \cite{garber2013linearly}.
Let \(P\subseteq \mathbb{R}^{n}\)
be any polytope, \(D\)
denote an upper bound on the \(\ell_{2}\)-diameter
of \(P\),
and \(\mu \geq 1\)
be an affine invariant parameter of \(P\)
satisfying Lemma \ref{lem:lloo-decompose} below,
see \cite[Section~2]{garber2013linearly} for a possible definition.
As the algorithm is not affine invariant by nature,
we need a non-invariant version of smoothness:
Recall that a convex function \(f\) is \emph{\(\beta\)-smooth}
if \[f(y) - f(x) \leq
\nabla f(x) (y - x) + \beta \norm{y - x}^{2} / 2.\]

\begin{algorithm}
  \caption{Lazy Local Conditional Gradient (LLCG)}
  \label{alg:llao-FW}
  \begin{algorithmic}[1]
    \REQUIRE
      feasible polytope \(P\),
      \(\beta\)-smooth and \(S\)-strongly convex function \(f\),
      parameters \(K\), \(S\), \(\beta\), \(\mu\);
      diameter \(D\)
    \ENSURE points \(x_t\)
    \STATE \(x_1 \in P\) arbitrary and \(\Phi_0 \ge f(x_1) - f(x^*)\)
    \STATE \(\alpha \leftarrow \frac{S}{2 K \beta n \mu^2}\)
    \FOR{\(t = 1, \dots, T\)}
      \STATE \(r_t \leftarrow \sqrt{\frac{2\Phi_{t-1}}{S}} \)
      \STATE \(\Phi_t \leftarrow
	\frac{\Phi_{t-1} +
          \frac{\beta}{2} \alpha^2 \min\{n\mu^2r_t^2, D^2\}}{
          1 + \alpha / K}\)
        \label{line:LLCG-Phi-t}
      \STATE \(p_t \leftarrow
        \LLPsep{P}\left( \nabla f(x_t), x_t, r_t, \Phi_t, K \right) \)
      \IF{\(p_t = \FALSE\)}
        \STATE
        \(x_{t+1} \leftarrow x_{t}\)
      \ELSE
        \STATE \label{line:locg-positive-case}
          \(x_{t+1} \leftarrow x_t + \alpha(p_t - x_t)\)
      \ENDIF
    \ENDFOR
  \end{algorithmic}
\end{algorithm}

As an intermediary step, we first
implement a \emph{local weak separation oracle}
in Algorithm~\ref{alg:weak-local-separate},
a \emph{local} version of Oracle~\ref{alg:weak-separate-oracle},
which finds improving points only in a small neighbourhood of the
original point,
analogously to the local linear optimization oracle in
\cite{garber2013linearly}.
To this end,
we recall a technical lemma from \cite{garber2013linearly}.
\begin{algorithm}[t]
  \caption{Weak Local Separation
    \(\LLPsep{P}(c, x, r, \Phi, K)\)}
  \label{alg:weak-local-separate}
  \begin{algorithmic}[1]
    \REQUIRE
      polytope \(P\) together with invariant \(\mu\),
    linear objective \(c \in \mathbb{R}^{n}\),
    point \(x \in P\),
    radius \(r > 0\),
    objective value \(\Phi > 0\),
    accuracy \(K \geq 1\)
    \ENSURE Either
    \begin{enumerate*}
    \item \(y \in P\) with \(\norm{x-y} \le \sqrt{n}\mu r\) and
      \(c (x - y) > \Phi / K\)
    \item \FALSE: \(c (x - z) \leq \Phi\) for all \(z \in P \cap \ball{r}{x}\)
    \end{enumerate*}
    \STATE \(\Delta \leftarrow \min \left\{
        \frac{\sqrt{n} \mu}{D} r, 1 \right\}\)
    \STATE Decompose \(x\):
      \(x = \sum_{j=1}^{M} \lambda_{j} v_{j}\),
      \ \(\lambda_{j} > 0\), \(\sum_{j} \lambda_{j} = 1\).
    \STATE Sort vertices:
      \(i_1, \dots, i_M\) \quad \(cv_{i_1} \ge \dots \ge cv_{i_M}\).
    \STATE \(k \leftarrow \min \{k :
      \sum_{j=1}^k \lambda_{i_j} \ge \Delta\}\)
    \STATE \(p_- \leftarrow \sum_{j=1}^{k-1} \lambda_{i_j} v_{i_j} + \left(\Delta - \sum_{j=1}^{k-1} \lambda_{i_j} \right)v_{i_k}\)
    \STATE \(v^* \leftarrow \LPsep{P}\left(c, \frac{p_{-}}{\Delta},
          \frac{\Phi}{\Delta} \right)\)
    \IF{\(v^* = \FALSE\)}
      \RETURN \label{line:LLP-false} \FALSE
    \ELSE
      \RETURN
      \label{line:LLP-true}
      \(y \leftarrow x - p_- + \Delta v^*\)
    \ENDIF
  \end{algorithmic}
\end{algorithm}

\begin{lemma}
  \label{lem:lloo-decompose}
  \cite[Lemma~7]{garber2013linearly}
  Let \(P \subseteq \R^{n}\) be a polytope and \(v_{1}, \dots, v_{N}\)
  be its vertices.
  Let \(x, y \in P\)
  and \(x = \sum_{i=1}^{N} \lambda_{i} v_{i}\)
  a convex combination of the vertices of \(P\).
  Then there are numbers \(0 \leq \gamma_{i} \leq \lambda_{i}\)
  and \(z \in P\) satisfying
  \begin{align}
    x - y &= \sum_{i \in [N]} \gamma_{i} (z - v_{i})
    \\
    \sum_{i \in [N]} \gamma_{i}
    &
    \leq \frac{\sqrt{n} \mu}{D}
    \norm{x - y}
    .
  \end{align}
\end{lemma}

Now we prove the correctness of the weak local separation algorithm.
\begin{lemma}
 \label{lem:lloo-guarantee}
 Algorithm~\ref{alg:weak-local-separate} is correct.
 In particular \(\LLPsep{P}(c, x, r, \Phi, K)\)
 \begin{enumerate}
  \item returns either an \(y \in P\)
    with \(\norm{x-y} \le \sqrt{n}\mu r\) and
    \(c (x - y) \geq \Phi / K\),
  \item\label{item:weak-local-negative}
    or returns \algorithmicfalse, and then
    \(c (x - z) \leq \Phi\) for all \(z \in P \cap \ball{r}{x}\).
 \end{enumerate}

 \begin{proof}
  We first consider the case
  when the algorithm exits in Line~\ref{line:LLP-true}.
  Observe that \(y\in P\) since
  \(y\) is a convex combination of vertices of by construction:
  \(y = \sum_{j=1}^{M} (\lambda_{i_j} - \gamma_j) v_{i_j} + \Delta v^*\) with
  \(\Delta = \sum_{j=1}^{M} \gamma_j \le \frac{\sqrt{n} \mu}{D} r\),
  where \(\gamma_{j} = \lambda_{i_{j}}\) for \(j < k\),
  and \(\gamma_{k} = \Delta - \sum_{j=1}^{k-1} \lambda_{i_{j}}\),
  and \(\gamma_{j} = 0\) for \(j > k\).
  Therefore
  \begin{equation*}
     \norm{x-y} = \Norm{\sum_{j=1}^{M} \gamma_j (v_{i_j} - v^{*})}
    \le \sum_{j=1}^{M} \gamma_j \norm{v_{i_j} - v^*}
    \le \sqrt{n} \mu r
     .
  \end{equation*}
  Finally using the guarantee of \LPsep{P} we get
  \begin{equation*}
    c(x-y)
    = \Delta c \left( \frac{p_{-}}{\Delta} - v^{*} \right)
    \geq \frac{\Phi}{K}
    .
  \end{equation*}

  If the algorithm exits in Line~\ref{line:LLP-false},
  we use Lemma~\ref{lem:lloo-decompose}
  to decompose any \(y \in P\cap \ball{r}{x}\):
  \[
   x - y = \sum_{i=1}^{M} \gamma_{i} (v_{i} - z),
  \]
  with \(z \in P\) and
  \(\sum_{i=1}^{M} \gamma_{i} \leq \frac{\sqrt{n} \mu}{D} \norm{x-y}
  \le \Delta\).
  Since \(\sum_{i=1}^{M} \lambda_{i}  = 1 \geq \Delta\),
  there are numbers \(\gamma_{i} \leq \eta_{i}^{-} \leq \lambda_{i}\)
  with \(\sum_{i=1}^{M} \eta_{i}^{-} = \Delta\).
  Let
  \begin{align}
    \tilde{p}_{-} &\coloneqq \sum_{i=1}^{M} \eta_i^{-} v_{i},
    \\
    \tilde{p}_{+} &\coloneqq y - x + \tilde{p}_{-}
    = \sum_{i=1}^{M} (\eta_{i}^{-} - \gamma_{i}) v_{i} +
    \sum_{i=1}^{M} \gamma_{i} z,
  \end{align}
  so that \(\tilde{p}_{+} / \Delta \in P\).
  To bound the function value we first observe that
  the choice of \(p_-\) in the algorithm assures that
  \(c u \leq c p_{-}\) for all \(u = \sum_{i=1}^{M} \eta_{i} v_{i}\)
  with \(\sum_{i=1}^{M} \eta_{i} = \Delta\)
  and all \(0 \leq \eta_{i} \leq \lambda_{i}\).
  In particular, \(c\tilde{p}_- \le cp_-\).
  The function value of the positive part \(\tilde{p}_{+}\)
  can be bounded with the guarantee
  of \LPsep{P}:
  \[
   c \left(
     \frac{p_{-}}{\Delta} - \frac{\tilde{p}_{+}}{\Delta}
   \right)
   \le \frac{\Phi}{\Delta}
   ,
  \]
  i.e., \(c (p_{-} - \tilde{p}_{+}) \leq \Phi\).
  Finally combining these bounds gives
  \[
   c(x-y) = c\left(\tilde{p}_- - \tilde{p}_+\right)
   \le c(p_{-} - \tilde{p}_{+}) \leq \Phi
  \]
  as desired.
 \end{proof}
\end{lemma}

We are ready to examine the Conditional Gradient Algorithm
based on \LLPsep{P}:

\begin{theorem}
\label{thm:lloo-conv-rate}
 Algorithm~\ref{alg:llao-FW} converges with the following rate:
 \[
 f(x_{t+1}) - f(x^*) \le \Phi_t \le
 \Phi_0 \left(\frac{1 + \alpha / (2 K)}{1 + \alpha / K}\right)^{t}.
 \]
 \begin{proof}
  The proof is similar to the proof of
  Theorem~\ref{thm:pairwise-cond-gradient}.
  We prove this rate by induction. For \(t=0\) the choice of
  \(\Phi_0\) guarantees that \(f(x_1) - f(x^*) \le \Phi_0\).
  Now assume the theorem holds for \(t\ge 0\).
  With strong convexity
  and the induction hypothesis we get
  \[
   \norm{x_t - x^*}^2 \le \frac{2}{S} (f(x_t) - f(x^*))
   \le \frac{2}{S} \Phi_{t-1} = r_t^2,
  \]
  i.e., \(x^* \in P \cap \ball{r_t}{x_t}\).
  In case of a negative call, i.e.,
  when \(p_t = \mathbf{false}\),
  then case \ref{item:weak-local-negative}
  of Lemma~\ref{lem:lloo-guarantee} applies:
  \begin{equation*}
   f(x_{t+1}) - f(x^*)
   = f(x_t) - f(x^*) \le \nabla f(x_t) (x_t - x^*)
   \le \Phi_t.
  \end{equation*}

  In case of a positive call, i.e.,
  when Line~\ref{line:locg-positive-case} is
  executed, we get the same inequality via:
  \begin{align*}
   f(x_{t+1}) - f(x^*)
   &\le \Phi_{t-1} + \alpha \nabla f(x_t)(p_t - x_t) + 
      \frac{\beta}{2} \alpha^2 \norm{x_{t} - p_t}^2 \\
   & \le \Phi_{t-1} - \alpha \frac{\Phi_t}{K}
   + \frac{\beta}{2} \alpha^2 \min\{n\mu^2 r_t^2, D^2\} \\
   & = \Phi_t.
  \end{align*}
  Therefore using the definition of \(\alpha\) and \(r_t\) we get the desired bound:
  \begin{equation*}
   \Phi_{t} \le \frac{\Phi_{t-1}
     + \frac{\beta}{2}\alpha^2 r_t^2 n \mu^2}{1 + \alpha / K}
   = \Phi_{t-1} \left( \frac{1 + \alpha / (2 K)}{1 + \alpha / K}
   \right)
   \leq
   \Phi_0 \left( \frac{1 + \alpha / (2 K)}{1 + \alpha / K}\right)^{t}
   .
   \qedhere
  \end{equation*}
 \end{proof}

\end{theorem}

\section{Parameter-free Conditional Gradient via Weak Separation}
\label{sec:frank-wolfe-with}

In this section we provide a parameter-free variant of the Lazy
Frank-Wolfe Algorithm, which is inspired by \cite{pokutta2017} and
which exhibits a very favorable behavior in computations; the same
technique applies to all other variants from
Section~\ref{sec:offline-frank-wolfe} as well.
The idea is that instead of using predetermined values for
progress parameters,
like \(\Phi_{t}\) and \(\gamma_{t}\) in Algorithm~\ref{alg:FW-WSep},
to optimize worst-case progress,
the parameters are adjusted adaptively, using data encountered by the
algorithm, and
avoiding hard-to-estimate parameters, like the curvature \(C\).
In practice, this leads to faster convergence, as usual for adaptive
methods, while the theoretical convergence rate is worse only by a
small constant factor.
See Figures~\ref{fig:video-colocalization} and
\ref{fig:matrix-completion} for a comparison
and Section~\ref{sec:offlineResults}
for more experimental results.

The occasional halving of the \(\Phi_{t}\) is reminiscent of an
adaptive restart strategy, considering successive iterates with the
same \(\Phi_{t}\) as an epoch.
It ensures that \(\Phi_{t}\) is always at least half of the primal
gap, while quickly reducing it if it is too large for the algorithm to
make progress, and as such it represents a reasonable amount of
expected progress throughout the whole run of the algorithm, not just
at the initial iterates.

\begin{algorithm}
  \caption{Parameter-free Lazy Conditional Gradient (LCG)}
\label{alg:CG-imp-imp}
  \begin{algorithmic}[1]
    \REQUIRE
      smooth convex function \(f\),
       start vertex \(x_{1} \in P\),
       weak linear separation oracle \(\LPsep{P}\),
      accuracy \(K \geq 1\)
    \ENSURE points \(x_{t}\) in \(P\)
    \STATE
    \label{line:phi_0}
      \(\Phi_{0} \leftarrow \max_{x \in P} \nabla f(x_{1}) (x_1 - x) /
      2\)
      \COMMENT{Initial bound}
    \FOR{\(t=1\) \TO \(T-1\)}
    \STATE
      \label{line:oracle_pfree}
      \(v_t \leftarrow
      \LPsep{P}(\nabla f(x_{t}), x_{t}, \Phi_{t-1}, K)\)
    \IF{\(v_{t} = \FALSE\)}
      \STATE \(x_{t+1} \leftarrow x_{t}\)
      \STATE
        \label{line:phi-update}
        \(\Phi_{t} \leftarrow \frac{\Phi_{t-1}}{2}\)
        \COMMENT{Update \(\Phi\)}
      \ELSE
        \STATE
          \(\gamma_t \leftarrow \argmin{0 \leq \gamma \leq 1}
          f((1 - \gamma) x_{t} + \gamma v_{t})\)
          \COMMENT{Line search}
        \STATE
          \(x_{t+1} \leftarrow
          (1 - \gamma_{t}) x_{t} + \gamma_{t} v_{t}\)
          \COMMENT{Update iterate}
        \STATE \(\Phi_{t} \leftarrow \Phi_{t-1}\)
      \ENDIF
    \ENDFOR
  \end{algorithmic}
\end{algorithm}

\begin{remark}[Additional LP call for initial bound]
  Note that Algorithm~\ref{alg:CG-imp-imp} finds a tight initial bound
  \(\Phi_0\) with a single extra LP call. If this is undesired, this
  can be also done approximately as long as \(\Phi_0\) is a valid
  upper bound, for example by means of binary search via the weak
  separation oracle.
\end{remark}

\begin{theorem}
  \label{thm:conv-parameter-free} Let \(f\) be a smooth convex
  function with curvature \(C\).  Algorithm~\ref{alg:CG-imp-imp}
  converges at a rate proportional to \(1/t\). In particular to
  achieve a bound \(f(x_t) - f(x^*) \leq \varepsilon\),
  given an initial upper bound
  \(f(x_{1}) - f(x^{*}) \leq 2 \Phi_{0}\),
  the number of
  required steps is upper bounded by
 \[
   t \le
   \left\lceil
     \log \frac{\Phi_{0}}{\varepsilon}
   \right\rceil
   + 1
   + 4K
   \left\lceil
     \log \frac{\Phi_{0}}{KC}
   \right\rceil
   + \frac{16 K^{2} C}{\varepsilon}
   .
 \]
\begin{proof}
  The main idea of the proof is to maintain an approximate upper bound
  on the optimality gap. We then show that negative calls halve the
  upper bound \(2 \Phi_{t}\) and positive oracle calls make significant
  objective function improvement.

We analyze iteration \(t\) of the algorithm.
If Oracle~\ref{alg:weak-separate-oracle}
in Line~\ref{line:oracle_pfree}
returns a negative answer
(i.e., \algorithmicfalse, case \ref{item:negative}),
then this guarantees \(\nabla f(x_{t}) (x_t - x) \leq \Phi_{t-1}\)
for all \(x \in P\), in particular, using convexity,
\(f(x_{t+1}) - f(x^*) = f(x_t) - f(x^*) \leq
\nabla f(x_{t}) (x_t - x^*) \leq \Phi_{t-1} = 2\Phi_t\).

If Oracle~\ref{alg:weak-separate-oracle}
returns a positive answer (case \ref{item:positive}), then we have
\(f(x_t) - f(x_{t+1}) \geq
\gamma_t \Phi_{t-1} / K - (C / 2) \gamma_t^2\)
by smoothness of \(f\).
By minimality of \(\gamma_{t}\),
therefore
\(f(x_t) - f(x_{t+1}) \geq \min_{0 \leq \gamma \leq 1}
(\gamma \Phi_{t-1} / K - (C/2) \gamma^{2})\),
which is \(\Phi_{t-1}^{2} / (2C K^{2})\) if \(\Phi_{t-1} < KC\),
and \(\Phi_{t-1} / K - C/2 \geq \frac{C}{2}\)
if \(\Phi_{t-1} \geq KC\).

Now we bound the number \(t'\) of consecutive positive oracle calls
immediately following an iteration \(t\) with a negative oracle call.
Note that the same argument bounds the number of initial consecutive
positive oracle calls with the choice \(t=0\), as we only use
\(f(x_{t+1}) - f(x^{*}) \leq 2 \Phi_{t}\) below.

Note that \(\Phi_{t} = \Phi_{t+1} = \dots = \Phi_{t+t'}\).
Therefore
  \begin{equation}
    2 \Phi_{t} \ge f(x_{t+1}) - f(x^*) \ge
    \sum_{\tau = t + 1}^{t + t'} (f(x_{\tau}) - f(x_{\tau +1}))
   \ge
    \begin{cases*}
      t' \frac{\Phi_{t}^2}{2C K^2} & if \(\Phi_{t} < KC\) \\
      t' \left(\frac{\Phi_{t}}{K}-\frac{C}{2} \right)
      & if \(\Phi_{t} \geq KC\)
    \end{cases*},
  \end{equation}
  which gives in the case \(\Phi_{t} < KC\) that \(t' \le 4CK^2 / \Phi_{t}\), and
  in the case \(\Phi_{t} \ge KC\) that
  \[
    t' \le \frac{2\Phi_{t}}{\frac{\Phi_{t}}{K}-\frac{C}{2}} = \frac{4K\Phi_{t}}{2\Phi_{t} -
  KC} \leq \frac{4K\Phi_{t}}{2\Phi_{t} - \Phi_{t}} = 4K.
  \]
Thus iteration \(t\) is followed by at most \(4K\) consecutive
positive oracle calls as long as \(\Phi_{t} \geq KC\),
and \(4CK^2 / \Phi_{t} < 2^{\ell + 1} \cdot 4K\) ones for
\(2^{- \ell - 1} KC < \Phi_{t} \leq 2^{-\ell} KC\) with \(\ell \geq 0\).

Adding up the number of oracle calls gives the desired rate:
in addition to the positive oracle calls we also have
at most \(\lceil \log (\Phi_0 / \varepsilon) \rceil + 1\)
negative oracle calls, where
\(\log(\cdot)\) is the binary logarithm and \(\varepsilon\) is the
(additive) accuracy.
Thus after a total of 
\begin{equation*}
  \left\lceil
    \log \frac{\Phi_{0}}{ \varepsilon}
  \right\rceil
  + 1
  + 4K
  \left\lceil
    \log \frac{\Phi_{0}}{KC}
  \right\rceil
  + \sum_{\ell = 0}^{\lceil \log (KC / \varepsilon) \rceil}
  2^{\ell + 1} \cdot 4K
  \leq
  \left\lceil
    \log \frac{\Phi_{0}}{\varepsilon}
  \right\rceil
  + 1
  + 4K
  \left\lceil
    \log \frac{\Phi_{0}}{KC}
  \right\rceil
  + \frac{16 K^{2} C}{\varepsilon}
\end{equation*}
iterations (or equivalently oracle calls) we have \(f(x_{t}) - f(x^*) \leq \varepsilon\).
\end{proof}
\end{theorem}

As seen from the proof, the algorithm receives few negative oracle
calls by design; these are usually more expensive than positive ones
as the oracle has to compute a certificate by, e.g., executing a full
linear optimization oracle call.

\begin{corollary}
\label{cor:negCalls}
Algorithm~\ref{alg:CG-imp-imp} receives at most \(\lceil \log
  \Phi_0 / \varepsilon \rceil + 1\) negative oracle answers.
\end{corollary}

\begin{remark}[Improved use of Linear Optimization oracle]
  A possible improvement to Line~\ref{line:phi-update} is
  \(\Phi_{t} \leftarrow \max_{x \in P} \nabla f(x_{t}) (x_{t} - x) /
  2\), assuming that at a negative call the oracle also provides the
  dual gap \(\max_{x \in P} \nabla f(x_{t}) (x_{t} - x)\) as well as
  the minimizer \(\bar x \in P\) of the oracle call.  This is the case e.g., when the
  weak separation oracle is implemented as in
  Algorithm~\ref{alg:LPSepLPOracle}.  Clearly, the minimizer
  \(\bar x\) can be also used to perform a progress step; albeit without guarantee w.r.t.~to \(\Phi_t\).
\end{remark}

\begin{remark}[Line Search]
\label{sec:param-free-cond-linesearch}
If line search is too expensive we can choose
\(\gamma_t = \min \{1, \Phi_t / KC\}\) in
Algorithm~\ref{alg:CG-imp-imp}. In this case an estimate of the
curvature \(C\) is required.
\end{remark}

\section{Lazy Online Conditional Gradient}
\label{sec:online-fw-error}

In this section we lazify the online conditional gradient
algorithm of \cite{hazan2012projection}
over arbitrary polytopes \(P = \{x \in \mathbb{R}^n \mid Ax\le b\}\),
resulting in Algorithm~\ref{alg:online-FW-WSep}.
We slightly improve constant factors by replacing
\cite[Lemma~3.1]{hazan2012projection}
with a better estimation via
solving a quadratic inequality arising from strong convexity.
In this section the norm \(\norm{\cdot}\) can be arbitrary.

\begin{algorithm}
  \caption{Lazy Online Conditional Gradient (LOCG)}
\label{alg:online-FW-WSep}
  \begin{algorithmic}[1]
    \REQUIRE
    functions \(f_{t}\),
    start vertex \(x_{1} \in P\),
    weak linear separation oracle \(\LPsep{P}\),
    parameters \(K\), \(C\), \(b\), \(S\), \(s\);
    diameter \(D\)
    \ENSURE points \(x_{t}\)
    \FOR{\(t=1\) \TO \(T-1\)}
      \STATE \(\nabla_{t} \leftarrow \nabla f_{t}(x_{t})\)
      \IF{\(t = 1\)}
        \STATE \(h_{1} \leftarrow \min \{
        \dualnorm{\nabla_{1}} D,
        2 \dualnorm{\nabla_{1}}^{2} / S\}\)
      \ELSE
        \STATE
        \label{line:upper-bound}
        \(h_{t} \leftarrow
        \Phi_{t-1}
        +
        \min\left\{
          \dualnorm{\nabla_{t}} D,
          \frac{\dualnorm{\nabla_{t}}^{2}}{S t^{1-s}}
          +
          2 \sqrt{\frac{\dualnorm{\nabla_{t}}^{2}}{2 S t^{1-s}}
            \left(
              \frac{\dualnorm{\nabla_{t}}^{2}}{2 S t^{1-s}}
              +
              \Phi_{t-1}
            \right)}
        \right\}\)
      \ENDIF
      \STATE
      \label{line:big-Phi-t}
      \(\Phi_{t} \leftarrow
      \frac{h_{t} + \frac{C t^{1 - b} \gamma_{t}^{2}}{2 (1 - b)}}
      {1 + \frac{\gamma_{t}}{K}}\)
      \STATE
      \label{line:OFW-augmented-v}
      \(v_t \leftarrow \LPsep{P}(\sum_{i=1}^{t} \nabla f_{i}(x_{t}),
      x_{t}, \Phi_{t}, K)\)
      \IF{\(v_{t} = \FALSE\)}
        \STATE
        \label{line:OFW-same}
        \(x_{t+1} \leftarrow x_{t}\)
      \ELSE
        \STATE
        \label{line:OFW-update}
        \(x_{t+1} \leftarrow
        (1 - \gamma_{t}) x_{t} + \gamma_{t} v_{t}\)
        \STATE
        \label{line:Phi-update}
        \(\Phi_{t} \leftarrow
        h_{t} - \sum_{i=1}^{t} f_{i}(x_{t})
        + \sum_{i=1}^{t} f_{i}(x_{t+1})\)
      \ENDIF
    \ENDFOR
  \end{algorithmic}
\end{algorithm}

\begin{theorem}
  \label{thm:OFW-error}
  Let \(0 \leq b, s < 1\).
  Let \(K \geq 1\) be an accuracy
  parameter.
  Assume \(f_{t}\) is \(L\)-Lipschitz,
  and smooth with curvature at most \(C t^{-b}\).
  Let \(D \coloneqq \max_{y_{1}, y_{2} \in P} \norm{y_{1} - y_{2}}\)
  denote the diameter of \(P\) in norm \(\norm{}\).
  Then the following hold for
  the points \(x_{t}\)
  computed by Algorithm~\ref{alg:online-FW-WSep}
  where
  \(x_{T}^{*}\) is the minimizer of \(\sum_{t=1}^{T} f_{t}\):
  \begin{enumerate}
  \item\label{item:OFW-stochastic}
    With the choice
    \begin{equation*}
      \gamma_{t}
      =
      t^{- (1 - b) / 2}
      ,
    \end{equation*}
    the \(x_{t}\) satisfy
    \begin{equation}
      \label{eq:OFW-error_stochastic}
      \frac{1}{T} \sum_{t=1}^{T} (f_{t}(x_{T}) - f_{t}(x_{T}^{*}))
      \leq
      A T^{- (1 - b) / 2}
      ,
    \end{equation}
    where
    \begin{equation*}
      A \coloneqq
      \frac{C K}{2 (1 - b)}
      +
      L (K + 1)
      D
      .
    \end{equation*}
  \item\label{item:OFW-adversarial}
    Moreover, if all the \(f_{t}\) are \(S t^{-s}\)-strongly convex,
    then with the choice
  \begin{equation*}
    \gamma_{t}
    =
    t^{(b + s - 2) / 3}
    ,
  \end{equation*}
  the \(x_{t}\) satisfy
  \begin{equation}
    \label{eq:OFW-error}
    \frac{1}{T} \sum_{t=1}^{T} (f_{t}(x_{T}) - f_{t}(x_{T}^{*}))
    \leq
    A T^{- (2 (1 + b) - s) / 3}
    ,
  \end{equation}
  where
  \begin{equation*}
    A
    \coloneqq
    2
    \left(
      (K + 1) (K + 2) \frac{L^{2}}{S}
      + \frac{C K}{2 (1 - b)}
    \right)
    .
  \end{equation*}
  \end{enumerate}
\begin{proof}
We prove only Claim~\ref{item:OFW-adversarial},
as the proof of Claim~\ref{item:OFW-stochastic}
is similar and simpler.
Let \(F_{T} \coloneqq \sum_{t=1}^{T} f_{t}\).
Furthermore, let
\(\overline{h}_{T} \coloneqq A T^{1 - (2 (1 + b) - s) / 3}\)
be \(T\) times the right-hand side of Equation~\eqref{eq:OFW-error}.
In particular, \(F_{T}\) is \(S_{T}\)-strongly convex,
and smooth with curvature at most \(C_{F_{T}}\)
where
\begin{align}
  C_{F_{T}} &\coloneqq
  \frac{C T^{1 - b}}{1 - b}
  \geq
  C \sum_{t=1}^{T} t^{-b}
  ,
  &
  S_{T} &\coloneqq
  S T^{1-s}
  \leq
  S \sum_{t=1}^{T} t^{-s}
  .
\end{align}
We prove \(F_{t}(x_{t}) - F_{t}(x_{t}^{*}) \leq h_{t}
\leq \overline{h}_{t}\)
by induction on \(t\).
The case \(t=1\) is clear.
Let \(\overline{\Phi}_{t}\) denote the value of \(\Phi_{t}\)
in Line~\ref{line:big-Phi-t}, while we reserve \(\Phi_{t}\)
to denote its value as used in Line~\ref{line:upper-bound}.
We start by showing
\(F_{t}(x_{t+1}) - F_{t}(x_{t}^{*}) \leq \Phi_{t}
\leq \overline{\Phi}_{t}\).
We distinguish two cases depending on the oracle answer \(v_{t}\)
from Line~\ref{line:OFW-augmented-v}.
For a negative oracle answer (\(v_{t} = \algorithmicfalse\)),
we have \(\Phi_{t} = \overline{\Phi}_{t}\)
and the weak separation oracle
asserts
\(\max_{y \in P} \nabla F_{t}(x_{t}) (x_{t} - y)
\leq \Phi_{t}\),
which combined with the convexity of \(F_{t}\) provides
\begin{equation*}
  F_{t}(x_{t+1}) - F_{t}(x_{t}^{*})
  =
  F_{t}(x_{t}) - F_{t}(x_{t}^{*})
  \leq
  \nabla F_{t}(x_{t}) (x_{t} - x_{t^{*}})
  \leq \Phi_{t} = \overline{\Phi}_{t}.
\end{equation*}
Otherwise, for a positive oracle answer,
Line~\ref{line:Phi-update} and
the induction hypothesis provides
\(F_{t}(x_{t+1}) - F_{t}(x_{t}^{*})
\leq h_{t} + F_{t}(x_{t+1}) - F_{t}(x_{t}) = \Phi_{t}\).
To prove \(\Phi_{t} \leq \overline{\Phi}_{t}\),
we apply the smoothness of \(F_{t}\)
followed by the inequality provided by the choice of \(v_{t}\):
\begin{equation*}
  F_{t}(x_{t+1}) - F_{t}(x_{t}) - \frac{C_{F_{t}} \gamma_{t}^{2}}{2}
  \leq
  \nabla F_{t}(x_{t}) (x_{t+1} - x_{t})
  =
  \gamma_{t} \nabla F_{t}(x_{t}) (v_{t} - x_{t})
  \leq
  - \frac{\gamma_{t} \overline{\Phi}_{t}}{K}
  .
\end{equation*}
Rearranging provides the inequality:
\begin{equation*}
  \Phi_{t}
  =
  h_{t} + F_{t}(x_{t+1}) - F_{t}(x_{t})
  \leq
  h_{t}
  - \frac{\gamma_{t} \overline{\Phi}_{t}}{K}
  + \frac{C_{F_{t}} \gamma_{t}^{2}}{2}
  =
  \overline{\Phi}_{t}
  .
\end{equation*}
For later use,
we bound the difference between \(\overline{h}_{t}\)
and \(\overline{\Phi}_{t}\)
using the value of parameters, \(h_{t} \leq \overline{h}_{t}\),
and \(\gamma_{t} \leq 1\):
\begin{equation*}
  \overline{h}_{t} - \overline{\Phi}_{t}
  \geq
  \overline{h}_{t} -
  \frac{\overline{h}_{t} + \frac{C_{F_{t}} \gamma_{t}^{2}}{2}}
  {1 + \frac{\gamma_{t}}{K}}
  =
  \frac{\frac{\overline{h}_{t} \gamma_{t}}{K}
    - \frac{C_{F_{t}} \gamma_{t}^{2}}{2}}
  {1 + \frac{\gamma_{t}}{K}}
  \geq
  \frac{\frac{\overline{h}_{t} \gamma_{t}}{K}
    - \frac{C_{F_{t}} \gamma_{t}^{2}}{2}}
  {1 + \frac{1}{K}}
  =
  \frac{A - \frac{C K}{2 (1 - b)}}{K + 1}
  t^{[2 s - (1 + b)] / 3}
  .
\end{equation*}

We now apply
\(F_{t}(x_{t+1}) - F_{t}(x_{t}^{*}) \leq \Phi_{t}\),
together with
convexity of \(f_{t+1}\),
and the minimality
\(F_{t}(x_{t}^{*}) \leq F_{t}(x_{t+1}^{*})\) of \(x_{t}^{*}\),
followed by strong convexity of \(F_{t+1}\):
\begin{equation}
  \label{eq:OFW-recursion}
 \begin{split}
  F_{t+1}(x_{t+1}) - F_{t+1}(x_{t+1}^{*})
  &
  \leq
  (F_{t}(x_{t+1}) - F_{t}(x_{t}^{*}))
  +
  (f_{t+1}(x_{t+1}) - f_{t+1}(x_{t+1}^{*}))
  \\
  &
  \leq
  \Phi_{t}
  + \dualnorm{\nabla_{t+1}} \cdot \norm{x_{t+1} - x_{t+1}^{*}}
  \\
  &
  \leq
  \Phi_{t}
  + \dualnorm{\nabla_{t+1}} \sqrt{\frac{2}{S_{t+1}}
    (F_{t+1}(x_{t+1}) - F_{t+1}(x_{t+1}^{*}))}
  .
 \end{split}
\end{equation}
Solving the quadratic inequality provides
\begin{equation}
  \label{eq:OFW-new-bound}
  F_{t+1}(x_{t+1}) - F_{t+1}(x_{t+1}^{*})
  \leq
  \Phi_{t}
  + \frac{\dualnorm{\nabla_{t+1}}^{2}}{S_{t+1}}
  +
  2 \sqrt{\frac{\dualnorm{\nabla_{t+1}}^{2}}{2 S_{t+1}}
    \left(
      \frac{\dualnorm{\nabla_{t+1}}^{2}}{2 S_{t+1}}
      +
      \Phi_{t}
    \right)}
  .
\end{equation}
From Equation~\eqref{eq:OFW-recursion},
ignoring the last line,
we also obtain
\(F_{t+1}(x_{t+1}) - F_{t+1}(x_{t+1}^{*})
\leq \Phi_{t} + \dualnorm{\nabla_{t+1}} D\)
via the estimate \(\norm{x_{t+1} - x_{t+1}^{*}} \leq D\).
Thus \(F_{t+1}(x_{t+1}) - F_{t+1}(x_{t+1}^{*}) \leq h_{t+1}\),
by Line~\ref{line:upper-bound},
as claimed.

Now we estimate the right-hand side of
Equation~\eqref{eq:OFW-new-bound} by
using the actual value of the parameters,
the estimate \(\dualnorm{\nabla_{t+1}} \leq L\),
and the inequality \(s + b \leq 2\).
In fact, we estimate a proxy for the right-hand side.
Note that \(A\) was chosen to satisfy
the second inequality:
\begin{equation*}
 \begin{split}
  \frac{L^{2}}{S_{t+1}}
  +
  2 \sqrt{\frac{L^{2}}{2 S_{t+1}} \overline{h}_{t}}
  &
  \leq
  \frac{L^{2}}{S t^{1-s}}
  +
  2 \sqrt{\frac{L^{2}}{2 S t^{1-s}} \overline{h}_{t}}
  \leq
  \frac{L^{2}}{S} t^{[2 s - (1 + b)] / 3}
  +
  2 \sqrt{\frac{L^{2}}{2 S t^{1-s}} \overline{h}_{t}}
  \\
  &
  =
  \left(
    \frac{L^{2}}{S}
    +
    \sqrt{2 \frac{L^{2}}{S} A}
  \right)
  t^{[2 s - (1 + b)] / 3}
  \leq
  \frac{A - \frac{C K}{2 (1 - b)}}{K + 1}
  t^{[2 s - (1 + b)] / 3}
  \\
  &
  \leq
  \overline{h}_{t} - \overline{\Phi}_{t}
  \leq
  \overline{h}_{t} - \Phi_{t}
  .
 \end{split}
\end{equation*}
In particular,
\(\frac{L^{2}}{2 S_{t+1}} + \Phi_{t} \leq \overline{h}_{t}\)
hence combining with Equation~\eqref{eq:OFW-new-bound}
we obtain
\begin{equation}
 \begin{split}
  h_{t+1}
  &
  \leq
  \Phi_{t}
  + \frac{L^{2}}{S_{t+1}}
  +
  2 \sqrt{\frac{L^{2}}{2 S_{t+1}}
    \left(
      \frac{L^{2}}{2 S_{t+1}}
      +
      \Phi_{t}
    \right)}
  \\
  &
  \leq
  \Phi_{t}
  +
  \frac{L^{2}}{S_{t+1}}
  +
  2 \sqrt{\frac{L^{2}}{2 S_{t+1}} \overline{h}_{t}}
  \\
  &
  \leq
  \overline{h}_{t}
  \leq
  \overline{h}_{t+1}
  .
  \qedhere
 \end{split}
\end{equation}
\end{proof}
\end{theorem}


\subsection{Stochastic and Adversarial Versions}
\label{sec:offl-stoch-vers}

Complementing the offline algorithms from
Section~\ref{sec:offline-frank-wolfe},
we will now derive various online versions.
The presented cases here are similar to
those in \cite{hazan2012projection} and thus we state them without proof.

For stochastic cost functions \(f_{t}\),
we obtain bounds from
Theorem~\ref{thm:OFW-error}~\ref{item:OFW-stochastic} similar to
\cite[Theorems~4.1 and~4.3]{hazan2012projection}
(with \(\delta\) replaced by \(\delta / T\) in the bound
to correct an inaccuracy in the original argument).
The proof is analogous and hence omitted, but
note that
\(\norm[2]{y_{1} - y_{2}} \leq \sqrt{\norm[1]{y_{1} - y_{2}}
  \norm[\infty]{y_{1} - y_{2}}} \leq \sqrt{k}\)
for all \(y_{1}, y_{2} \in P\).
\begin{corollary}
  \label{cor:stochastic-error}
  Let \(f_{t}\) be convex functions
  sampled i.i.d.\ with expectation \(\expectation{f_{t}} = f^{*}\),
  and \(\delta > 0\).
  Assume that the \(f_{t}\) are \(L\)-Lipschitz in the \(2\)-norm.
  \begin{enumerate}
  \item\label{item:stochastic-error_smooth}
    If all the \(f_{t}\) are smooth with curvature at most \(C\),
    then Algorithm~\ref{alg:online-FW-WSep}
    applied to the \(f_{t}\)
    (with \(b = 0\))
    yields with probability \(1 - \delta\)
    \begin{equation}
      \label{eq:stochastic-error_smooth}
      \sum_{t=1}^{T} f^{*}(x_{t})
      - \min_{x \in P} \sum_{t=1}^{T} f^{*}(x)
      \leq
      O\left(
        C \sqrt{T}
        + L k\sqrt{n T \log (n T^{2} / \delta) \log T}
      \right)
      .
    \end{equation}
  \item\label{item:stochastic-error_non-smooth}
    Without any smoothness assumption,
    Algorithm~\ref{alg:online-FW-WSep}
    (applied to smoothenings of the \(f_{t}\))
    provides with probability \(1 - \delta\)
    \begin{equation}
      \label{eq:stochastic-error_non-smooth}
      \sum_{t=1}^{T} f^{*}(x_{t})
      - \min_{x \in P} \sum_{t=1}^{T} f^{*}(x)
      \leq
      O\left(
        \sqrt{n} L k T^{2/3}
        + L k\sqrt{n T \log (n T^{2} / \delta) \log T}
      \right)
      .
    \end{equation}
  \end{enumerate}
\end{corollary}

Similar to \cite[Theorem~4.4]{hazan2012projection},
from Theorem~\ref{thm:OFW-error}~\ref{item:OFW-adversarial}
we obtain the following regret bound for adversarial cost functions
with an analogous proof.
\begin{corollary}
  \label{cor:adversarial-error}
  For any \(L\)-Lipschitz convex cost functions \(f_{t}\),
  Algorithm~\ref{alg:online-FW-WSep}
  applied to the functions
  \(\tilde{f}_{t}(x) \coloneqq
  \nabla f_{t}(x_{t}) x
  + \frac{2L}{\sqrt{k}} t^{-1/4} \norm[2]{x - x_{1}}^{2}\)
  (with \(b = s = 1/4\), \(C = L \sqrt{k}\), \(S = L / \sqrt{k}\),
  and Lipschitz constant \(3 L\))
  achieving regret
  \begin{equation}
    \label{eq:adversarial-error}
    \sum_{t=1}^{T} f_{t}(x_{t})
    - \min_{x \in P} \sum_{t=1}^{T} f_{t}(x)
    \leq
    O(L \sqrt{k} T^{3/4})
  \end{equation}
  with at most \(T\) calls to the weak separation oracle.
\end{corollary}

Note that the gradient of the \(\tilde{f}_{t}\)
are easily computed via the formula
\(\nabla \tilde{f}_{t}(x) = \nabla f_{t}(x_{t}) + 4 L t^{-1/4} (x -
x_{1}) / \sqrt{k}\),
particularly because the gradient of the \(f_t\)
need not be recomputed, so that we obtain a weak separation-based
stochastic gradient descent algorithm, where we only have access to
the \(f_t\)
through a stochastic gradient oracle, while retaining all the
favorable properties of the Frank-Wolfe algorithm with a convergence
rate \(O(T^{-1/4})\) (c.f., \cite{garber2013linearly}).

\section{Weak Separation through Augmentation}
\label{sec:weak-separate-augment}

So far we realized the weak separation oracle via lazy optimization.
We will now create a (weak) separation oracle for integral polytopes,
employing an even weaker, so-called augmentation oracle,
which only provides an improving solution
but provides no guarantee with respect to optimality.
We call this approach \emph{lazy augmentation}.
This is especially useful when a fast augmentation oracle is available
or the vertices of the underlying polytope \(P\) are particularly
sparse, i.e.,  \(\norm[1]{y_{1} - y_{2}} \leq k \ll n\)
  for all \(y_{1}, y_{2} \in P\), where \(n\) is the ambient dimension
  of \(P\). As before theoretical convergence rates are maintained. 

For simplicity of exposition we restrict to \(0/1\)
polytopes \(P\)
here.  For general integral polytopes, one considers a
so-called \emph{directed augmentation oracle}, which can be similarly
linearized after splitting variables in positive and negative parts;
we refer the interested reader to see
\citep{schulz2002complexity,bodic2015solving} for an in-depth
discussion.
 
Let \(k\) denote the \(\ell_{1}\)-diameter of \(P\).
Upon presentation
with a 0/1 solution \(x\)
and a linear objective \(c \in \R^n\),
an augmentation oracle
either provides an improving 0/1 solution \(\bar x\)
with \(c \bar x < c x\) or asserts optimality for \(c\):

\begin{oracle}[H]
  \label{alg:linearAugmentation}
  \caption{Linear Augmentation Oracle \AUG{P}(c, x)}
  \begin{algorithmic}
    \REQUIRE
    linear objective \(c \in \mathbb{R}^{n}\),
    vertex \(x \in P\)
    \ENSURE vertex \(\bar{x} \in P\) with \(c \bar{x} < c x\)
    when exists, otherwise \(\bar{x} = x\)
  \end{algorithmic}
\end{oracle}

Such an oracle is significantly weaker than a linear optimization
oracle but also significantly easier to implement and much faster; we
refer the interested reader to
\citep{grotschel1993combinatorial,schulz19950,schulz2002complexity}
for an extensive list of examples.  While augmentation and
optimization are polynomially equivalent (even for convex integer
programming \citep{OerWW14}) the current best linear optimization
algorithms based on an augmentation oracle are slow for general
objectives. While optimizing  an \emph{integral} objective \(c \in \R^{n}\)
needs \(O(k \log \norm[\infty]{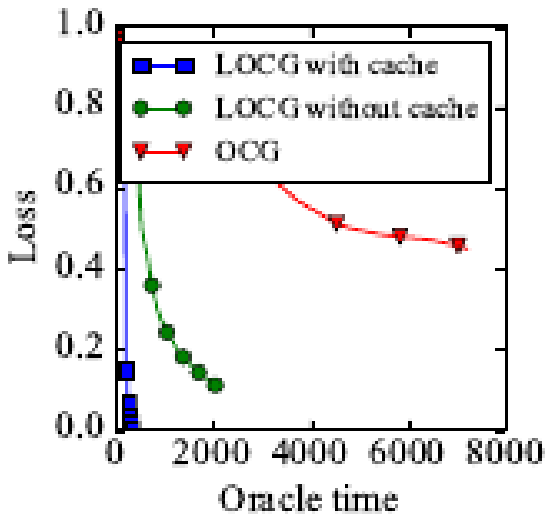})\) calls to an augmentation oracle (see
\citep{schulz19950,schulz2002complexity,bodic2015solving}), a general
objective function, such as the gradient in Frank–Wolfe algorithms has
only an \(O(k n^3)\) guarantee in terms of required oracle calls (e.g., via simultaneous diophantine approximations
\citep{frank1987application}), which is not desirable for large \(n\).
In contrast, here we use an augmentation oracle to perform separation,
without finding the optimal solution.  Allowing a multiplicative error
\(K > 1\),
we realize an augmentation-based weak separation oracle (see
Algorithm~\ref{alg:weak-separate}), which decides given a linear
objective function \(c \in \mathbb{R}^{n}\),
an objective value \(\Phi > 0\),
and a starting point \(x \in P\),
whether there is a \(y \in P\)
with \(c (x - y) > \Phi / K\)
or \(c (x - y) \leq \Phi\)
for all \(y \in P\).
In the former case, it actually provides a certifying \(y \in P\),
i.e., with \(c (x - y) > \Phi / K\).
Note that a constant accuracy \(K\)
requires a linear number of oracle calls in the diameter \(k\),
e.g., \(K = (1 - 1/e)^{-1} \approx 1.582\) 
needs at most \(N \leq k\) oracle calls, which can be much smaller
than the ambient dimension of the polytope. 

At the beginning, in Line~\ref{line:make-integral}, the algorithm
has to replace the input point \(x\) with an integral point \(x_{0}\).
If the point \(x\) is given as a convex combination of integral
points, then a possible solution is to evaluate the objective \(c\)
on these integral points, and choose \(x_{0}\) the first one with
\(c x_{0} \leq cx\).
This can be easily arranged for Frank–Wolfe algorithms as they
maintain convex combinations.

\begin{algorithm}
  \caption{Augmenting Weak Separation
   \(\LPsep{P}(c, x, \Phi, K)\)}
  \label{alg:weak-separate}
  \begin{algorithmic}[1]
    \REQUIRE
    linear objective \(c \in \mathbb{R}^{n}\),
    point \(x \in P\),
    objective value \(\Phi > 0\);
    accuracy \(K > 1\)
    \ENSURE Either
    \begin{enumerate*}
    \item \(y \in P\) vertex with
      \(c (x - y) > \Phi / K\)
    \item \FALSE: \(c (x - z) \leq \Phi\) for all \(z \in P\)
    \end{enumerate*}
    \STATE
    \(N \leftarrow \lceil \log (1 - 1/K) / \log (1 - 1/k) \rceil\)
    \STATE Choose \(x_{0} \in P\) vertex with \(c x_{0} \leq c x\).
    \label{line:make-integral}
    \FOR{\(i = 1\) \TO \(N\)}
      \IF{\(c (x - x_{i-1}) \geq \Phi\)}
        \RETURN \(x_{i-1}\)
        \label{line:ws-too-big}
      \ENDIF
      \STATE
      \label{line:oracle-call}
      \(x_{i} \leftarrow \AUG{P}(c + \frac{\Phi - c (x - x_{i-1})}{k}
      (\allOne - 2 x_{i-1}),x_{i-1})\)
      \IF{\(x_{i} = x_{i-1}\)}
        \RETURN \FALSE
        \label{line:ws-infeasible}
      \ENDIF
    \ENDFOR
    \RETURN \(x_{N}\)
    \label{line:ws-end}
  \end{algorithmic}
\end{algorithm}

\begin{proposition}
  \label{prop:weak-separate}
  Assume \(\norm[1]{y_{1} - y_{2}} \leq k\)
  for all \(y_{1}, y_{2} \in P\).
  Then Algorithm~\ref{alg:weak-separate} is correct,
  i.e., it outputs either
  \begin{enumerate*}
  \item
    \(y \in P\) with
    \(c (x - y) > \Phi / K\)
  \item \(\mathbf{false}\)
  \end{enumerate*}
  In the latter case
  \(c (x - y) \leq \Phi\) for all \(y \in P\) holds.
  The algorithm calls \(\AUG{P}\) at most
  \(N \leq \lceil \log (1 - 1/K) / \log (1 - 1/k) \rceil\)
  many
  times.
\begin{proof}
First note
that \((\allOne - 2x) v + \norm[1]{x} = \norm[1]{v - x}
\) for \(x,v \in \{0,1\}^n\),
hence Line~\ref{line:oracle-call} is equivalent to
\(x_{i} \leftarrow\AUG{P}(c + \frac{\Phi - c (x - x_{i-1})}{k}
      \norm[1]{\cdot - x_{i-1}},x_{i-1})\).

The algorithm obviously calls the oracle at most \(N\) times by
design, and always returns a value, so we need to verify only the
correctness of the returned value.
We distinguish cases according to the output.

Clearly, Line~\ref{line:ws-too-big} always returns an
\(x_{i-1}\) with \(c (x - x_{i-1}) \geq \Phi > [1 - (1 - 1/k)^{N}] \Phi\).
When Line~\ref{line:ws-infeasible} is executed,
the augmentation oracle just returned \(x_{i} = x_{i-1}\),
i.e., for all \(y \in P\)
\begin{equation}
  c x_{i-1}
  \leq
  c y + \frac{\Phi - c (x - x_{i-1})}{k} \norm[1]{y - x_{i-1}}
  \leq
  c y + \frac{\Phi - c (x - x_{i-1})}{k} k
  = c (y - x) + c x_{i-1} + \Phi
  ,
\end{equation}
so that \(c (x - y) \leq \Phi\), as claimed.

Finally, when Line~\ref{line:ws-end} is executed,
the augmentation oracle has found an improving vertex \(x_{i}\)
at every iteration,
i.e.,
\begin{equation}
  c x_{i-1}
  >
  c x_{i} + \frac{\Phi - c (x - x_{i-1})}{k} \norm[1]{x_{i} - x_{i-1}}
  \geq
  c x_{i} + \frac{\Phi - c (x - x_{i-1})}{k}
  ,
\end{equation}
using \(\norm[1]{x_{i} - x_{i-1}} \geq 1\) by integrality.
Rearranging provides the convenient form
\begin{equation}
  \Phi - c (x - x_{i})
  <
  \left(
    1 - \frac{1}{k}
  \right)
  [\Phi - c (x - x_{i-1})]
  ,
\end{equation}
which by an easy induction provides
\begin{equation}
  \Phi - c (x - x_{N})
  <
  \left(
    1 - \frac{1}{k}
  \right)^{N}
  [\Phi - c (x - x_{0})]
  \leq
  \left(
    1 - \frac{1}{K}
  \right)
  \Phi
  ,
\end{equation}
i.e., \(c (x - x_{N}) \geq \frac{\Phi}{K}\),
finishing the proof.
\end{proof}
\end{proposition}

\section{Experiments}
\label{sec:experiments}

We implemented and compared the
parameter-free variant of LCG (Algorithm~\ref{alg:CG-imp-imp})
to the standard Frank-Wolfe algorithm (CG),
then Algorithm~\ref{alg:pairwise-cond-gradient}
(LPCG) to the Pairwise Conditional Gradient algorithm (PCG) of
\cite{LDLCC2016}, as well as
Algorithm~\ref{alg:online-FW-WSep} (LOCG) to the Online Frank-Wolfe
algorithm (OCG) of \cite{hazan2012projection}. While we did implement
the Local Conditional Gradient algorithm of \cite{garber2013linearly}
as well, the very large
constants in the original algorithms made it impractical to run.
Unless stated otherwise the weak separation oracle is implemented as
sketched in Algorithm~\ref{alg:LPSepLPOracle} through caching and
early termination of
the original LP oracle.

We have used \(K = 1.1\) and \(K = 1\) as multiplicative factors for
the weak separation oracle; for the impact of the choice of \(K\) see
Section~\ref{sec:effect-k}. For the baseline algorithms we use inexact
variants, i.e., we solve linear optimization problems only
approximately. This is a significant speedup in favor of non-lazy
algorithms at the (potential) cost of accuracy, while neutral to lazy
optimization as it solves an even more relaxed problem anyways.  To
put things in perspective, the non-lazy baselines could not complete
even a single iteration for a significant fraction of the considered
problems in the given time frame if we were to exactly solve the
linear optimization problems. In terms of using line search, for all
tests we treated all algorithms equally: either \emph{all} or
\emph{none} used line search. If not stated otherwise, we used (simple
backtracking) line search. 

The linear optimization oracle over \(P \times P\) for LPCG was
implemented by calling the respective oracle over \(P\) twice: once
for either component.  Contrary to the non-lazy version, the lazy
algorithms depend on the initial upper bound \(\Phi_{0}\).  For the
instances that need a very long time to solve the (approximate) linear
optimization even once, we used a binary search for \(\Phi_{0}\) for
the lazy algorithms: starting from a conservative initial value, using
the update rule \(\Phi_{0} \leftarrow \Phi_{0} / 2\) until the
separation oracle returns an improvement for the first time and then
we start the algorithm with \(2 \Phi_0\), which is an upper bound on
the Wolfe gap and hence also on the primal gap.  This initial phase is
also included in the reported wall-clock time.  Alternatively, if the
linear optimization was less time consuming we used a single
(approximate) linear optimization at the start to obtain an initial
bound on \(\Phi_0\) (see e.g., Section~\ref{sec:frank-wolfe-with}).

In some cases, especially when the underlying feasible region
has a high dimension and the (approximate) linear
optimization can be solved relatively fast compared to the
cost of computing an inner product, we observed that the costs
of maintaining the cache was very high. In these cases we
reduced the cache size every \(100\) steps by keeping only the
\(100\) points that were used the most so far.  Both the number
of steps and the approximate size of the cache were chosen
arbitrarily, however \(100\) for both worked very well for all
our examples. Of course there are many different strategies
for maintaining the cache, which could be used here and
which could lead to further improvements in performance.

The stopping criteria for each of the experiments was
a given wall clock time limit in seconds.
The time limit was enforced separately for the main code
and the oracle code,
so in some cases the actual time used can be larger,
when the last oracle call started before the time limit
was reached and took longer than the time left.

We implemented all algorithms in \texttt{Python 2.7} with critical
functions \emph{cythonized} for performance employing \texttt{Numpy}.
We used these packages from the \texttt{Anaconda 4.2.0} distribution
as well as \texttt{Gurobi 7.0} \citep{optimization2016gurobi} as a
black box solver for the linear optimization oracle.  The weak
separation oracle was implemented via a callback function to stop
linear optimization as soon as a good enough feasible solution has
been found in a schema as outlined in
Algorithm~\ref{alg:LPSepLPOracle}. The parameters for Gurobi were kept
at their default settings except for enforcing the time limit of the
tests and setting the acceptable duality gap to \(10\%\), allowing
Gurobi to terminate the linear optimization early avoiding the
expensive \emph{proof} of optimality. This is used to realize the
inexact versions of the baseline algorithms.  All experiments were
performed on a 16-core machine with Intel Xeon E5-2630 v3 @ 2.40GHz
CPUs and 128GB of main memory. While our code does not explicitly use
multiple threads, both Gurobi and the numerical libraries use multiple
threads internally.

\subsection{Computational results}
\label{sec:comp}

We performed computational tests on a large variety of different
problems that are instances of the three machine learning tasks
\emph{video colocalization}, \emph{matrix completion}, and \emph{structured regression}.

\paragraph{Video colocalization.} Video colocalization is the problem of identifying objects
in a sequence of multiple frames in a video. In \cite{joulin2014efficient}
it is shown that video colocalization can be reduced to optimizing a quadratic objective
function over a flow or a path polytope, which is the problem we
are going to solve.
The resulting linear program is an instance of the minimum-cost
network flow problem,
see \cite[Eq.~(3)]{joulin2014efficient} for the concrete linear
program and more details.
The quadratic functions are of the form \(\norm{Ax - b}^2\)
where we choose the non-zero entries in
\(A\) according to a density parameter at random and then each of these
entries to be \([0,1]\)-uniformly distributed, while
\(b\) is chosen as a linear combination of the columns of \(A\) with
random multipliers from \([0,1]\). For some of the instances
we also use \(\norm{x-b}^2\) as the objective function with
\(b_i \in [0,1]\) uniformly at random.

\paragraph{Matrix completion.} The formulation of the matrix completion problem we are going
to use is the following:
\begin{equation}
  \min_{X} \sum_{(i,j) \in \Omega} \abs{X_{i,j} - A_{i,j}}^2
  \qquad \text{ s.t.} \quad \norm[*]{X} \le R,
 \label{eq:matrix_completion}
\end{equation}
where \(\norm[*]{\cdot}\) denotes the nuclear norm, i.e.,
\(\norm[*]{A} = \operatorname{Tr}(\sqrt{A^t A})\).  The set
\(\Omega\), the matrix \(A\) and \(R\) are given parameters.
Similarly to \cite{lan2014conditional} we generate the \(m\times n\)
matrix \(A\) as the product of \(A_L\) of size \(m\times r\) and
\(A_R\) of size \(r\times n\). The entries in \(A_L\) and \(A_R\) are
chosen from a standard Gaussian distribution.  The set \(\Omega\) is
chosen uniformly of size
\(s = \min \{5r(m + n - r), \lceil 0.99mn\rceil\}\).  The linear
optimization oracle is implemented in this case by a singular value
decomposition of the linear objective function and we essentially
solve the LP to (approximate) optimality. The matrix completion tests
will only demonstrate the impact of caching solutions. Note that this
test is also informative as due to the \lq{}roundness\rq{} of the
feasible region the solution of the actual LP oracle will induce a
direction that is equal to the true gradient and as such it provides
insight into how much per-iteration progress is lost due to working
with gradient approximations from the weak separation oracle.

\paragraph{Structured regression.}
The structured regression problem consists of solving a
quadratic function of the form \(\norm{Ax - b}^2\) over some
structured feasible set or a polytope \(P\), i.e., we solve \(\min_{x \in P}
\norm{Ax - b}^2\). We construct the objective functions
in the same way as for the video colocalization problem.

\paragraph{Tests.}
In the following two sections we will present 
our results for various problems grouped
by the versions of the considered algorithms.
Every figure contains two columns, each containing one experiment.
We use different measures to report performance:
we report progress of loss or function value
in wall-clock time in the first row
(including time spent by the oracle),
in the number of iterations
in the second row,
and in the number of linear optimization calls in the last row.
Obviously, the latter only makes sense for the lazy algorithms.
In some other cases we
report in another row
the dual bound or Wolfe gap in wall-clock
time.
The red line denotes the non-lazy algorithm and
the green line denotes the lazy variants.
For each experiment we also report
the cache hit rate, which is the number of oracle calls answered with
a point from the cache over all oracle calls given in percent.

While we found convergence rates in the number of iterations quite
similar (as expected!), we consistently observe a significant speedup
in wall-clock time. In particular for many large-scale or hard
combinatorial problems, lazy algorithms performed several thousand
iterations whereas the non-lazy versions completed only a handful of
iterations due to the large time spent approximately solving the
linear optimization problem.
The observed cache hit rate was 
at least \(90\%\) in most cases,
and often even above \(99\%\).

\subsubsection{Offline Results}
\label{sec:offlineResults}

We describe the considered instances in the offline case separately
for the vanilla Frank-Wolfe method and the Pairwise Conditional
Gradient method.

\paragraph{Vanilla Frank-Wolfe Method}

We tested the vanilla Frank-Wolfe algorithm on the six video
colocalization instances with underlying path polytopes from
\url{http://lime.cs.elte.hu/~kpeter/data/mcf/netgen/}
(Figures~\ref{fig:netgen-small-vanilla},
\ref{fig:netgen-medium-vanilla} and
\ref{fig:netgen-large-vanilla}). In these instances we additionally
report the dual bound or Wolfe gap in wall clock time.  We further
tested the vanilla Frank-Wolfe algorithm on eight instances of the
matrix completion problem generated as described above, for which we
did not use line search; the parameter-free lazy variant is run with
approximate minimization as described in
Remark~\ref{sec:param-free-cond-linesearch}, the others use their
respective standard step sizes.  We provide the used
parameters for each example in the figures below
(Figures~\ref{fig:matrix-completion-small-vanilla},
\ref{fig:matrix-completion-large-vanilla},
\ref{fig:matrix-completion-alt-1-vanilla} and
\ref{fig:matrix-completion-alt-2-vanilla}).  The last tests for this
version were performed on three instances of the structured regression
problem, two with the feasible region containing flow-based
formulations of Hamiltonian cycles in graphs
(Figure~\ref{fig:tsp-vanilla}), and further tests on two cut polytope
instances (Figure~\ref{fig:maxcut-vanilla}) and on two spanning tree
instances of different size (Figure~\ref{fig:spanning-tree-vanilla}).

We observed a significant speedup of LCG compared to CG,
due to the faster iteration of the lazy algorithm.

\begin{figure*}
  \centering
  \small
  \begin{tabular}{*{2}{c}}
  \includegraphics[height=0.3\linewidth]{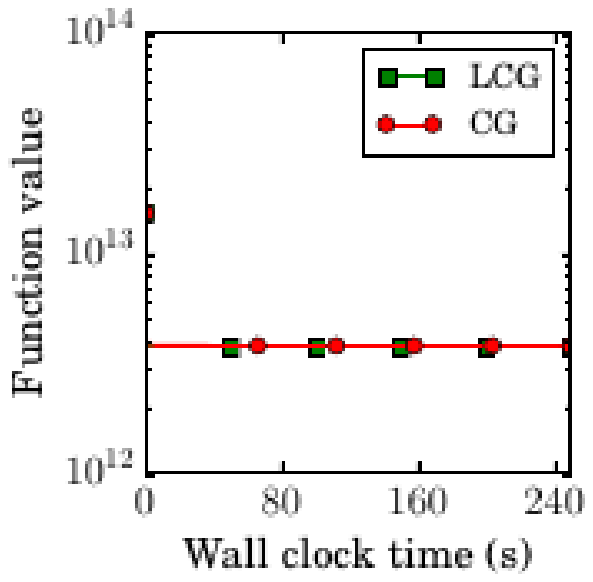}
  &
  \includegraphics[height=0.3\linewidth]{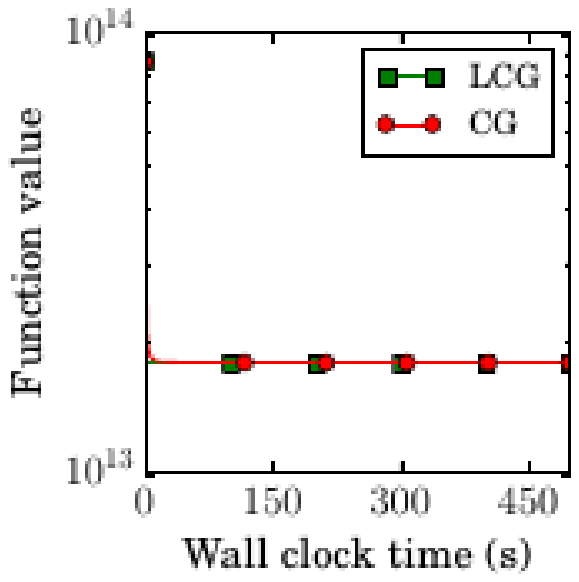}
  \\
  \includegraphics[height=0.3\linewidth]{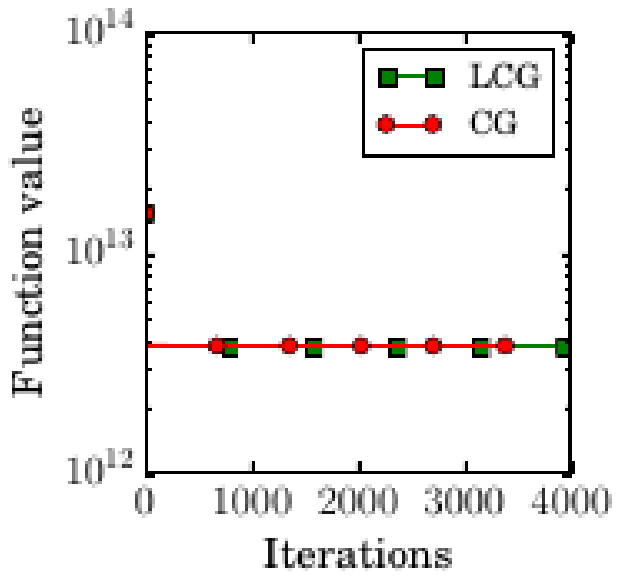}
  &
  \includegraphics[height=0.3\linewidth]{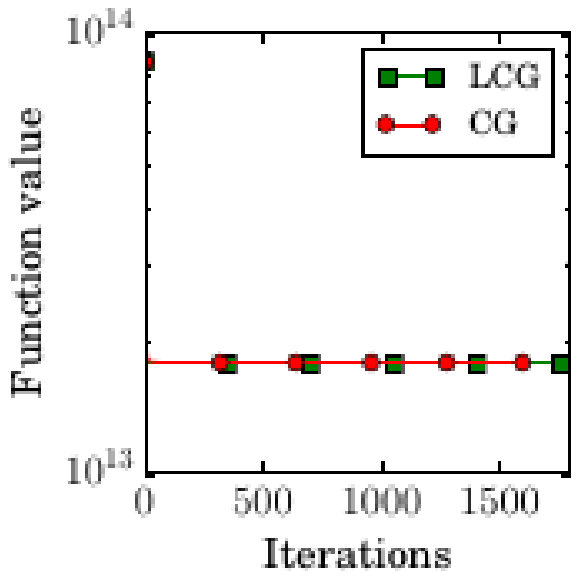}
  \\
  \includegraphics[height=0.3\linewidth]{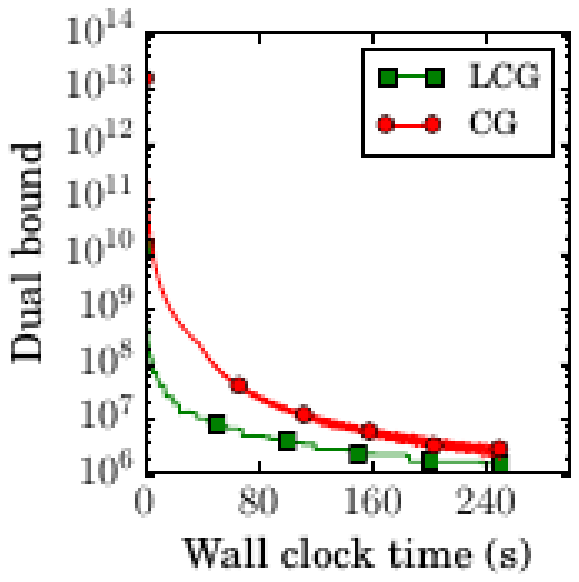}
  &
  \includegraphics[height=0.3\linewidth]{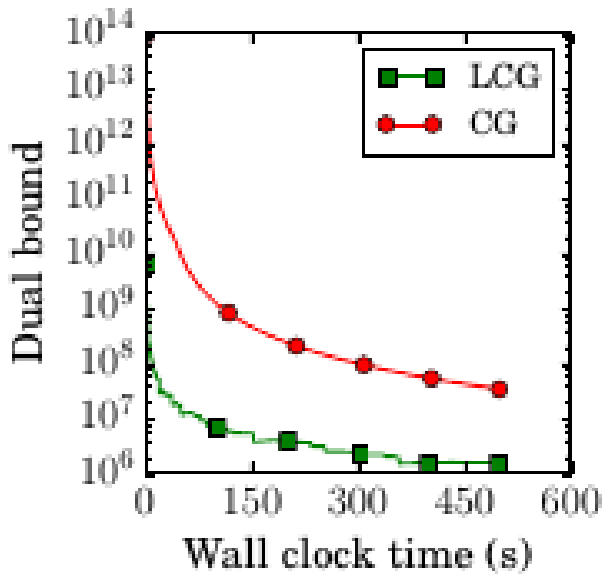}
  \\
  \includegraphics[height=0.3\linewidth]{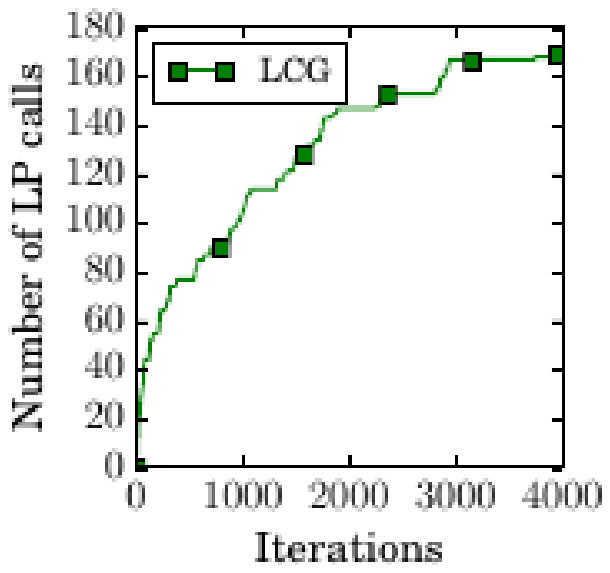}
  &
  \includegraphics[height=0.3\linewidth]{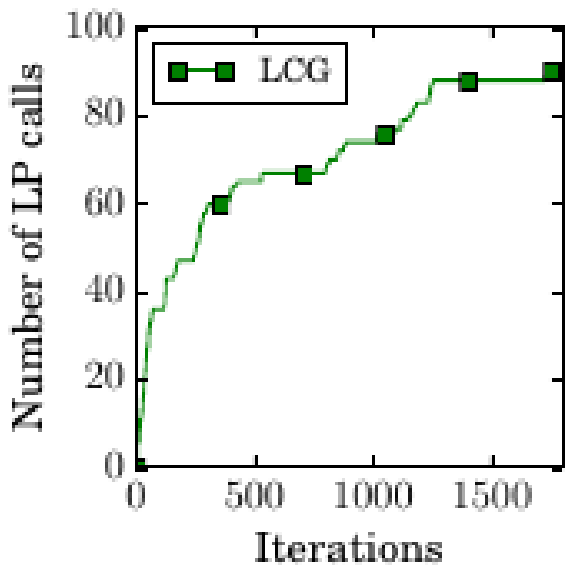}
  \\
  cache hit rate: \(95.72\%\)
  &
  cache hit rate: \(94.83\%\)
  \end{tabular}
  \caption{\label{fig:netgen-small-vanilla} LCG vs. CG on small netgen
  instances \texttt{netgen 08a} (left) and
  \texttt{netgen 10a} (right) with quadratic objective functions. In both cases both algorithms
  are able to reduce the function value very fast, however the dual bound or Wolfe gap
  is reduced much faster by LCG. Observe that the vertical axis is given with a logscale.
   }
\end{figure*}

\begin{figure*}
  \centering
  \small
  \begin{tabular}{*{2}{c}}
  \includegraphics[height=0.3\linewidth]{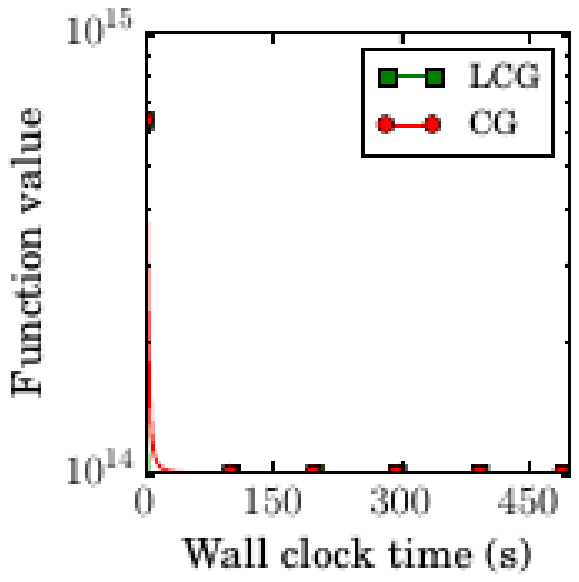}
  &
  \includegraphics[height=0.3\linewidth]{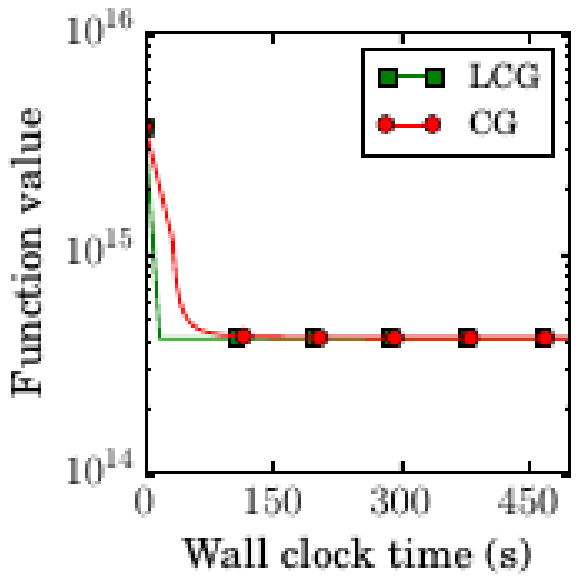}
  \\
  \includegraphics[height=0.3\linewidth]{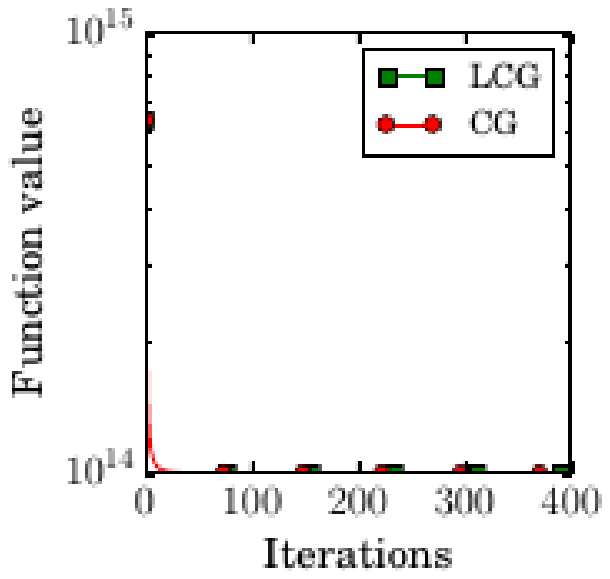}
  &
  \includegraphics[height=0.3\linewidth]{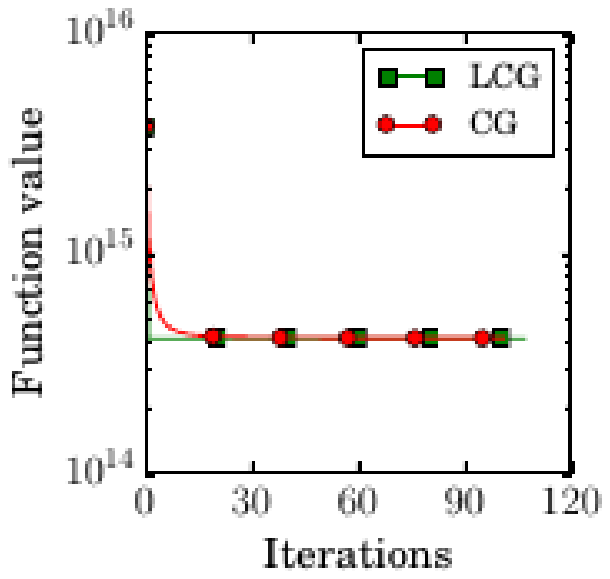}
  \\
  \includegraphics[height=0.3\linewidth]{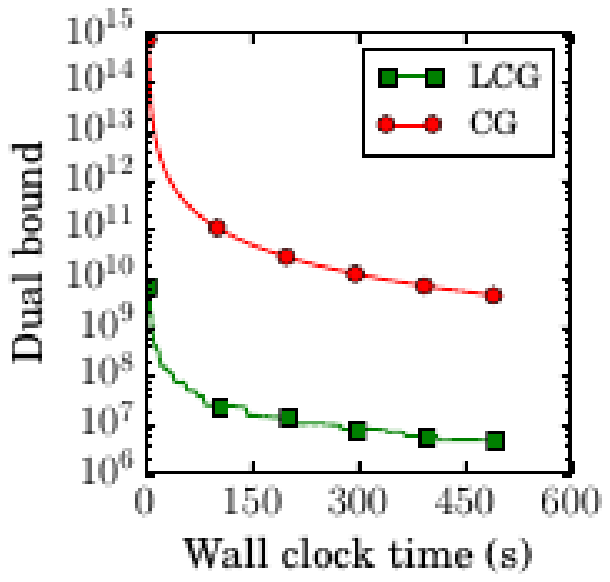}
  &
  \includegraphics[height=0.3\linewidth]{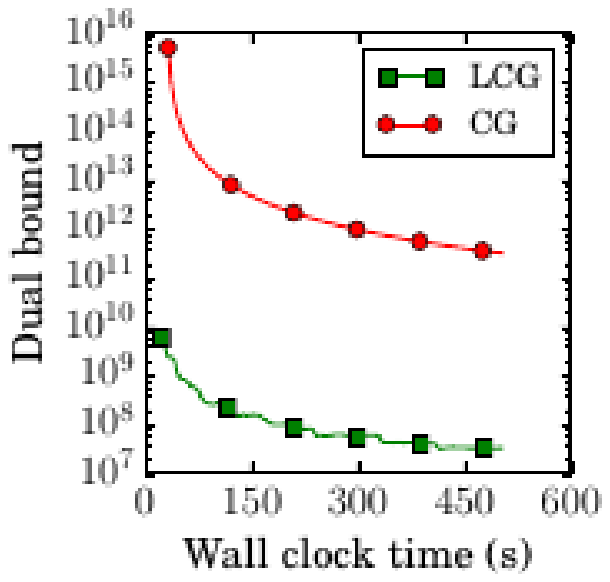}
  \\
  \includegraphics[height=0.3\linewidth]{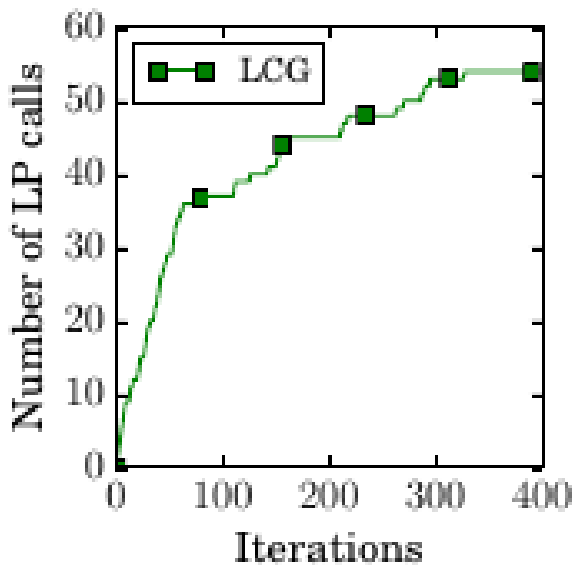}
  &
  \includegraphics[height=0.3\linewidth]{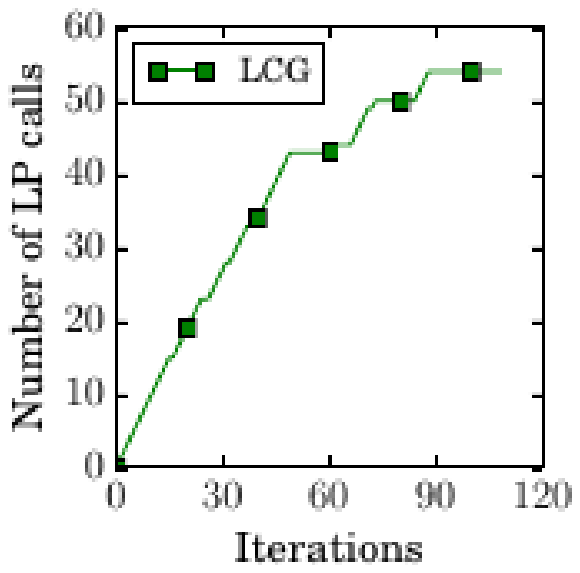}
  \\
  cache hit rate: \(86.43\%\)
  &
  cache hit rate: \(50.00\%\)
  \end{tabular}
  \caption{\label{fig:netgen-medium-vanilla} LCG vs. CG on medium sized netgen instances \texttt{netgen 12b} (left) and
  \texttt{netgen 14a} (right) with quadratic objective functions. The behavior of both versions
  on these instances is very similar to the small netgen instances (Figure~\ref{fig:netgen-small-vanilla}),
  however both in the function value and the dual bound the difference between the lazy and the non-lazy
  version is more prominent. Again, we used a logscale for the vertical axis.
  }
\end{figure*}

\begin{figure*}
  \centering
  \small
  \begin{tabular}{*{2}{c}}
  \includegraphics[height=0.3\linewidth]{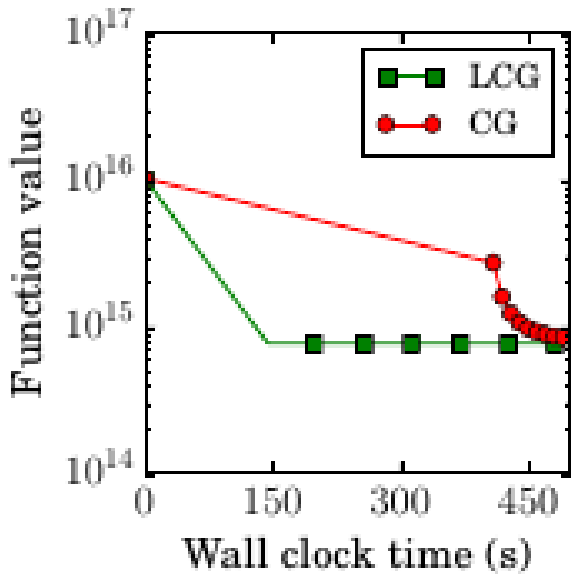}
  &
  \includegraphics[height=0.3\linewidth]{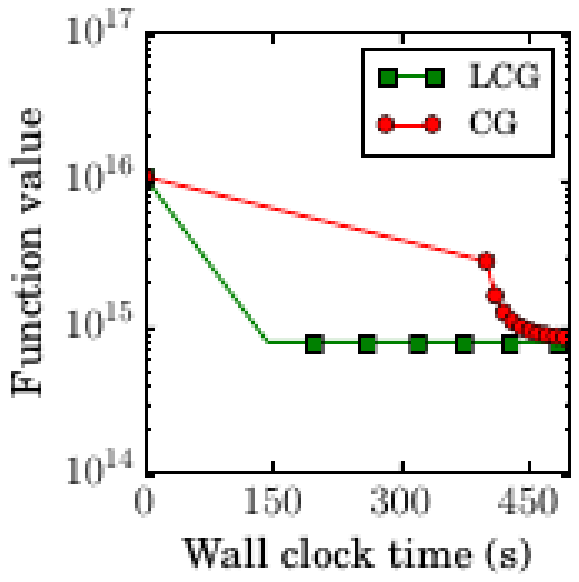}
  \\
  \includegraphics[height=0.3\linewidth]{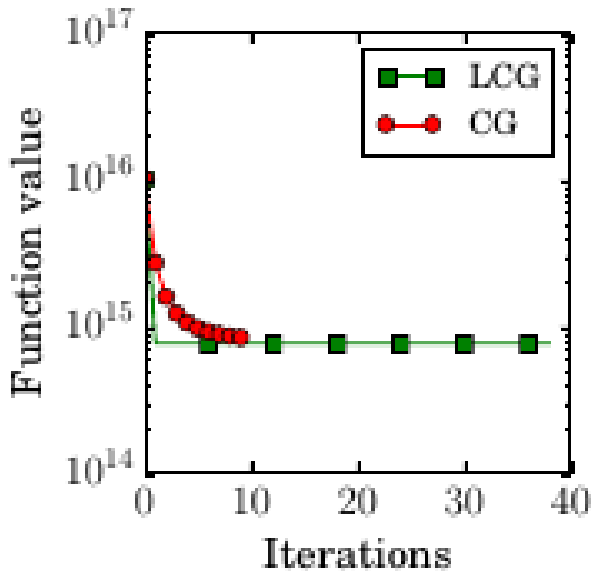}
  &
  \includegraphics[height=0.3\linewidth]{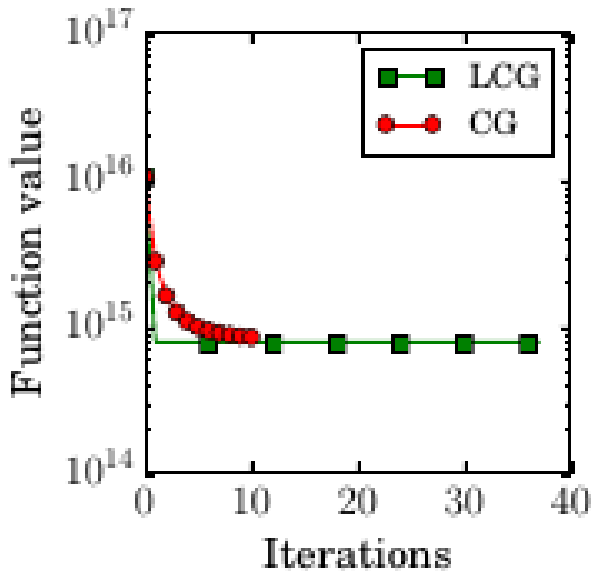}
  \\
  \includegraphics[height=0.3\linewidth]{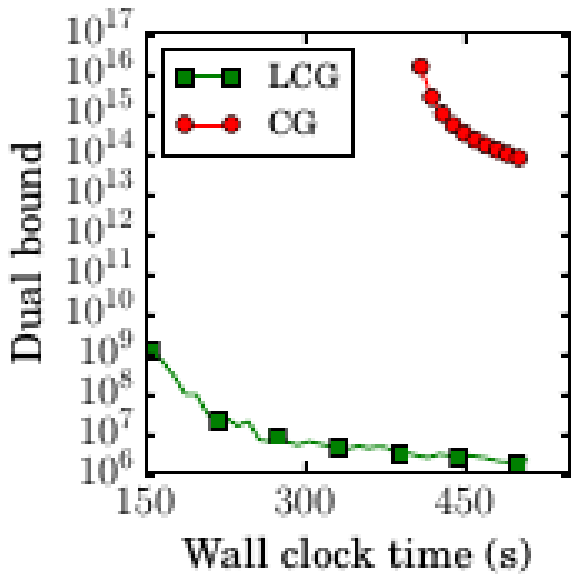}
  &
  \includegraphics[height=0.3\linewidth]{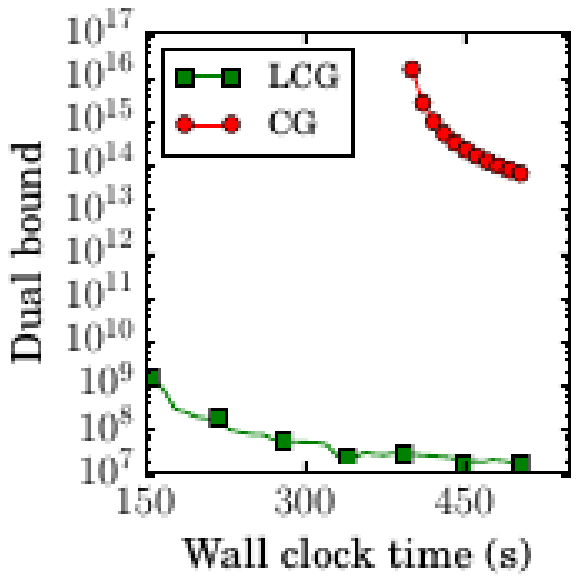}
  \\
  \includegraphics[height=0.3\linewidth]{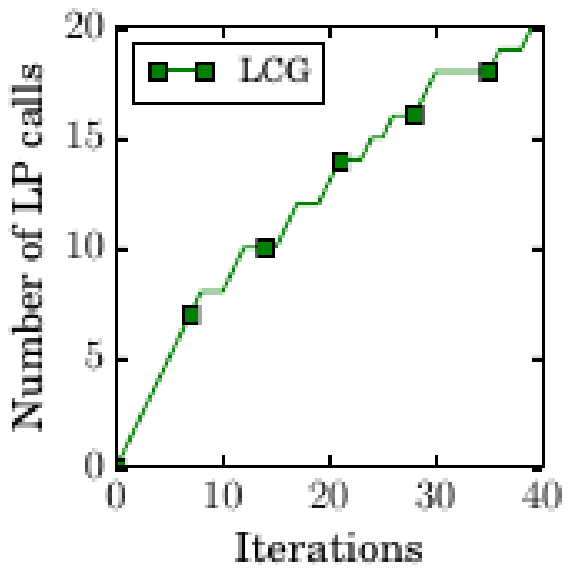}
  &
  \includegraphics[height=0.3\linewidth]{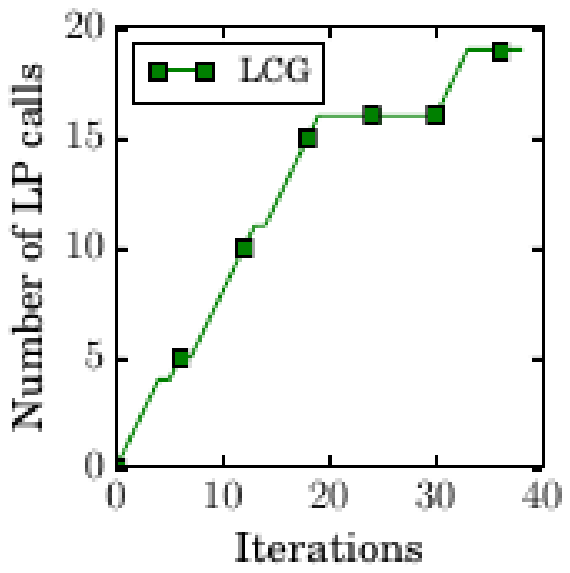}
  \\
  cache hit rate: \(48.72\%\)
  &
  cache hit rate: \(50.00\%\)
  \end{tabular}
  \caption{\label{fig:netgen-large-vanilla} LCG vs. CG on large
  netgen instances \texttt{netgen 16a} (left) and
  \texttt{netgen 16b} (right) with quadratic objective functions. In both cases
  the difference in function value between the two versions of the algorithm is large.
  In the dual bound the performance of the lazy version is multiple orders of
  magnitude better than the performance of the non-lazy counterpart. The cache
  hit rates for these two instances are lower due to the high dimension of the polytope.
  }
\end{figure*}

\begin{figure*}
  \centering
  \small
  \begin{tabular}{*{2}{c}}
  \includegraphics[height=0.35\linewidth]{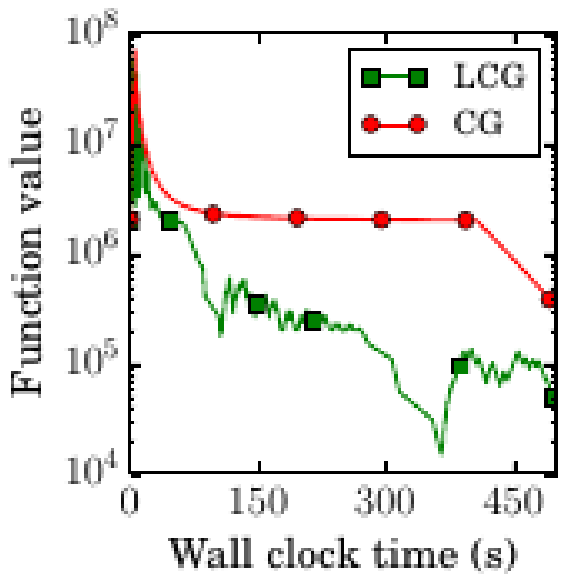}
  &
  \includegraphics[height=0.35\linewidth]{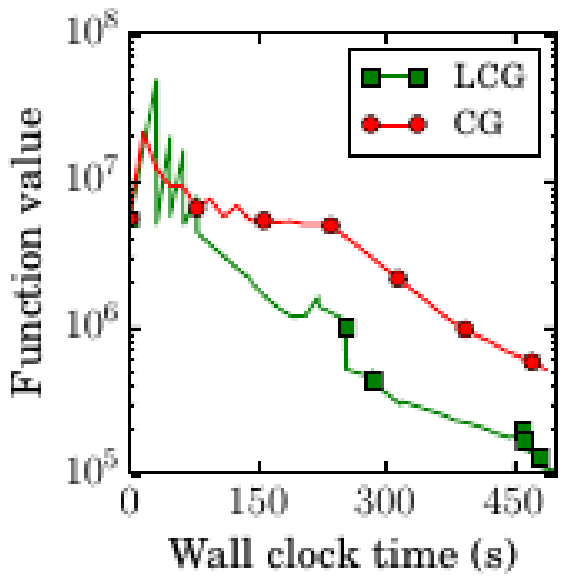}
  \\
  \includegraphics[height=0.35\linewidth]{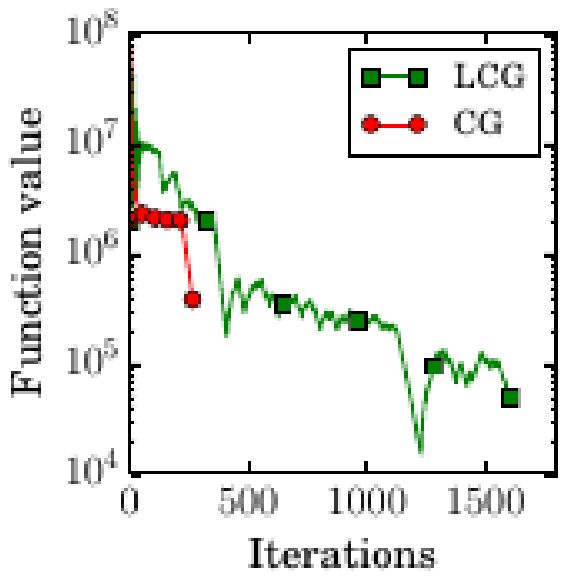}
  &
  \includegraphics[height=0.35\linewidth]{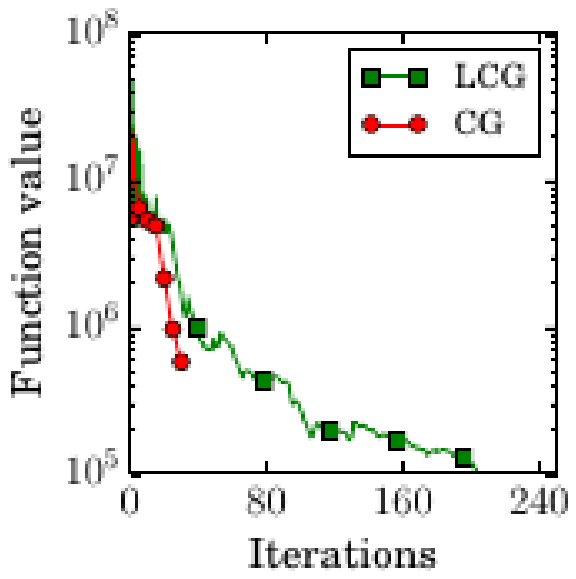}
  \\
  \includegraphics[height=0.35\linewidth]{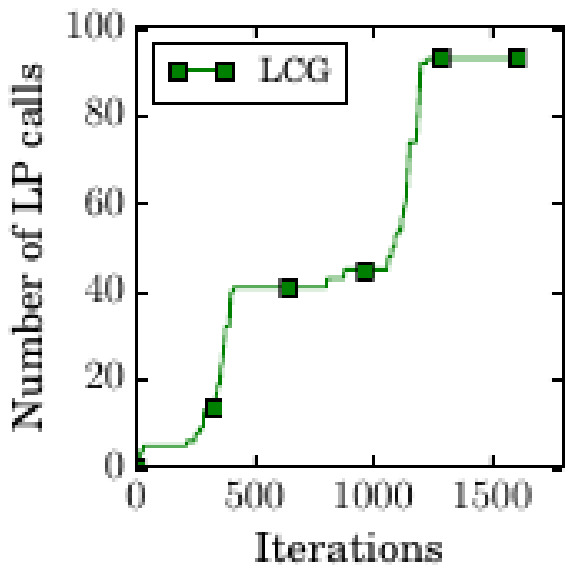}
  &
  \includegraphics[height=0.35\linewidth]{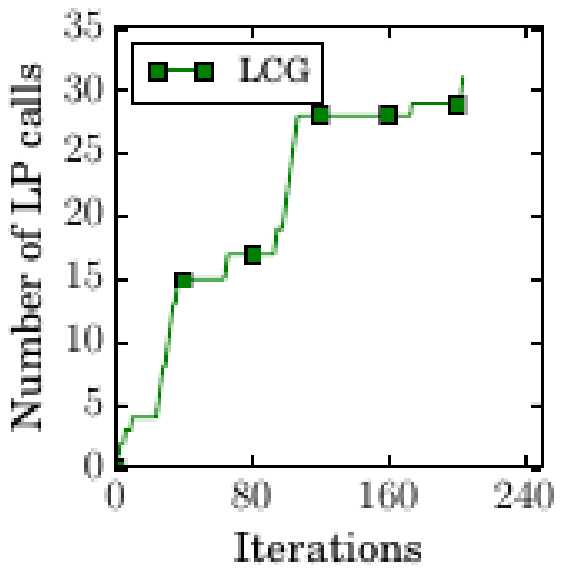}
  \\
  cache hit rate: \(94.24\%\)
  &
  cache hit rate: \(84.80\%\)
  \end{tabular}
  \caption{\label{fig:matrix-completion-small-vanilla}
  LCG vs. CG on two matrix completion instances. We solve the problem as given in
  Equation~\eqref{eq:matrix_completion} with the paramters \(n=3000\), \(m=1000\), \(r=10\) and \(R=30000\)
  for the left instance and \(n=10000\), \(m=100\), \(r=10\) and \(R=10000\) for the right
  instance. In both cases the lazy version is slower in interations, however significantly faster
  in wall clock time.
  }
\end{figure*}

\begin{figure*}
  \centering
  \small
  \begin{tabular}{*{2}{c}}
  \includegraphics[height=0.35\linewidth]{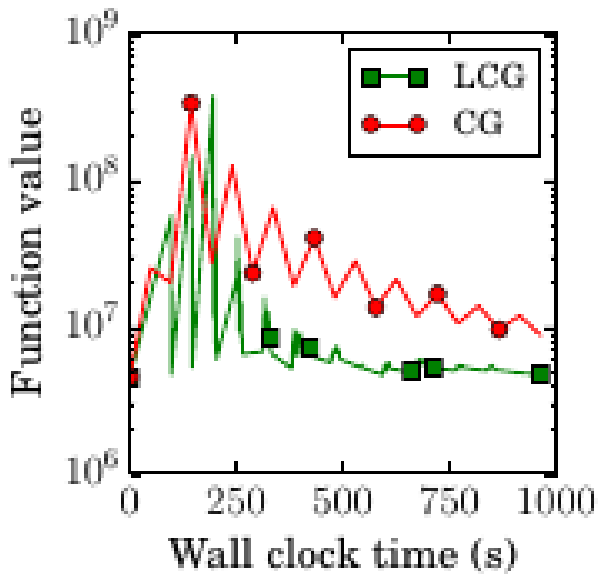}
  &
  \includegraphics[height=0.35\linewidth]{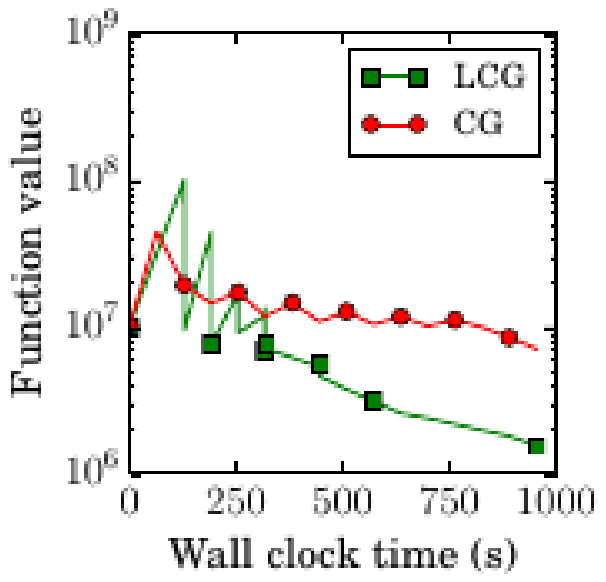}
  \\
  \includegraphics[height=0.35\linewidth]{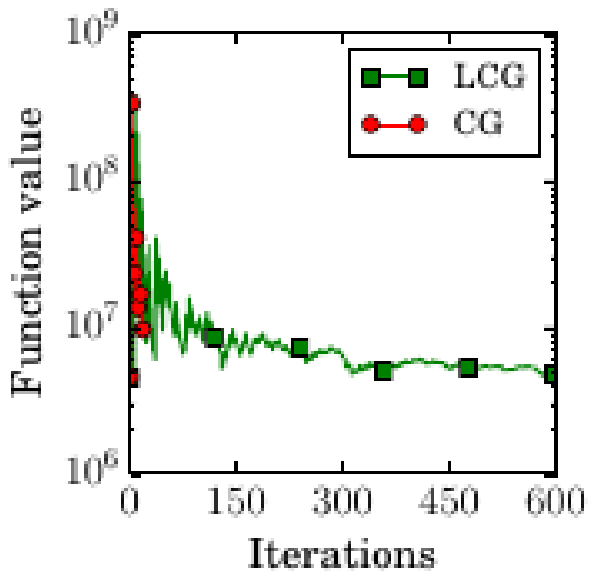}
  &
  \includegraphics[height=0.35\linewidth]{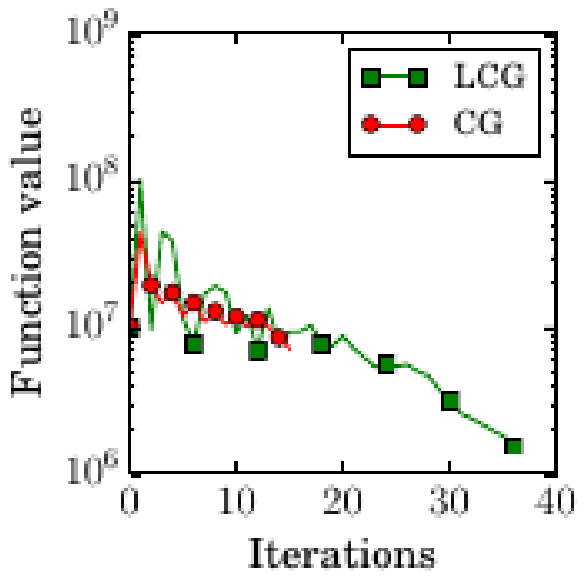}
  \\
  \includegraphics[height=0.35\linewidth]{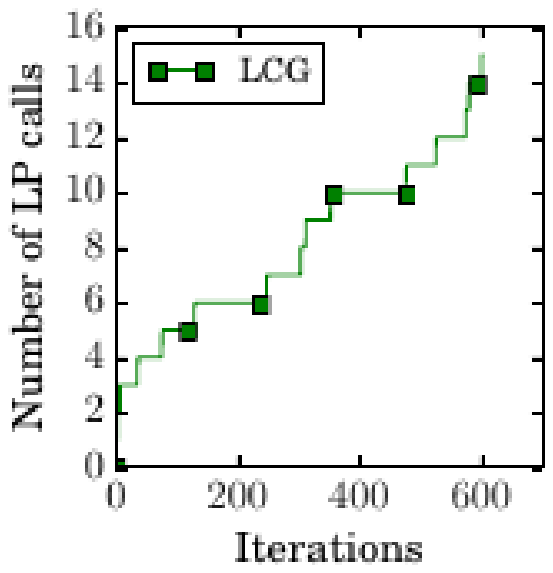}
  &
  \includegraphics[height=0.35\linewidth]{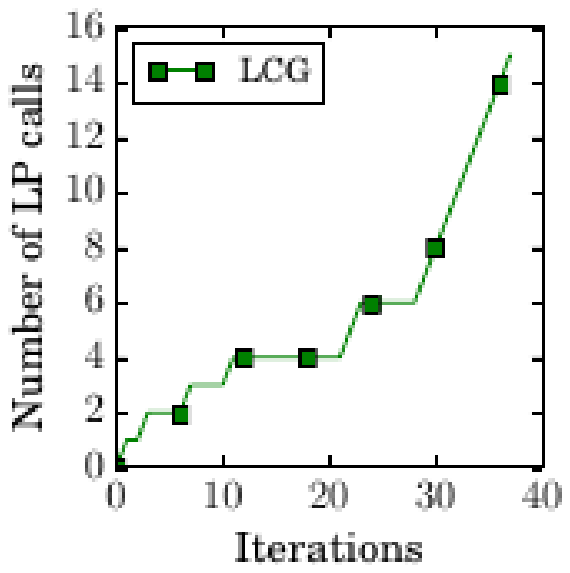}
  \\
  cache hit rate: \(97.50\%\)
  &
  cache hit rate: \(59.46\%\)
  \end{tabular}
  \caption{\label{fig:matrix-completion-large-vanilla}
  LCG vs. CG on two more matrix completion instances. The parameters for
  Equation~\eqref{eq:matrix_completion} are given by \(n=5000\), \(m=4000\), \(r=10\) and \(R=50000\)
  for the left instance and \(n=100\), \(m=20000\), \(r=10\) and \(R=15000\) for the right
  instance. In both of these cases the performance of the lazy and the non-lazy version
  are comparable in interations, however in wall clock time the lazy version reaches lower function
  values faster.
  }
\end{figure*}

\begin{figure*}
  \centering
  \small
  \begin{tabular}{*{2}{c}}
  \includegraphics[height=0.35\linewidth]{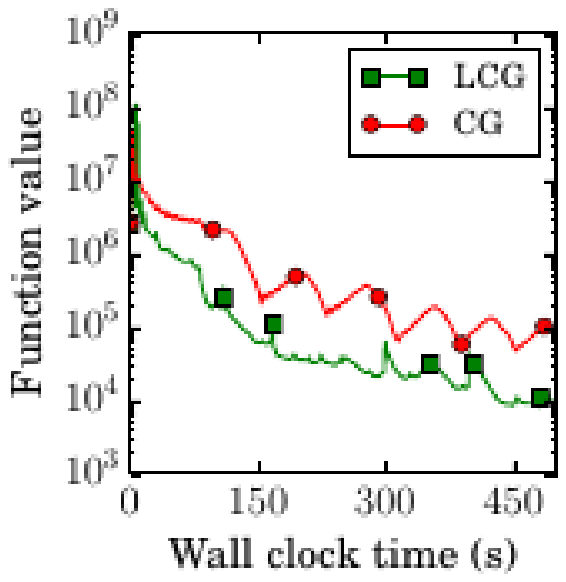}
  &
  \includegraphics[height=0.35\linewidth]{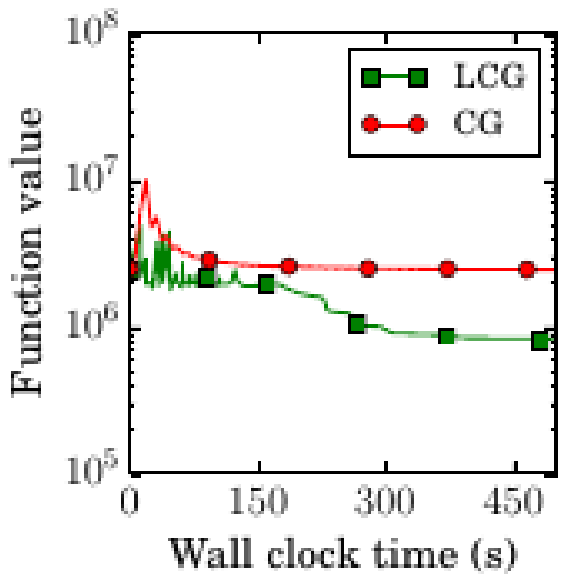}
  \\
  \includegraphics[height=0.35\linewidth]{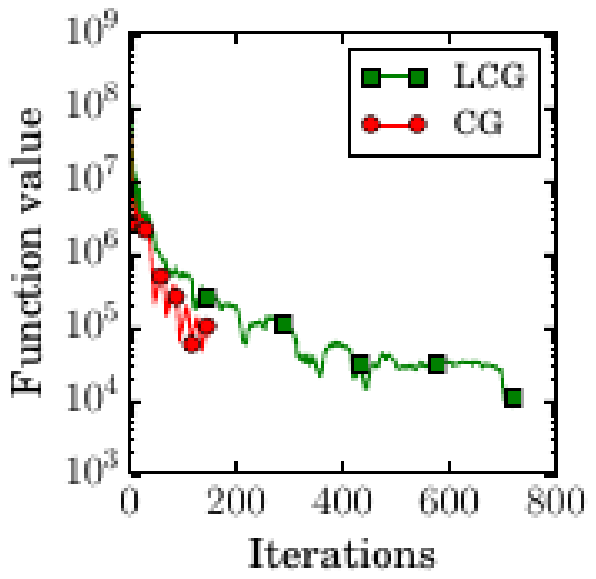}
  &
  \includegraphics[height=0.35\linewidth]{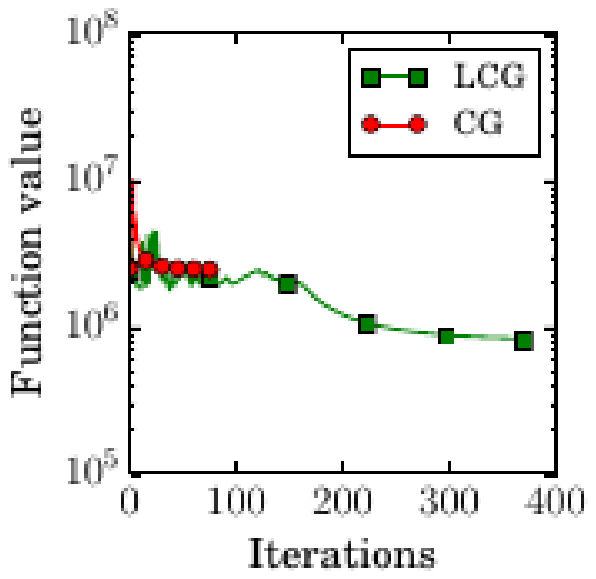}
  \\
  \includegraphics[height=0.35\linewidth]{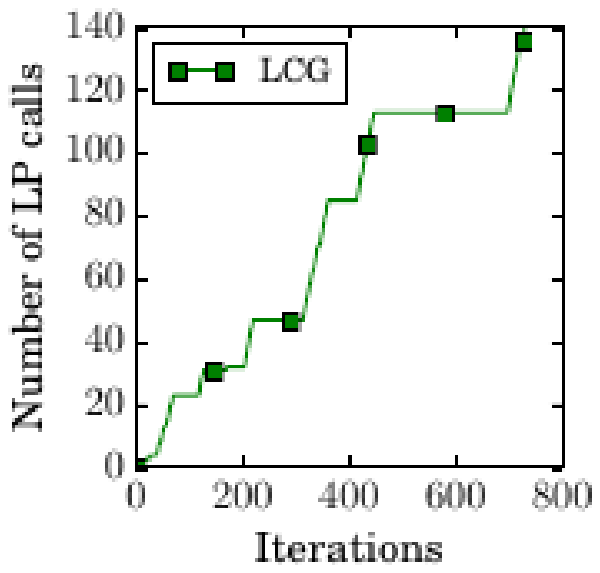}
  &
  \includegraphics[height=0.35\linewidth]{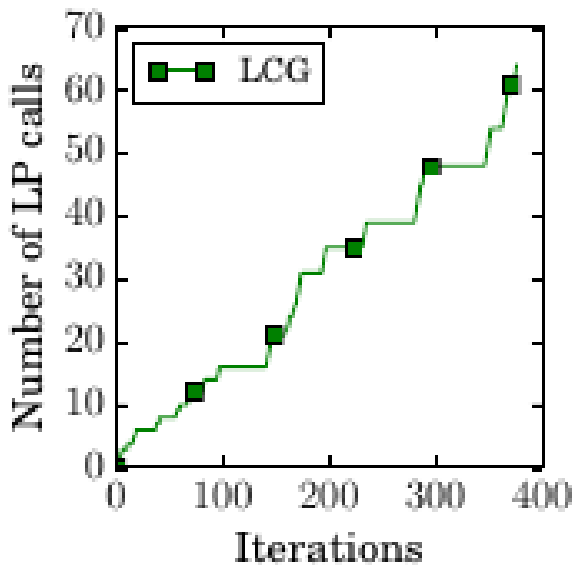}
  \\
  cache hit rate: \(80.80\%\)
  &
  cache hit rate: \(82.98\%\)
  \end{tabular}
  \caption{\label{fig:matrix-completion-alt-1-vanilla}
  LCG vs. CG on our fifth and sixth instance of the matrix completion problem. The parameters are
  \(n=5000\), \(m=100\), \(r=10\) and \(R=15000\)
  for the left instance and \(n=3000\), \(m=2000\), \(r=10\) and \(R=10000\) for the right
  instance. The behavior is very similar to Figure~\ref{fig:matrix-completion-large-vanilla}.
  similar performance over iterations however advantages for the lazy version in wall clock
  time.
  }
\end{figure*}

\begin{figure*}
  \centering
  \small
  \begin{tabular}{*{2}{c}}
  \includegraphics[height=0.35\linewidth]{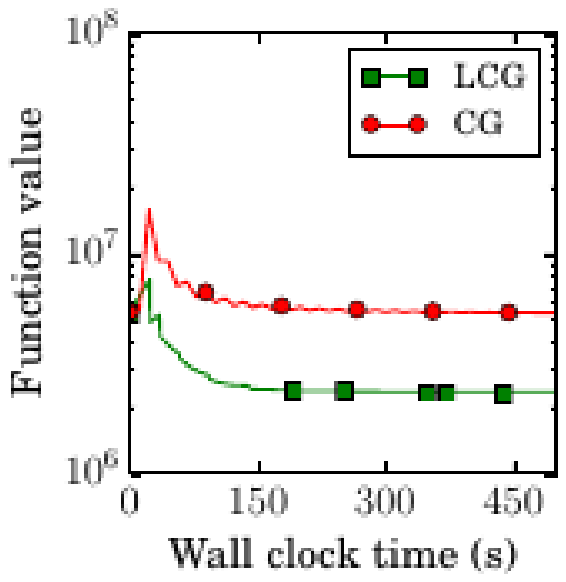}
  &
  \includegraphics[height=0.35\linewidth]{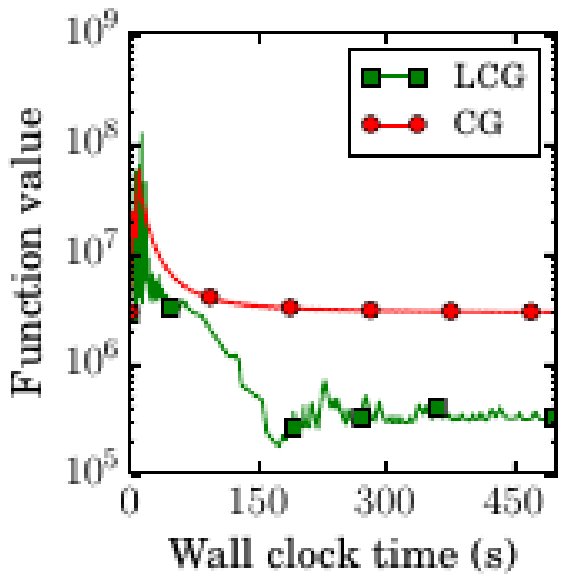}
  \\
  \includegraphics[height=0.35\linewidth]{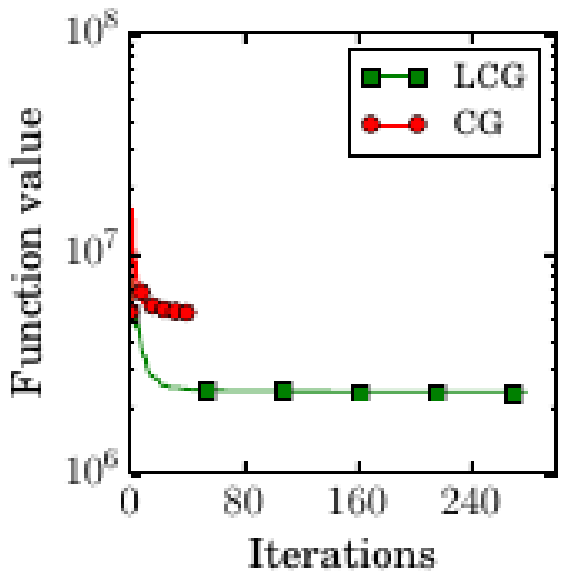}
  &
  \includegraphics[height=0.35\linewidth]{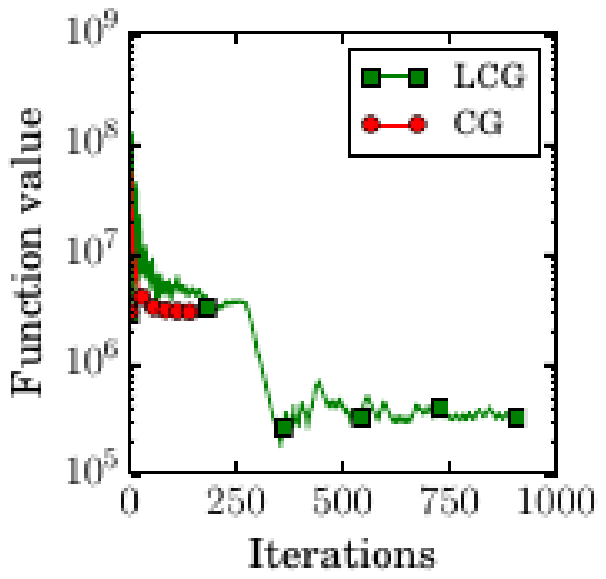}
  \\
  \includegraphics[height=0.35\linewidth]{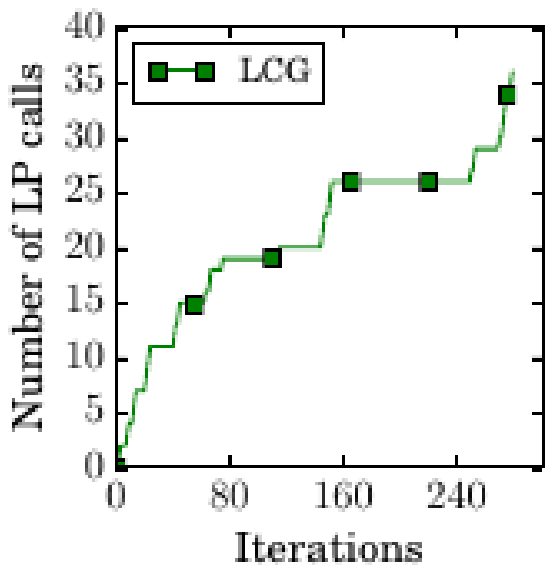}
  &
  \includegraphics[height=0.35\linewidth]{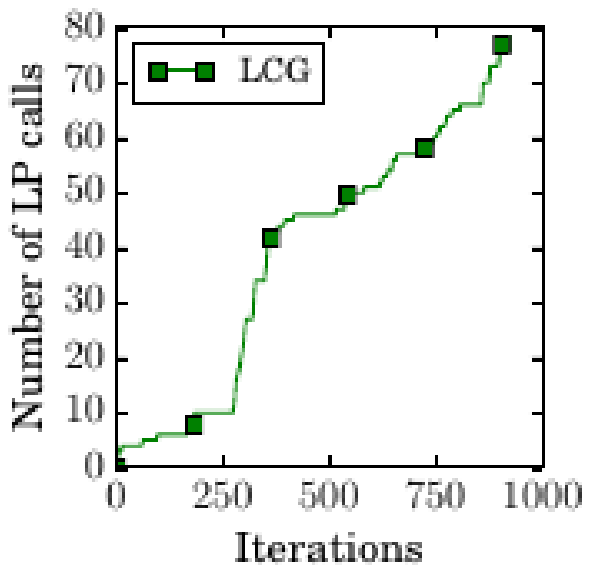}
  \\
  cache hit rate: \(87.10\%\)
  &
  cache hit rate: \(91.55\%\)
  \end{tabular}
  \caption{\label{fig:matrix-completion-alt-2-vanilla}
  LCG vs. CG on the final two matrix completion instances. The parameters are
  \(n=10000\), \(m=1000\), \(r=10\) and \(R=1-000\)
  for the left instance and \(n=5000\), \(m=1000\), \(r=10\) and \(R=30000\) for the right
  instance. On the left in both measures, instances and wall clock time, the lazy
  version performs better than the non-lazy counterpart, due to a suboptimal direction at the
  beginning with a fairly large step size in the non-lazy version.
  }
\end{figure*}

\begin{figure*}
  \centering
  \small
  \begin{tabular}{*{2}{c}}
  \includegraphics[height=0.35\linewidth]{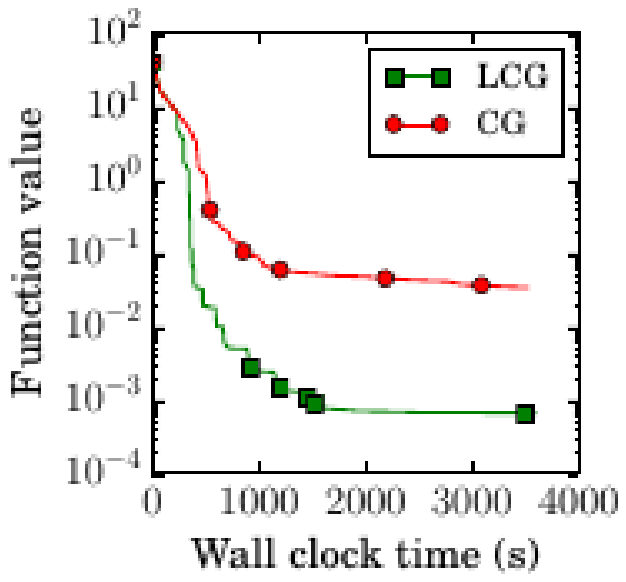}
  &
  \includegraphics[height=0.35\linewidth]{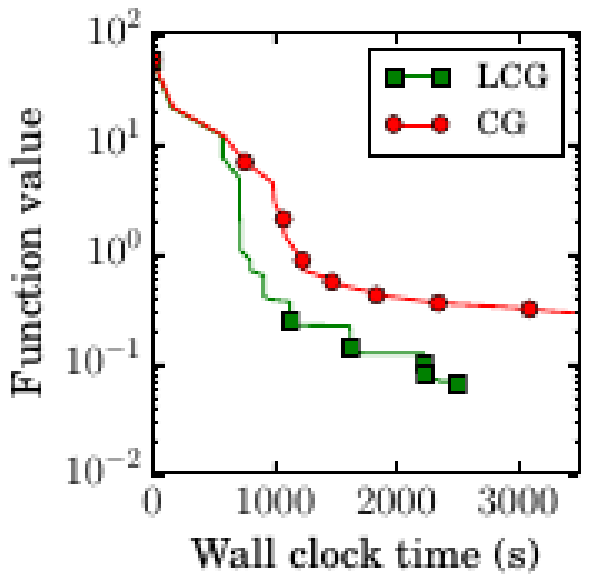}
  \\
  \includegraphics[height=0.35\linewidth]{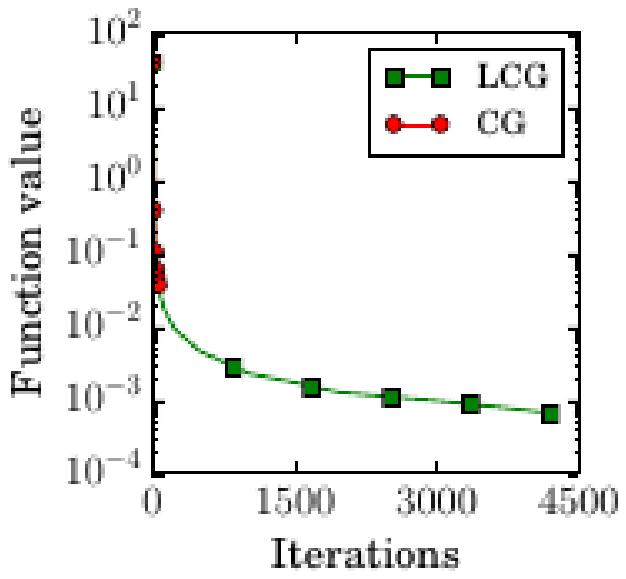}
  &
  \includegraphics[height=0.35\linewidth]{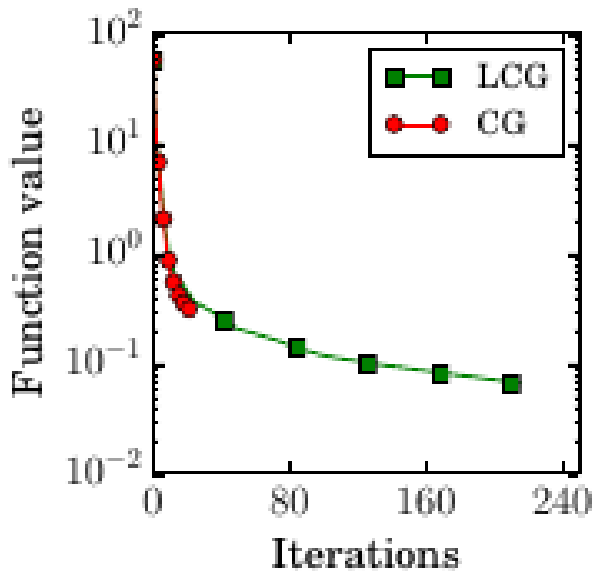}
  \\
  \includegraphics[height=0.35\linewidth]{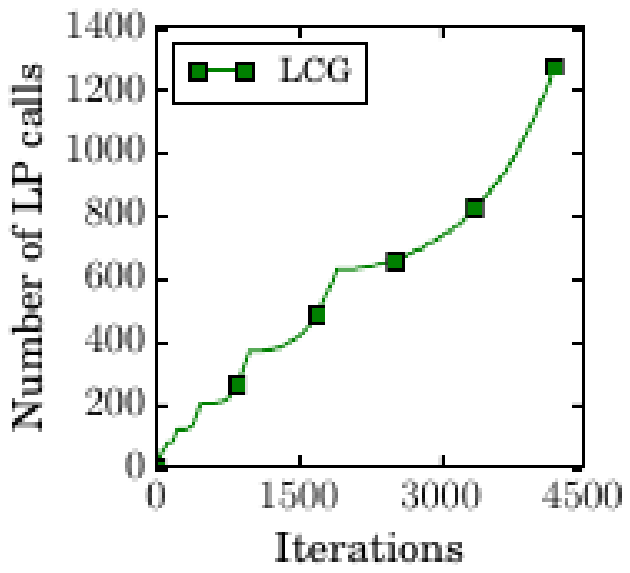}
  &
  \includegraphics[height=0.35\linewidth]{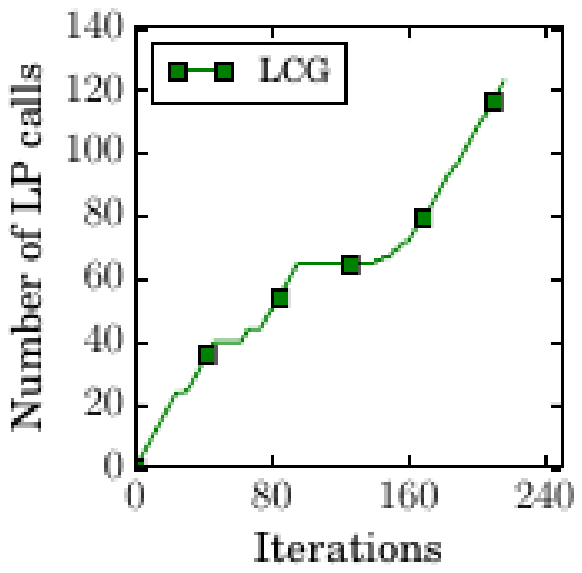}
  \\
  cache hit rate: \(69.46\%\)
  &
  cache hit rate: \(43.06\%\)
  \end{tabular}
  \caption{\label{fig:tsp-vanilla} LCG vs. CG on structured regression problems with
    feasible regions being a TSP polytope over \(11\) nodes
    (left) and \(12\) nodes (right). In both cases LCG is
    significantly faster in wall-clock time.
  }
\end{figure*}

\begin{figure*}
  \centering
  \small
  \begin{tabular}{*{2}{c}}
  \includegraphics[height=0.35\linewidth]{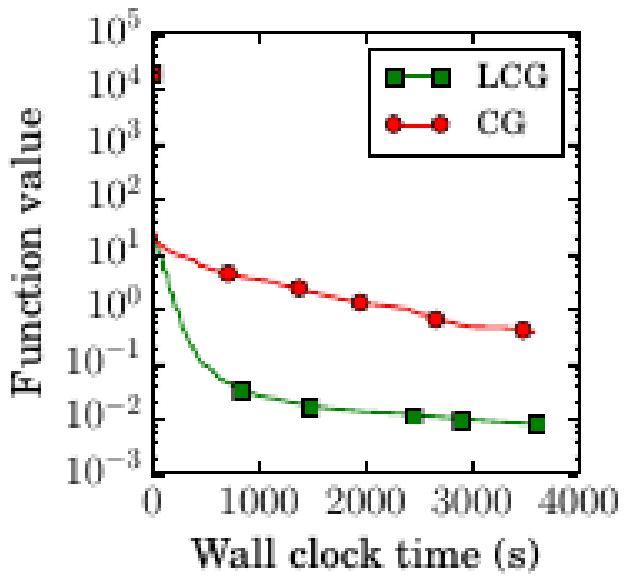}
  &
  \includegraphics[height=0.35\linewidth]{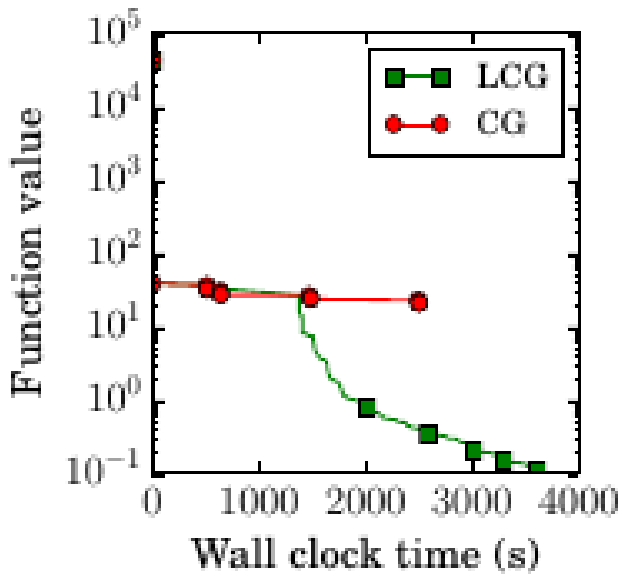}
  \\
  \includegraphics[height=0.35\linewidth]{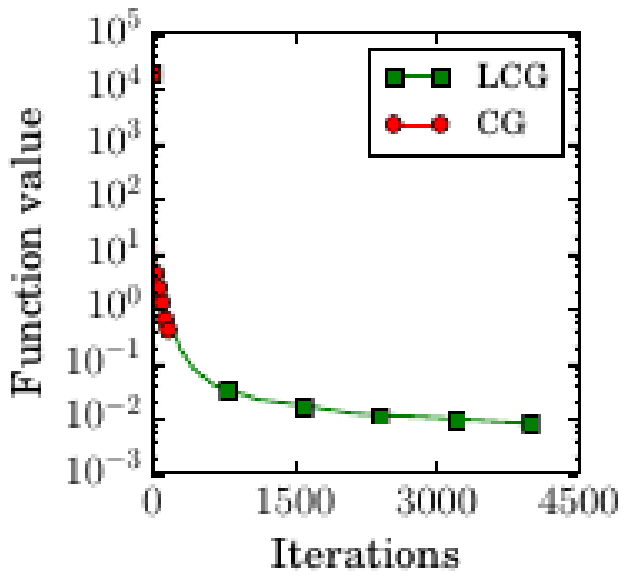}
  &
  \includegraphics[height=0.35\linewidth]{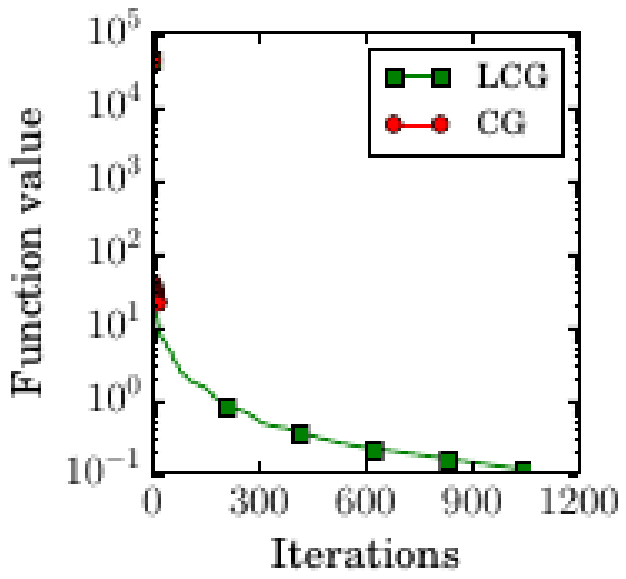}
  \\
  \includegraphics[height=0.35\linewidth]{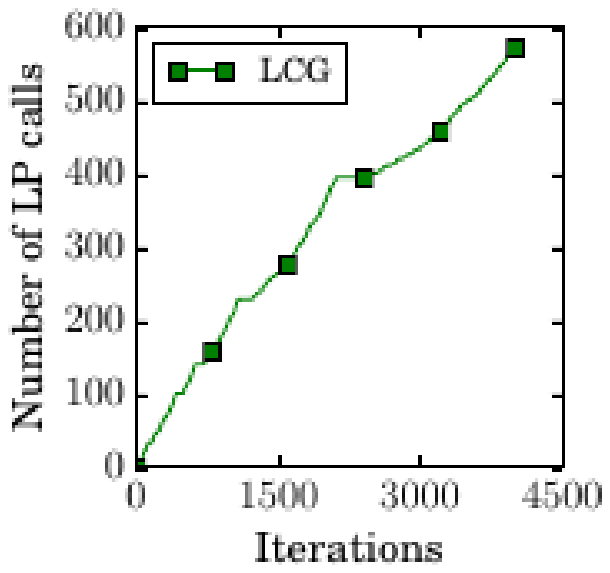}
  &
  \includegraphics[height=0.35\linewidth]{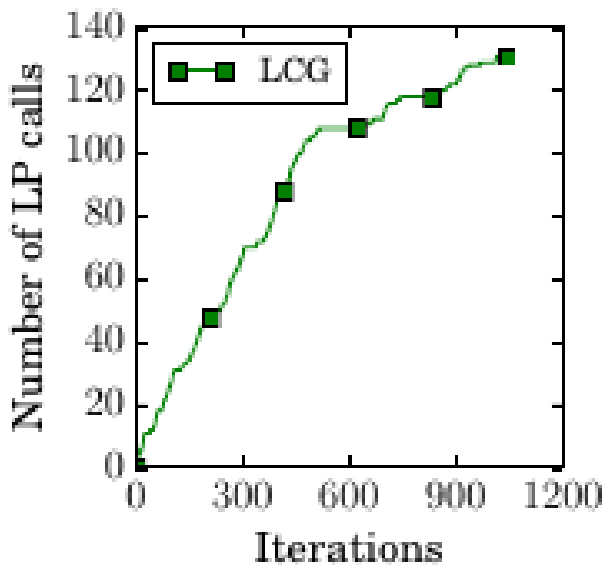}
  \\
  cache hit rate: \(85.61\%\)
  &
  cache hit rate: \(87.48\%\)
  \end{tabular}
  \caption{\label{fig:maxcut-vanilla}
    LCG vs. CG on structured regression instances using cut polytopes over a graph on \(23\) nodes (left)
    and over \(28\) nodes (right) as feasible region. In both instances LCG performs
    significantly better than CG.
  }
\end{figure*}

\begin{figure*}
  \centering
  \small
  \begin{tabular}{*{2}{c}}
  \includegraphics[height=0.35\linewidth]{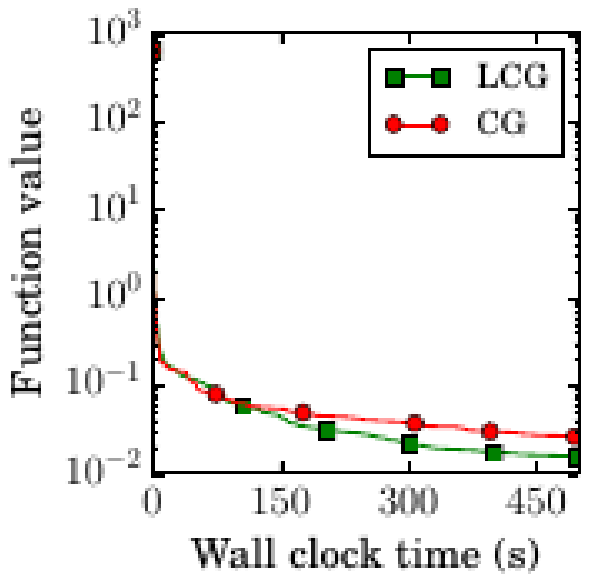}
  &
  \includegraphics[height=0.35\linewidth]{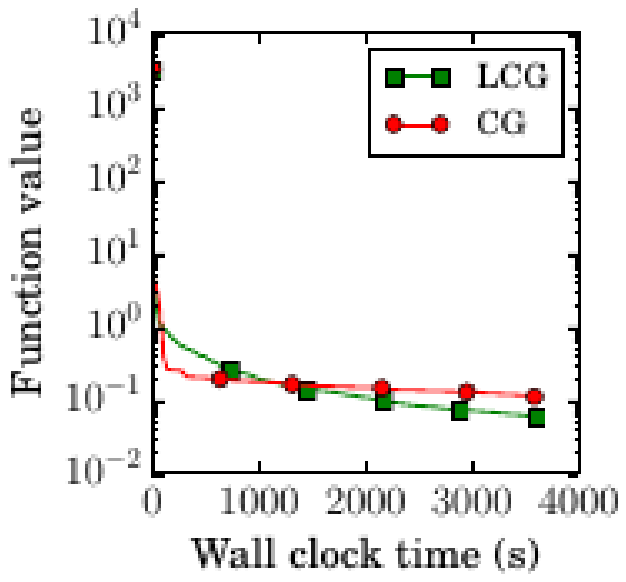}
  \\
  \includegraphics[height=0.35\linewidth]{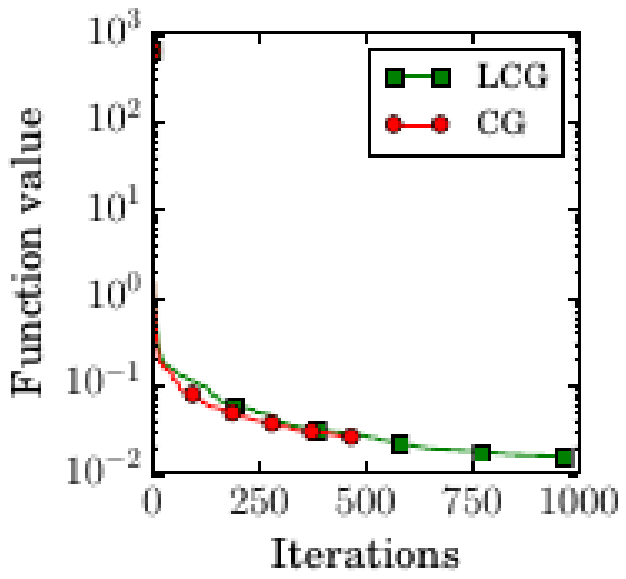}
  &
  \includegraphics[height=0.35\linewidth]{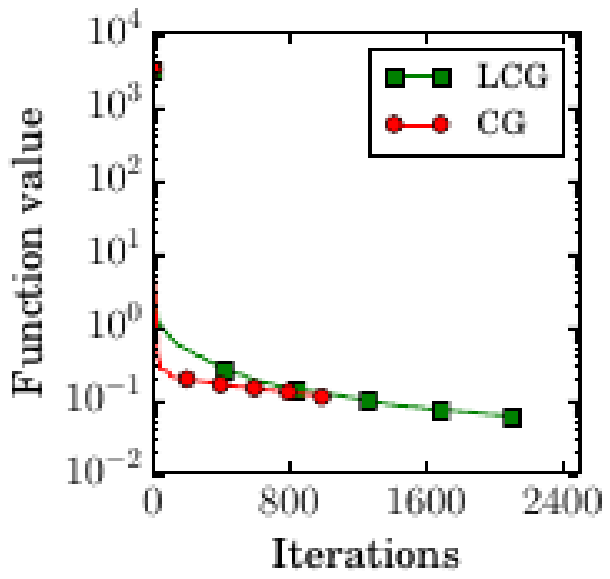}
  \\
  \includegraphics[height=0.35\linewidth]{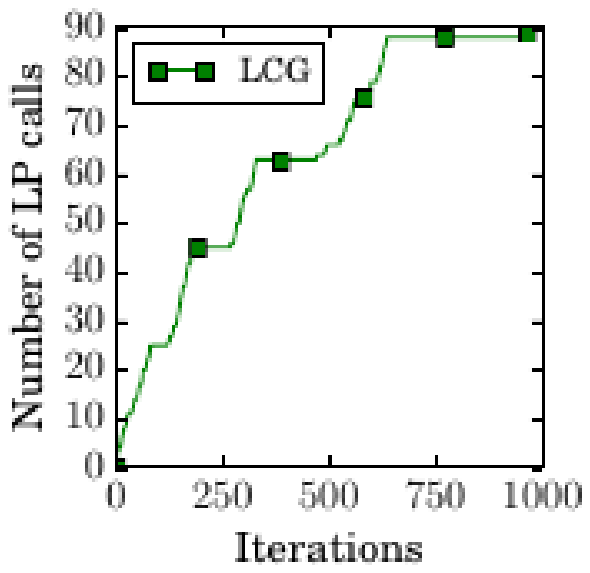}
  &
  \includegraphics[height=0.35\linewidth]{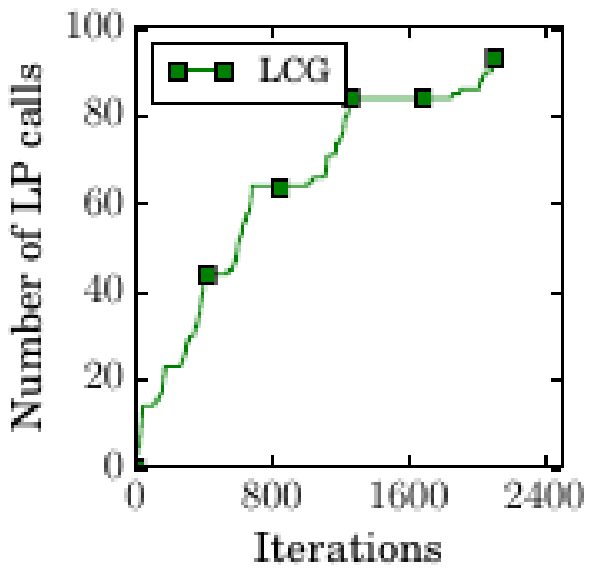}
  \\
  cache hit rate: \(90.83\%\)
  &
  cache hit rate: \(95.59\%\)
  \end{tabular}
  \caption{\label{fig:spanning-tree-vanilla}
    LCG vs. CG on structured regression instances with extended formulation of the
    spanning tree problem on a \(10\)
    node graph on the left and a \(15\)
    node graph on the right.
    }
\end{figure*}

\paragraph{Pairwise Conditional Gradient Algorithm}

As we inherit structural restrictions of PCG
on the feasible region,
the problem repertoire is limited in this case.
We tested the Pairwise Conditional Gradient algorithm on the structured
regression problem
with feasible regions from the MIPLIB instances
\texttt{eil33-2}, \texttt{air04}, \texttt{eilB101}, \texttt{nw04},
\texttt{disctom}, \texttt{m100n500k4r1}
(Figures~\ref{fig:eil33-air04-offline},~\ref{fig:nw04-eilB101-offline}
and~\ref{fig:disctom-m100n500k4r1-offline}).

Again similarly to the vanilla Frank-Wolfe algorihtm,
we observed a significant improvement in wall-clock time
of LPCG compared to CG,
due to the faster iteration of the lazy algorithm.

\begin{figure*}
  \centering
  \small
  \begin{tabular}{*{2}{c}}
    eil33-2, \(4516\) dimensions
    &
    air04, \(8904\) dimensions
    \\
  \includegraphics[height=0.35\linewidth]{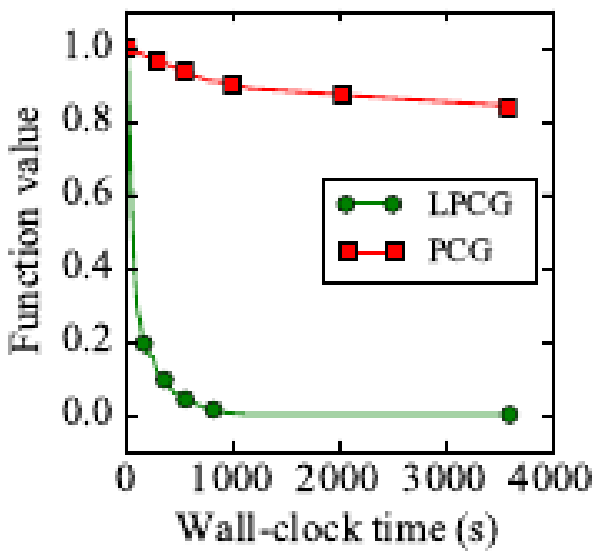}
  &
  \includegraphics[height=0.35\linewidth]{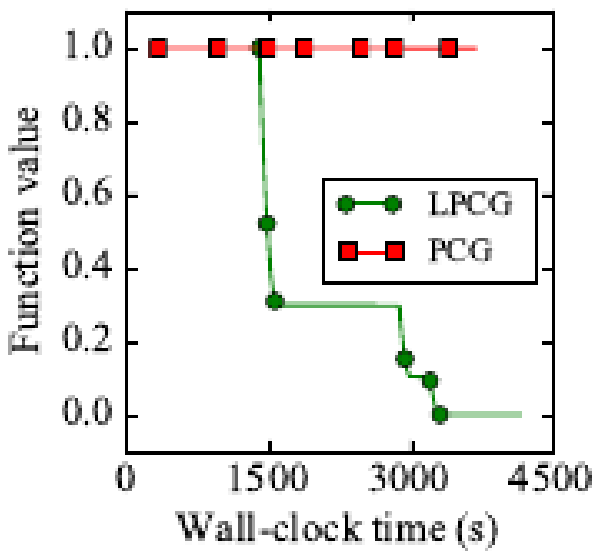}
  \\
  \includegraphics[height=0.35\linewidth]{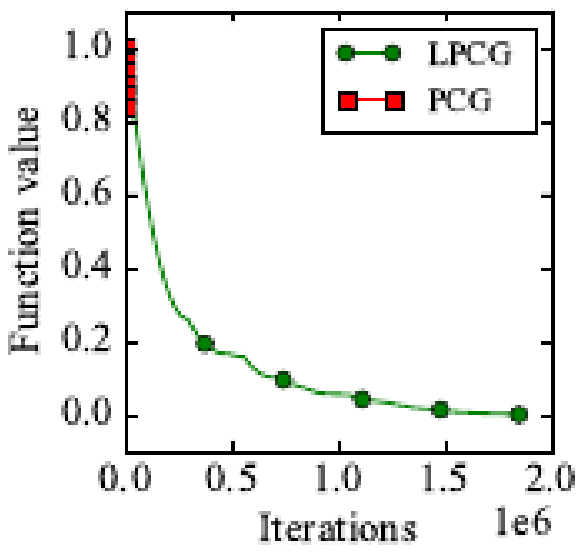}
  &
  \includegraphics[height=0.35\linewidth]{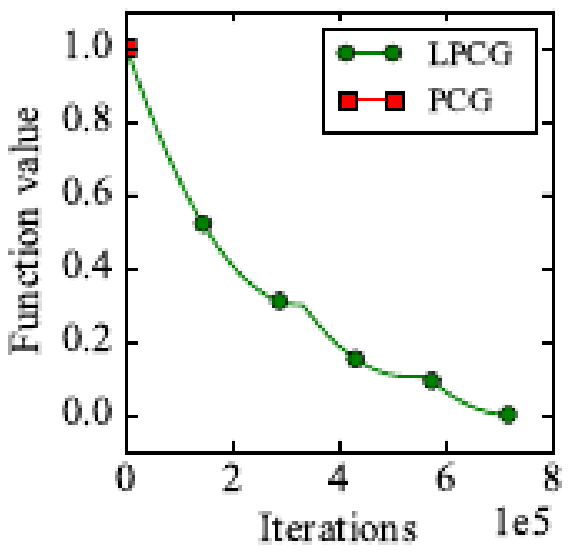}
  \\
  \includegraphics[height=0.35\linewidth]{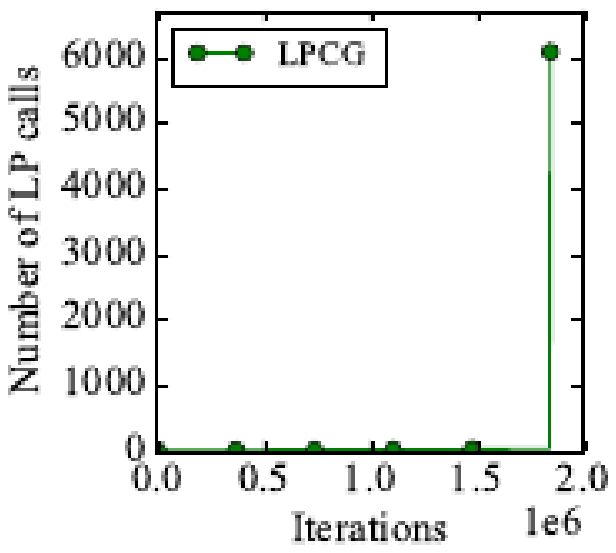}
  &
  \includegraphics[height=0.35\linewidth]{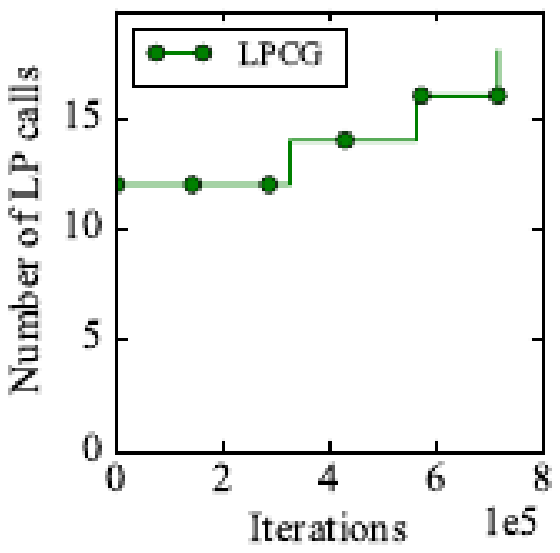}
  \\
  cache hit rate: \(99.8\%\)
  &
  cache hit rate: \(99.999\%\)
  \end{tabular}
  \caption{\label{fig:eil33-air04-offline}
    LPCG vs. PCG on two MIPLIB instances \texttt{eil33-2} and \texttt{air04}.
    LPCG converges very fast, making millions of iterations
    with a relatively few oracle calls, while PCG
    completed only comparably few iterations
    due to the time-consuming oracle calls.
    This clearly illustrates the advantage of lazy methods
    when the cost of linear optimization is non-negligible.
    On the left, when reaching \(\varepsilon\)-optimality,
    LPCG performs many (negative) oracle calls to (re-)prove
    optimality; at that point one might opt for stopping the algorithm. 
    On the right LPCG needed a rather long time
    for the initial bound tigthening of \(\Phi_{0}\),
    before converging significantly faster than PCG.
  }
\end{figure*}
\begin{figure*}
  \centering
  \begin{tabular}{cc}
  eilB101, \(2818\) dimensions
  &
  nw04, \(87482\) dimensions
  \\
  \includegraphics[height=0.35\linewidth]{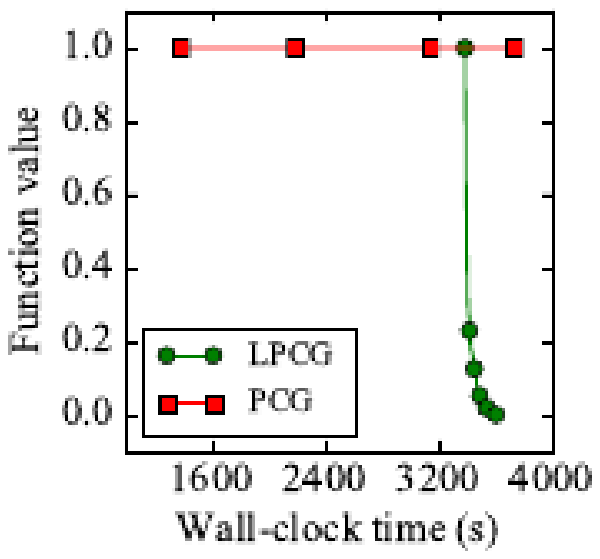}
  &
  \includegraphics[height=0.35\linewidth]{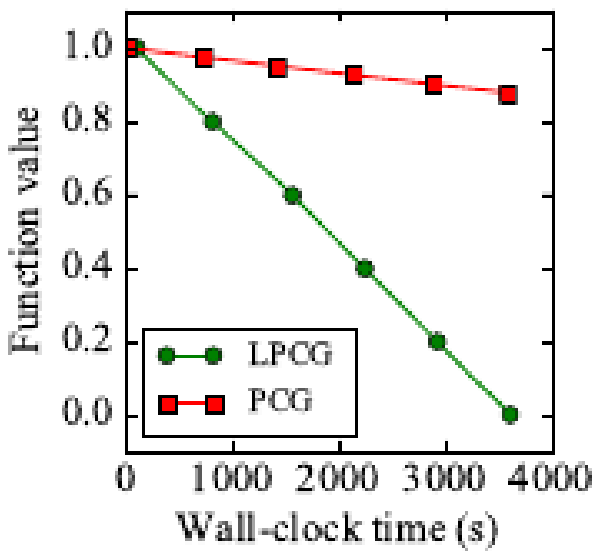}
  \\
  \includegraphics[height=0.35\linewidth]{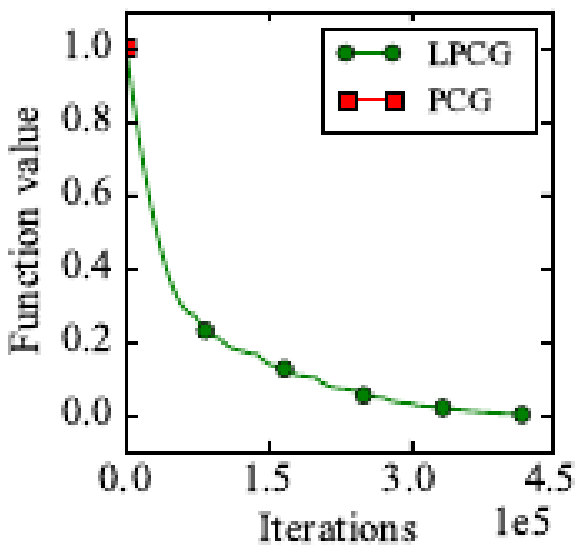}
  &
  \includegraphics[height=0.35\linewidth]{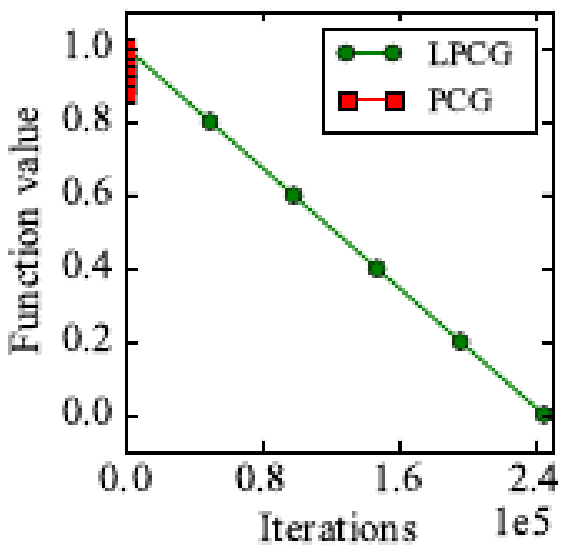}
  \\
  \includegraphics[height=0.35\linewidth]{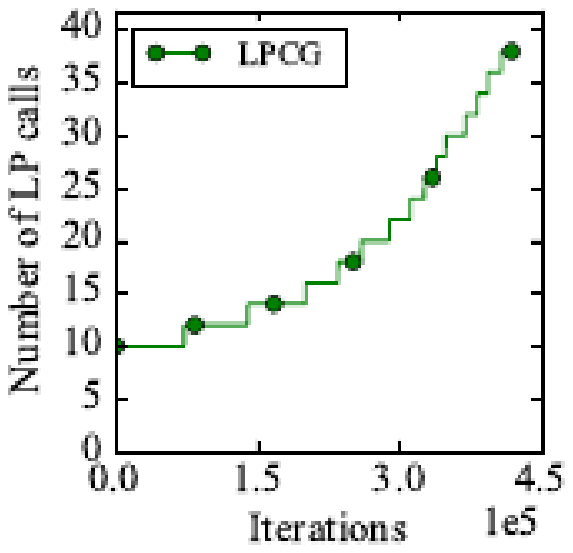}
  &
  \includegraphics[height=0.35\linewidth]{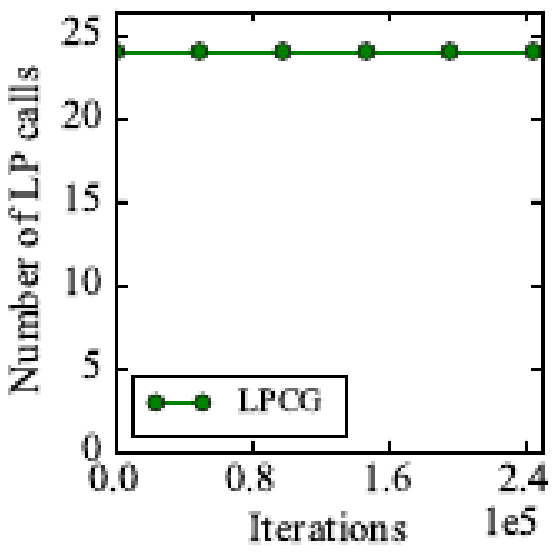}
  \\
  cache hit rate: \(99.995\%\)
  &
  cache hit rate: \(99.995\%\)
  \end{tabular}
  \caption{\label{fig:nw04-eilB101-offline}
    LPCG vs. PCG on MIPLIB instances \texttt{eilB101} and \texttt{nw04}
    with quadratic loss functions.
    For the \texttt{eilB101} instance,
    LPCG spent most of the time tightening \(\Phi_{0}\),
    after which it converged very fast,
    while PCG was unable to complete a single iteration even solving
    the problem only approximately.
    For the \texttt{nw04} instance LPCG needed no more oracle calls
    after an initial phase, while significantly outperforming PCG.
  }
\end{figure*}
\begin{figure*}
  \centering
  \begin{tabular}{cc}
  disctom, \(10000\) dimensions
  &
  m100n500k4r1, \(600\) dimensions
  \\
  \includegraphics[height=0.35\linewidth]{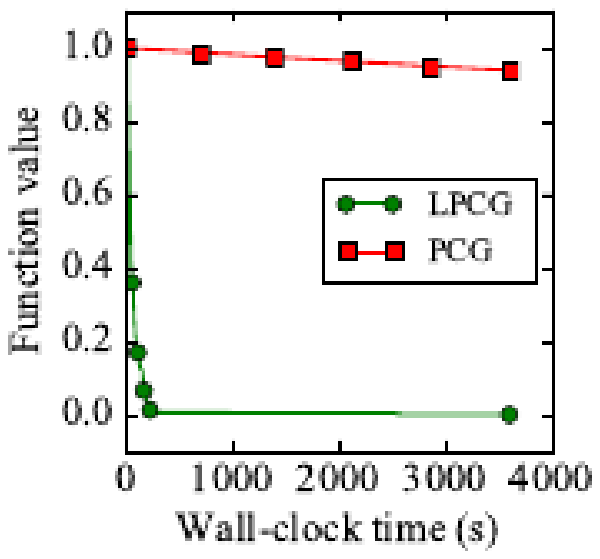}
  &
  \includegraphics[height=0.35\linewidth]{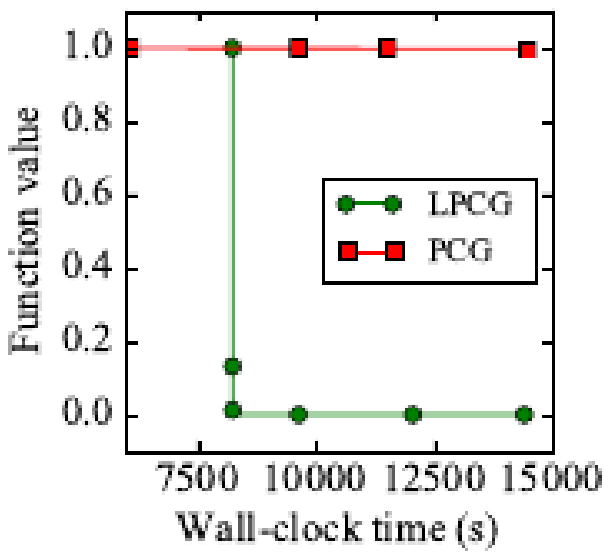}
  \\
  \includegraphics[height=0.35\linewidth]{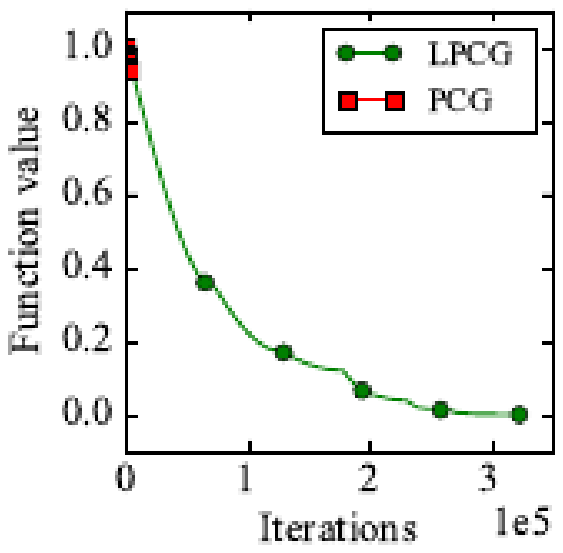}
  &
  \includegraphics[height=0.35\linewidth]{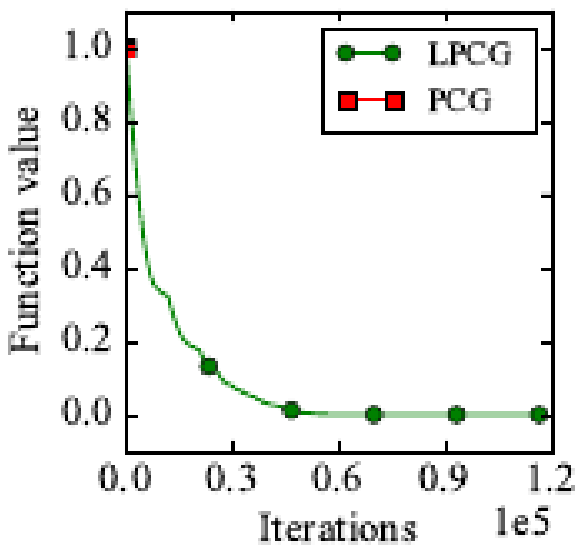}
  \\
  \includegraphics[height=0.35\linewidth]{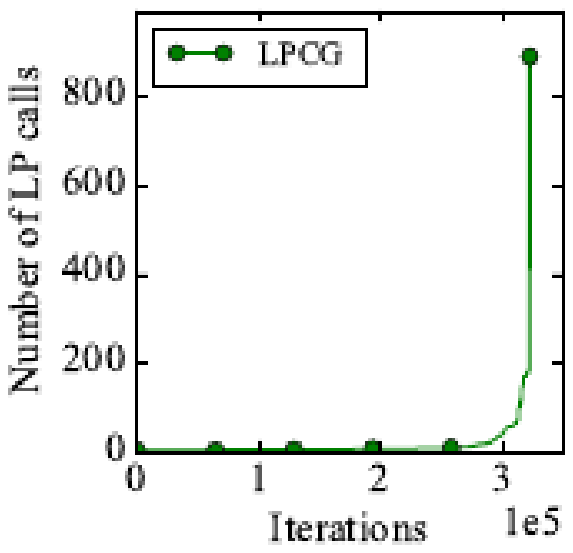}
  &
  \includegraphics[height=0.35\linewidth]{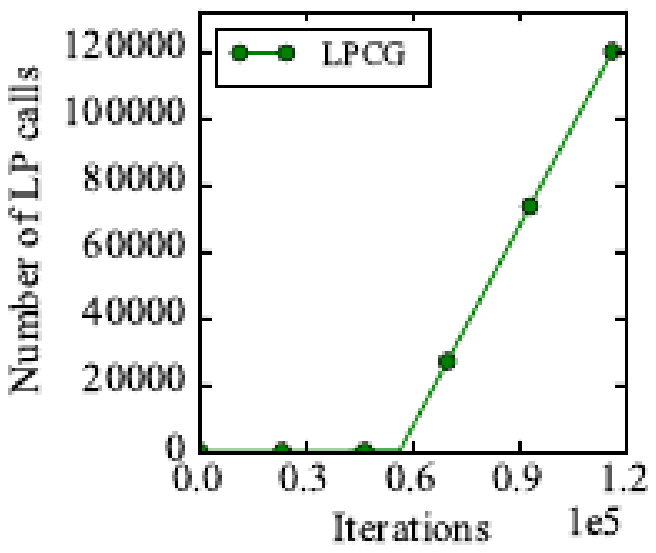}
  \\
  cache hit rate: \(99.9\%\)
  &
  cache hit rate: \(48.4\%\)
  \end{tabular}
  \caption{\label{fig:disctom-m100n500k4r1-offline}
    LPCG vs. PCG on MIPLIB instances \texttt{disctom} and \texttt{m100n500k4r1}.
    After very fast convergence, there is a huge increase
    in the number of oracle calls for the lazy algorithm LPCG
    due to reaching \(\varepsilon\)-optimality as explained before.
    On the right the initial bound tightening for \(\Phi_{0}\)
    took a considerable amount of time but then convergence is almost instantaneous.
  }
\end{figure*}

\subsubsection{Online Results}
\label{sec:onlineResults}

Additionally to the quadratic objective functions above we
tested the online version on random linear functions
\(c x + b\) with \(c \in [-1, +1]^{n}\) and \(b \in [0, 1].\)
For online algorithms,
each experiment used a random sequence of
\(100\) different random loss functions.
In every figure
the left column uses linear loss functions,
while the right one uses quadratic loss functions
over the same polytope. As customary, we did not use line search here
but used the respective prescribed step sizes. 

As an instance of the structured regression problem
we used the flow-based formulation
for Hamiltonian cycles in graphs,
i.e., the traveling salesman problem (TSP)
for graphs with \(11\)
and \(16\)
nodes (Figures~\ref{fig:tspSmall} and~\ref{fig:tspLarge}).
\begin{figure*}[ht] 
  \centering
  \begin{tabular}{cc}
  \includegraphics[height=0.35\linewidth]{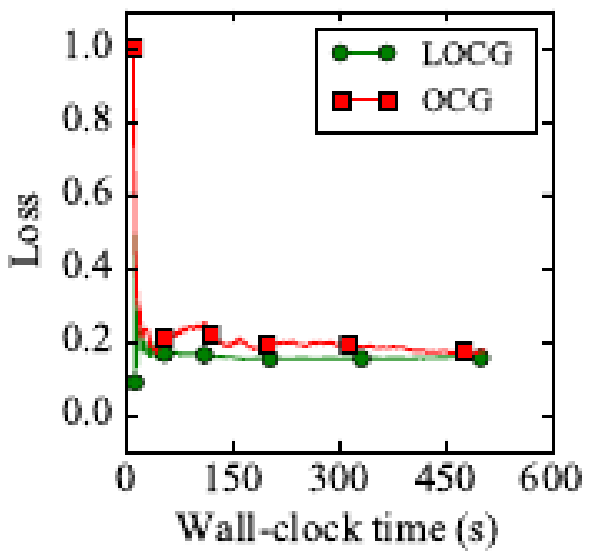}
  &
  \includegraphics[height=0.35\linewidth]{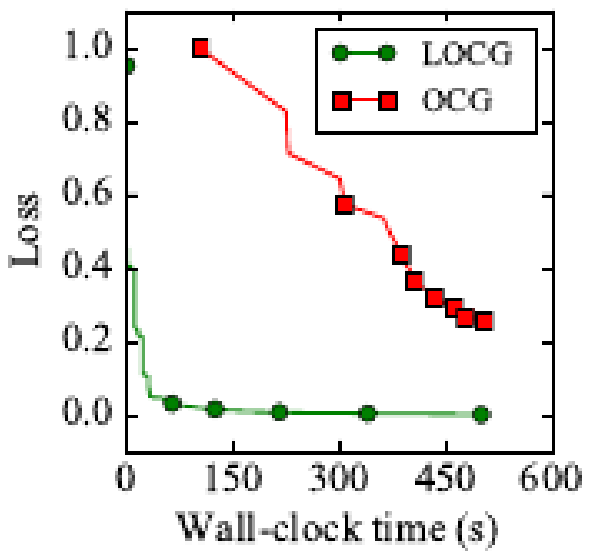}
  \\
  \includegraphics[height=0.35\linewidth]{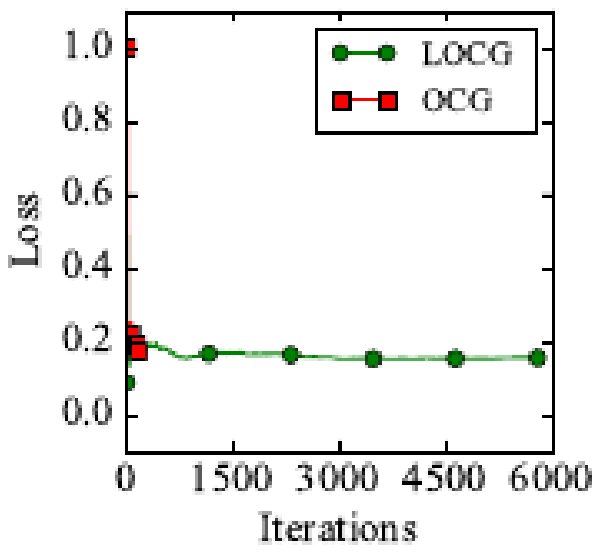}
  &
  \includegraphics[height=0.35\linewidth]{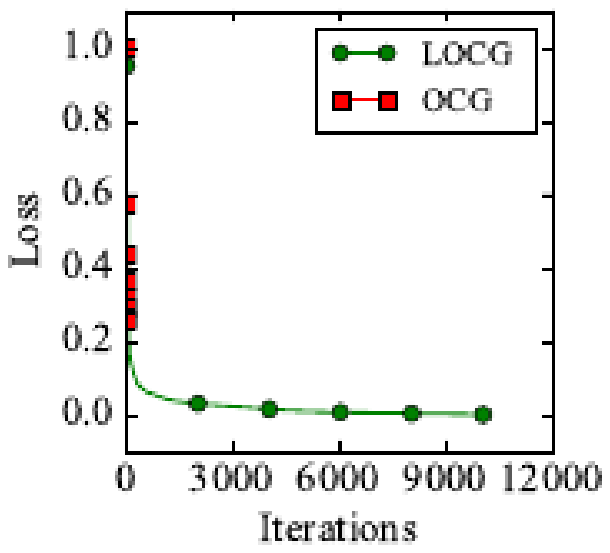}
  \\
  \includegraphics[height=0.35\linewidth]{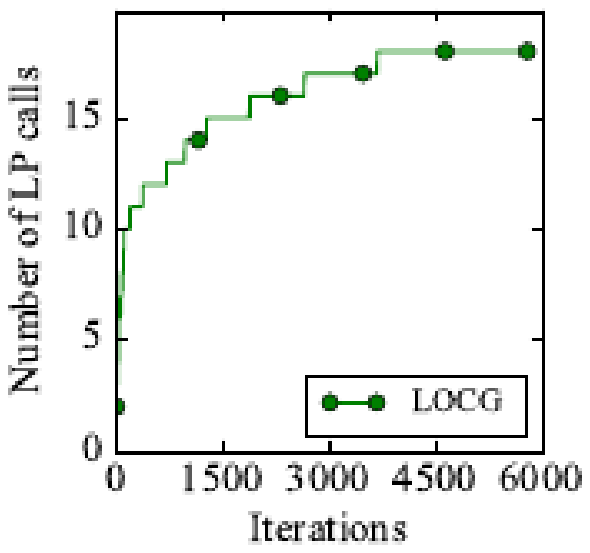}
  &
  \includegraphics[height=0.35\linewidth]{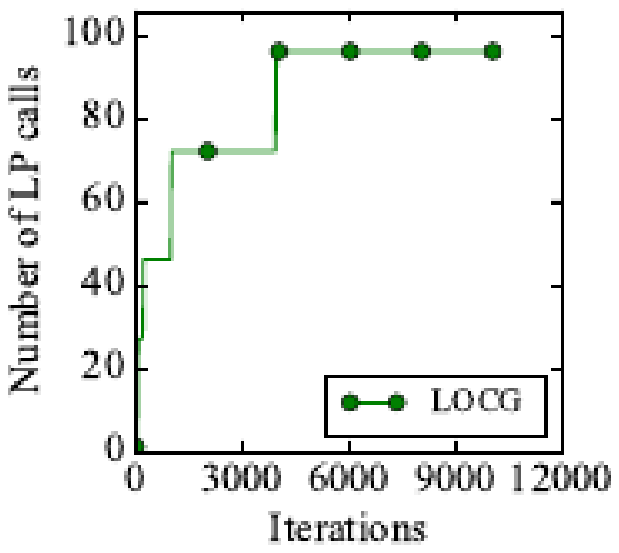}
  \\
  cache hit rate: \(99.7\%\)
  &
  cache hit rate: \(99.0\%\)
  \end{tabular}
  \caption{\label{fig:tspSmall} LOCG vs. OCG over TSP polytope
    for a graph with \(11\) nodes as feasible region and
    with a \(500\) seconds time limit.
    OCG completed only a few iterations,
    resulting in a several times larger final loss
    for quadratic loss functions (right).
    Notice that with time LOCG needed fewer and fewer LP calls.
  }
\end{figure*}
\begin{figure*}
  \centering
  \begin{tabular}{cc}
  \includegraphics[height=0.35\linewidth]{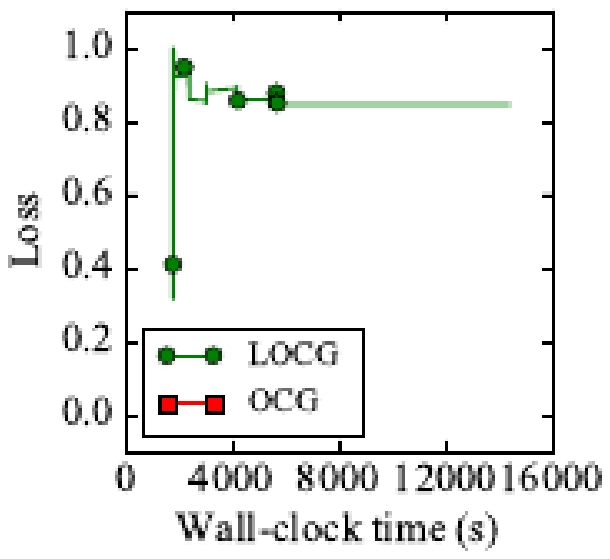}
  &
  \includegraphics[height=0.35\linewidth]{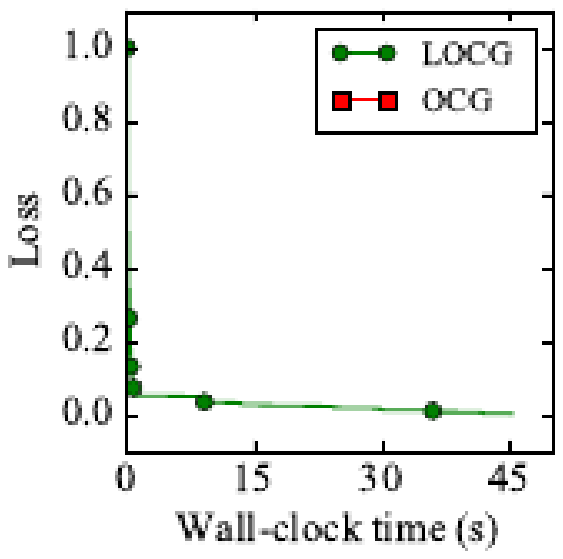}
  \\
  \includegraphics[height=0.35\linewidth]{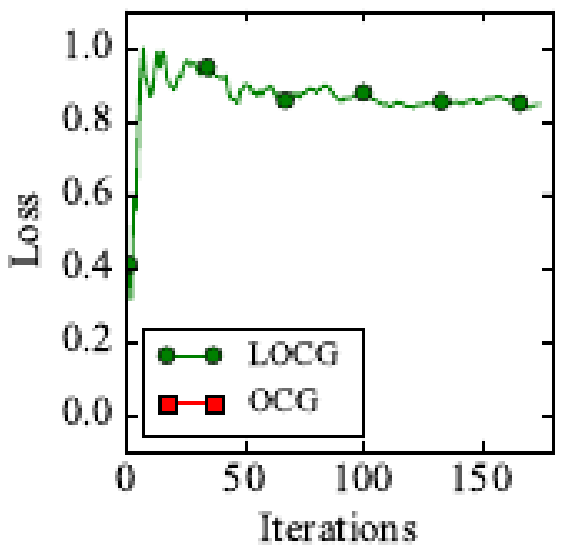}
  &
  \includegraphics[height=0.35\linewidth]{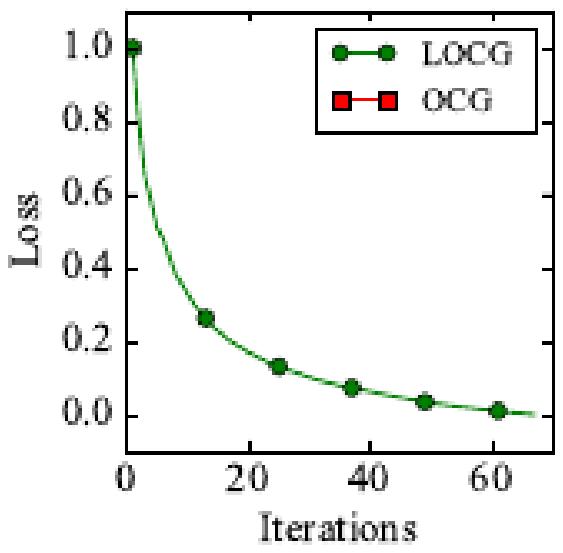}
  \\
  \includegraphics[height=0.35\linewidth]{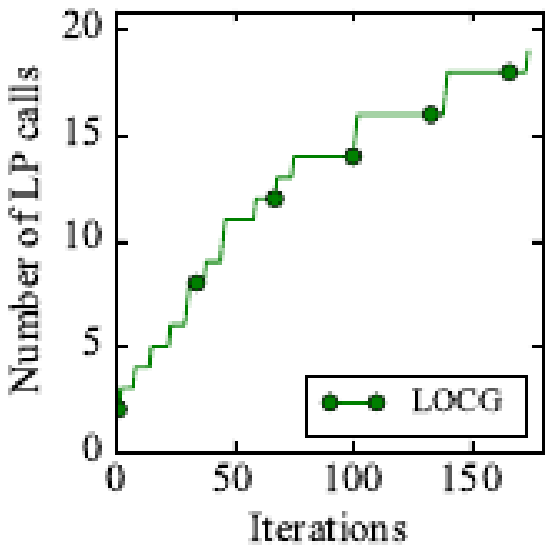}
  &
  \includegraphics[height=0.35\linewidth]{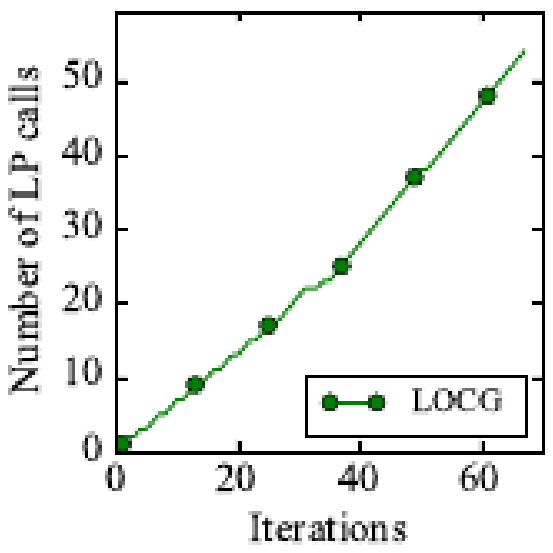}
  \\
  cache hit rate: \(89.1\%\)
  &
  cache hit rate: \(20.6\%\)
  \end{tabular}
  \caption{\label{fig:tspLarge}
    LOCG vs. OCG over TSP polytope for a graph with \(16\)
    nodes with a time limit of \(7200\)
    seconds.
    OCG was not able to complete a single
    iteration and in the quadratic case (right)
    even LOCG could not complete any more iteration after 50s.
    The quadratic losses nicely demonstrate
    speed improvements (mostly) through early termination of the linear
    optimization as the cache rate is only \(20.6\%\).
  }
\end{figure*}
For these small instances, the oracle problem can be solved in
reasonable time.
Another instance of the structured regression problem
uses the standard formulation of the
cut polytope for graphs with \(23\) and \(28\) nodes
as the feasible region
(Figures~\ref{fig:maxcutSmall} and~\ref{fig:maxcutLarge}).
\begin{figure*}
  \centering
  \begin{tabular}{cc}
  \includegraphics[height=0.35\linewidth]{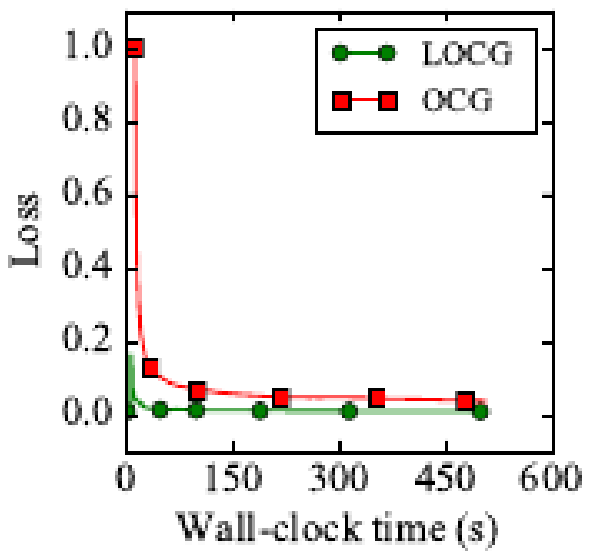}
  &
  \includegraphics[height=0.35\linewidth]{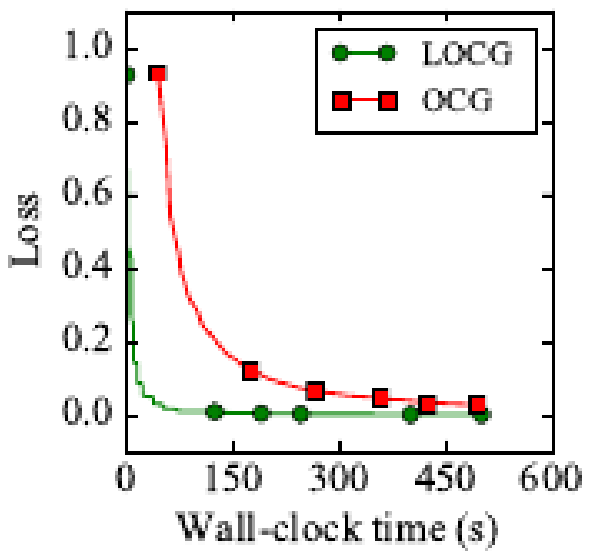}
  \\
  \includegraphics[height=0.35\linewidth]{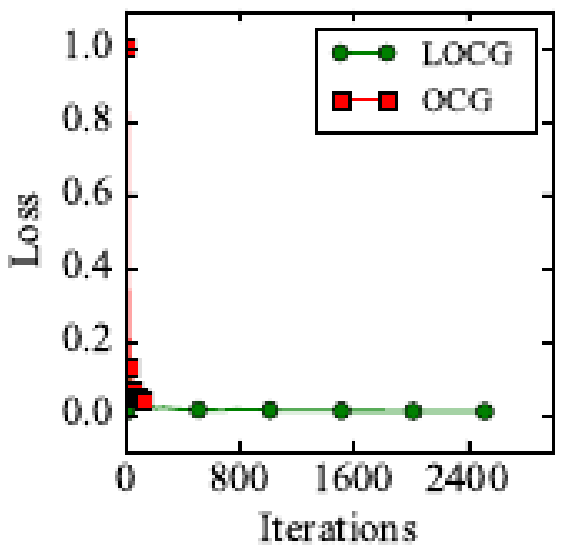}
  &
  \includegraphics[height=0.35\linewidth]{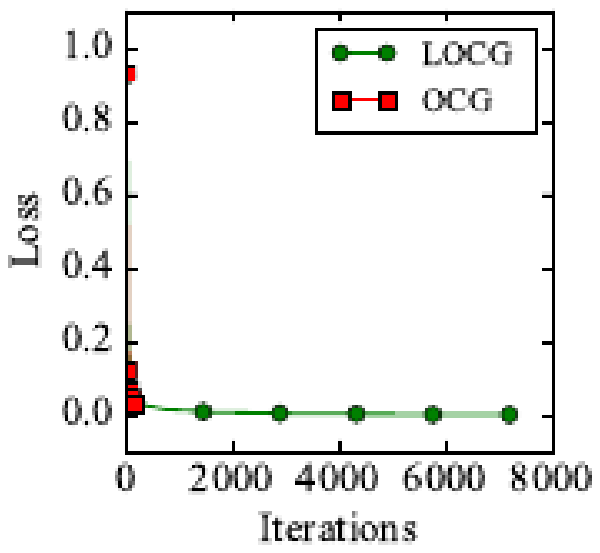}
  \\
  \includegraphics[height=0.35\linewidth]{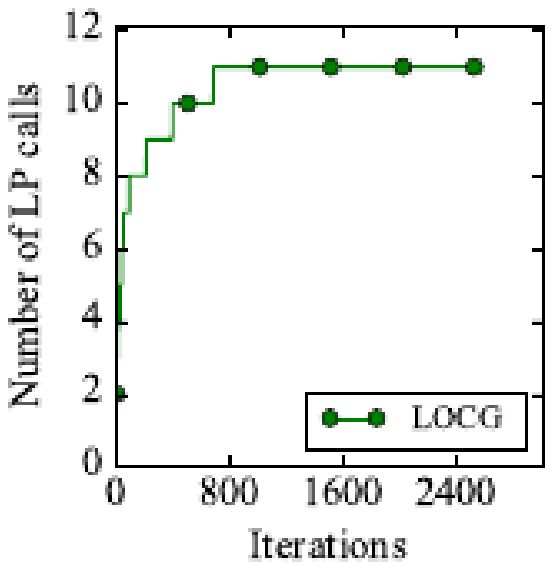}
  &
  \includegraphics[height=0.35\linewidth]{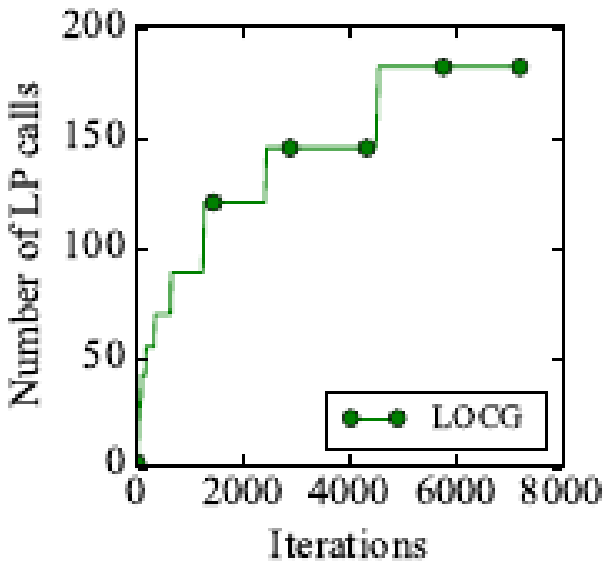}
  \\
  cache hit rate: \(99.6\%\)
  &
  cache hit rate: \(97.5\%\)
  \end{tabular}
  \caption{\label{fig:maxcutSmall} LOCG vs. OCG on the cut polytope
    for a graph with 23 nodes.  Both LOCG and OCG converge to the
    optimum in a few iterations for linear losses, while LOCG is
    remarkably faster for quadratic losses. As can be seen here the
    advantage of lazy algorithms strongly correlates with the
    difficulty of linear optimization.  For linear losses, remarkably
    LOCG needed no LP oracle calls after one third of the time.  }
\end{figure*}
\begin{figure*}
  \centering
  \begin{tabular}{cc}
  \includegraphics[height=0.35\linewidth]{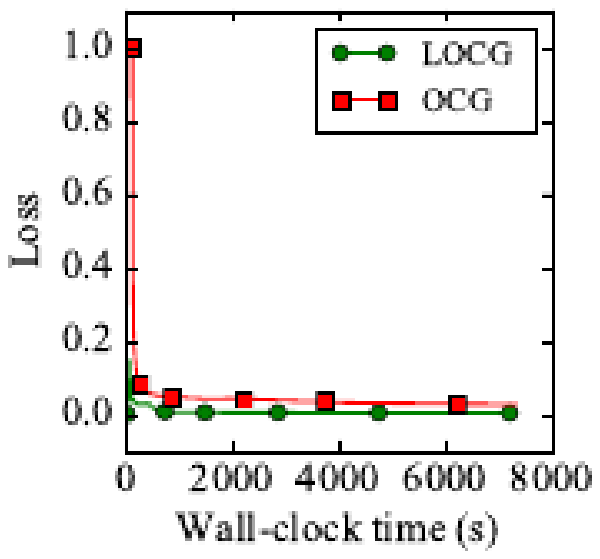}
  &
  \includegraphics[height=0.35\linewidth]{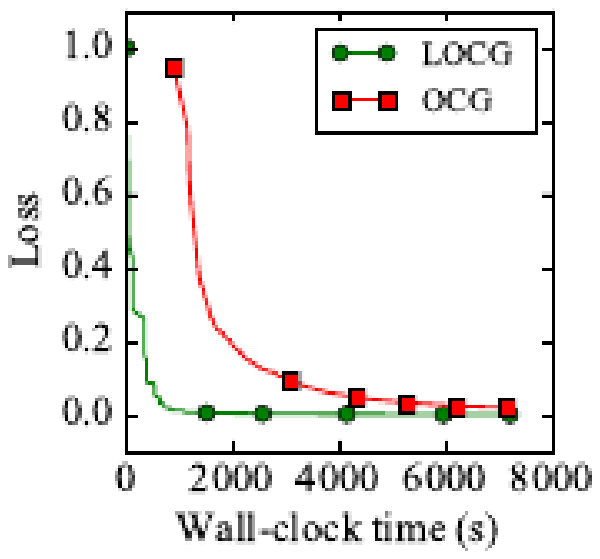}
  \\
  \includegraphics[height=0.35\linewidth]{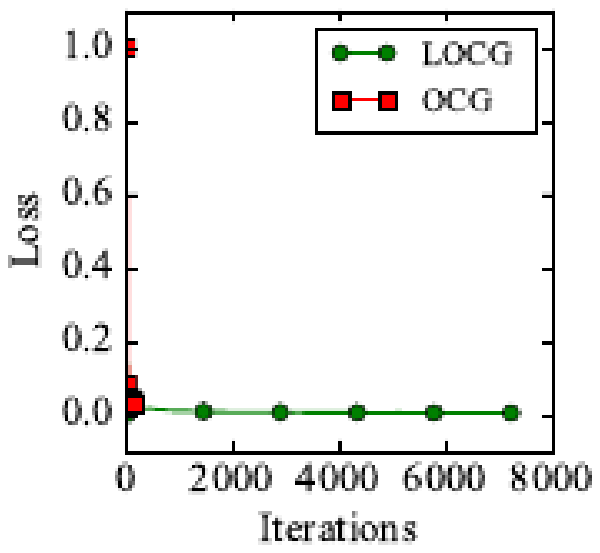}
  &
  \includegraphics[height=0.35\linewidth]{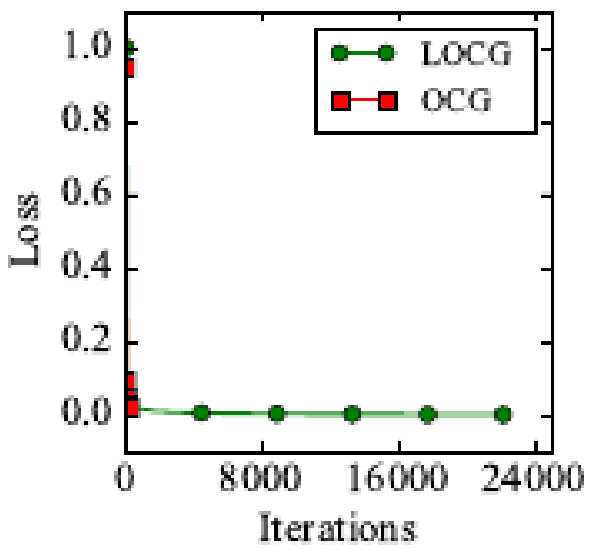}
  \\
  \includegraphics[height=0.35\linewidth]{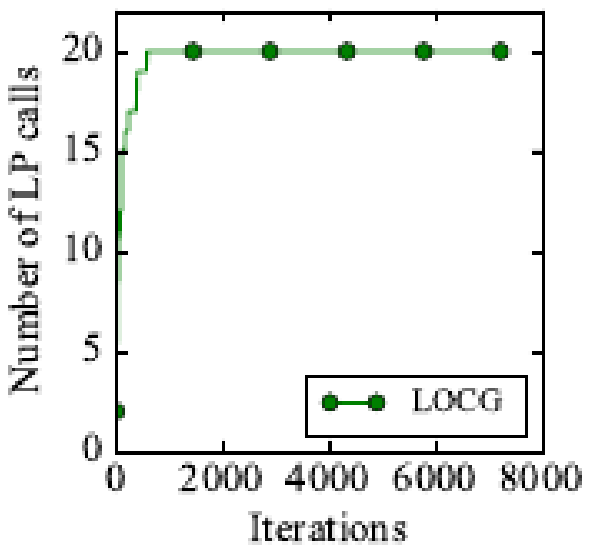}
  &
  \includegraphics[height=0.35\linewidth]{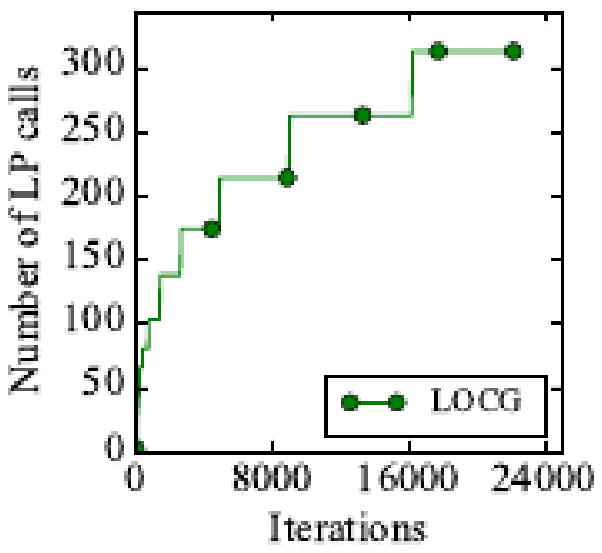}
  \\
  cache hit rate: \(99.7\%\)
  &
  cache hit rate: \(98.6\%\)
  \end{tabular}
  \caption{\label{fig:maxcutLarge}
    LOCG vs. OCG over cut polytope for a 28-node graph.
    As for the smaller problem,
    this also illustrates the advantage of lazy algorithms
    when linear optimization is expensive.
    Again, LOCG needed no oracle calls after a small initial amount of
    time.
  }
\end{figure*}
We also tested our algorithm on are the quadratic unconstrained
boolean optimization (QUBO) instances defined on Chimera graphs
\citep{dash2013note}, which are available at
\url{http://researcher.watson.ibm.com/researcher/files/us-sanjeebd/chimera-data.zip}. The
instances are relatively hard albeit their rather small size and in
general the problem is NP-hard.  (Figure~\ref{fig:quboSmall}
and~\ref{fig:quboLarge}).

\begin{figure*}
  \centering
  \begin{tabular}{cc}
  \includegraphics[height=0.35\linewidth]{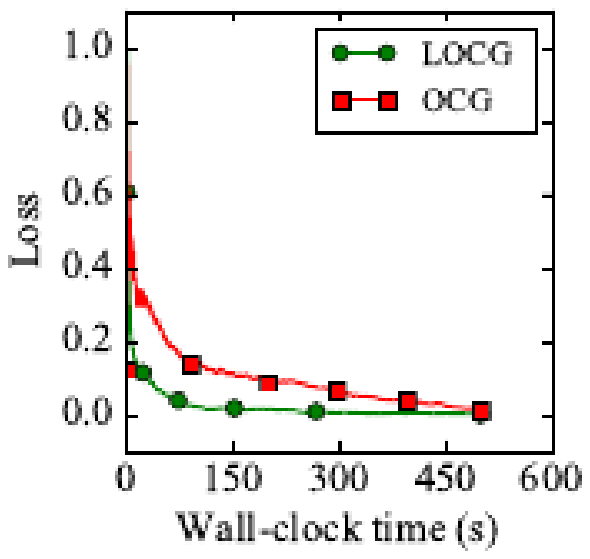}
  &
  \includegraphics[height=0.35\linewidth]{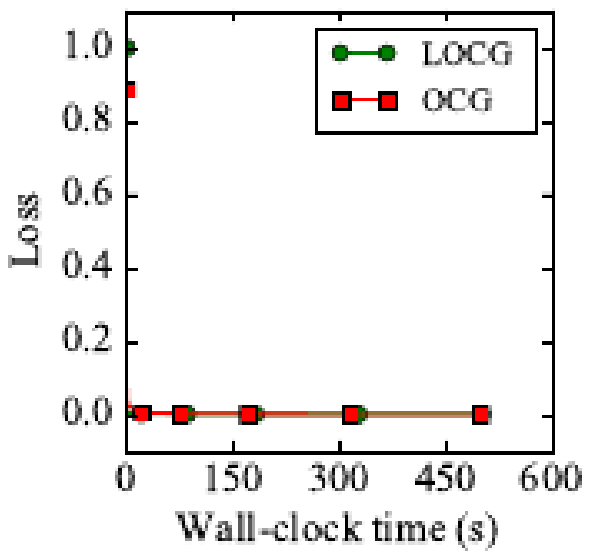}
  \\
  \includegraphics[height=0.35\linewidth]{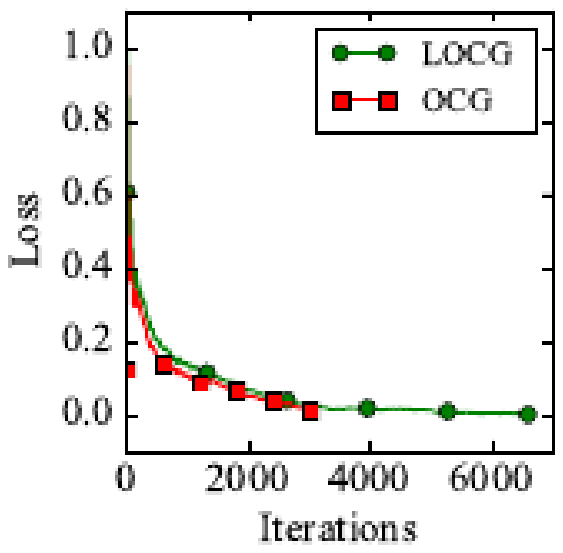}
  &
  \includegraphics[height=0.35\linewidth]{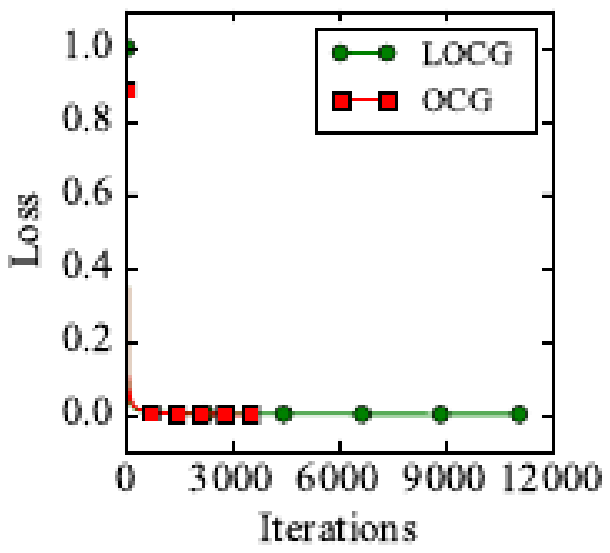}
  \\
  \includegraphics[height=0.35\linewidth]{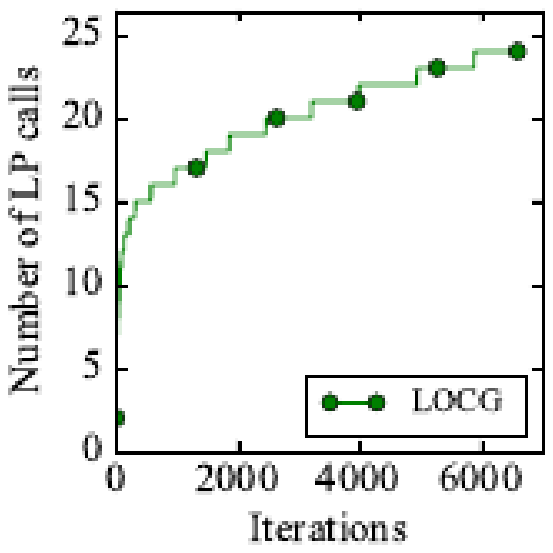}
  &
  \includegraphics[height=0.35\linewidth]{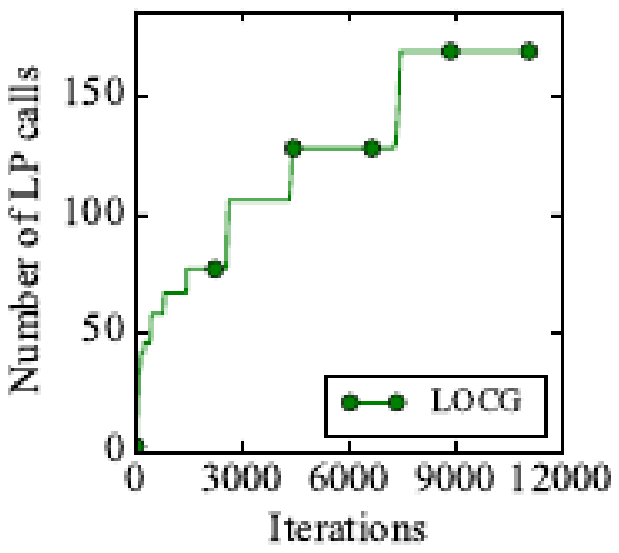}
  \\
  cache hit rate: \(99.6\%\)
  &
  cache hit rate: \(98.5\%\)
  \end{tabular}
  \caption{\label{fig:quboSmall} LOCG vs. OCG on a small QUBO instance.
    For quadratic losses (right), both algorithms converged very fast while
    LOCG still has a significant edge.
    This time, for linear losses (left) LOCG is noticeably faster than OCG.
  }
\end{figure*}
\begin{figure*}
  \centering
  \begin{tabular}{cc}
  \includegraphics[height=0.35\linewidth]{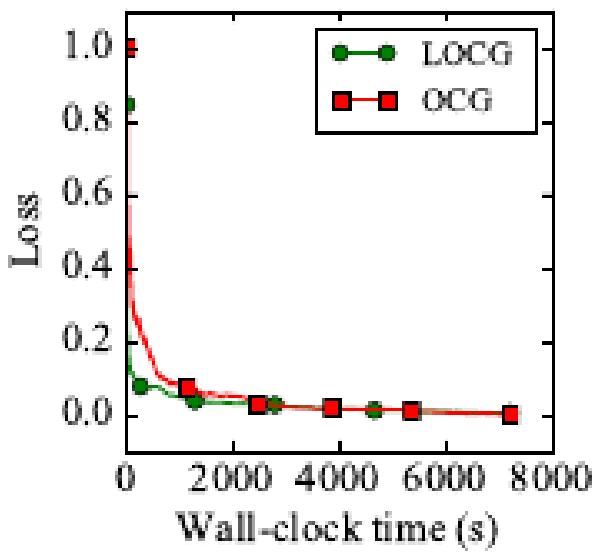}
  &
  \includegraphics[height=0.35\linewidth]{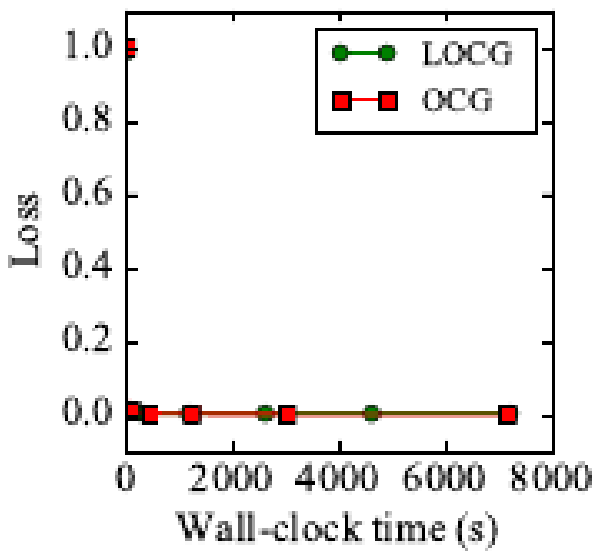}
  \\
  \includegraphics[height=0.35\linewidth]{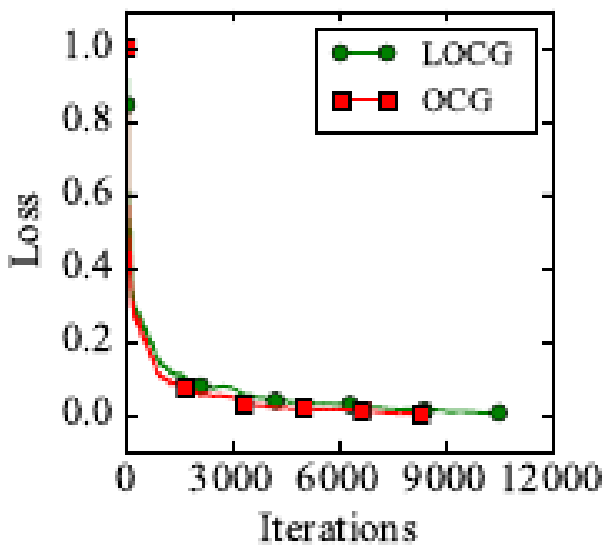}
  &
  \includegraphics[height=0.35\linewidth]{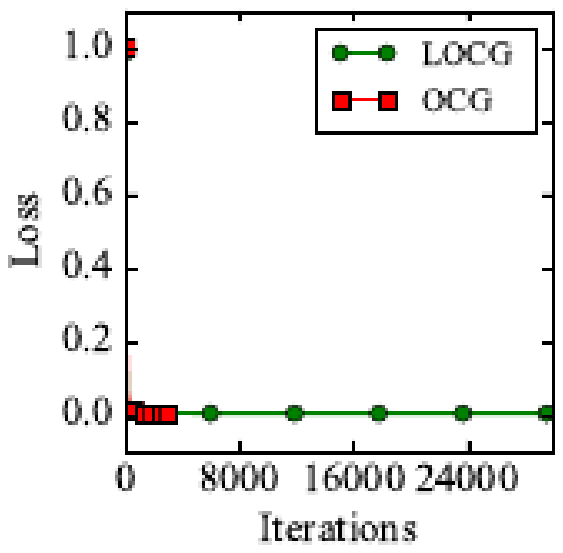}
  \\
  \includegraphics[height=0.35\linewidth]{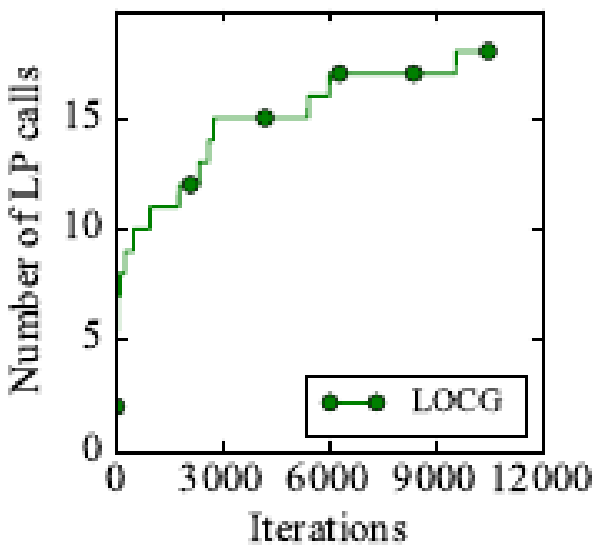}
  &
  \includegraphics[height=0.35\linewidth]{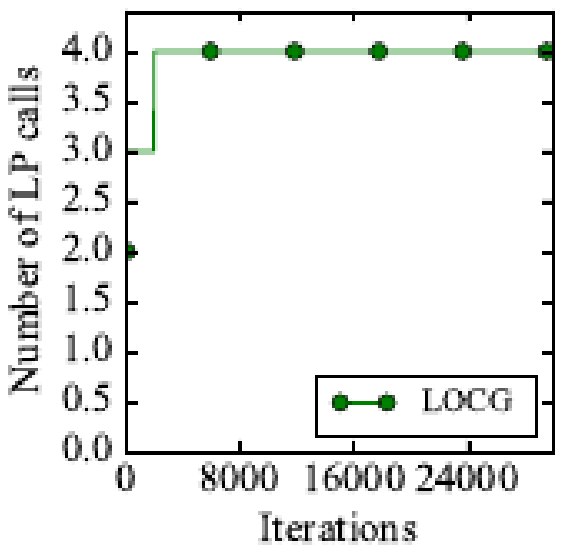}
  \\
  cache hit rate: \(99.8\%\)
  &
  cache hit rate: \(99.99\%\)
  \end{tabular}
  \caption{\label{fig:quboLarge} LOCG vs. OCG on a large QUBO instance.
    Both algorithms converge fast to the optimum. Interestingly, LOCG
    only performs \(4\) LP calls.
  }
\end{figure*}
One instance of the video colocalization problem uses a path polytope from
\url{http://lime.cs.elte.hu/~kpeter/data/mcf/netgen/} that was
generated with the \texttt{netgen} graph generator
(Figure~\ref{fig:path}).
Most of these instances are very large-scale minimum cost flow
instances with several tens of thousands nodes in the
underlying graphs,
therefore solving still takes considerable
time despite the problem being in P.
\begin{figure*}
  \centering
  \begin{tabular}{cc}
  \includegraphics[height=0.35\linewidth]{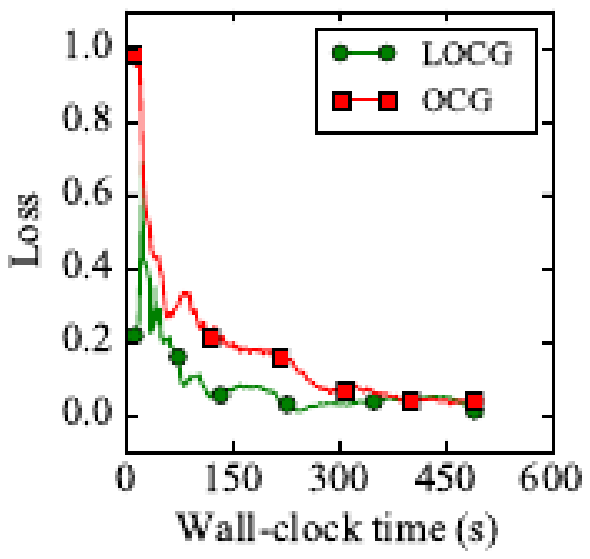}
  &
  \includegraphics[height=0.35\linewidth]{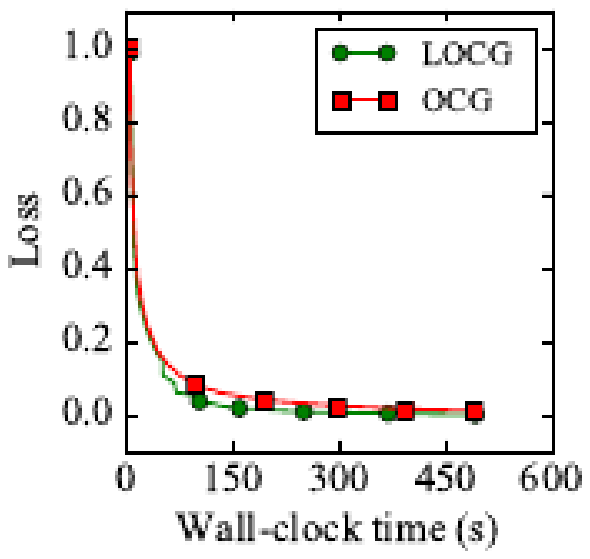}
  \\
  \includegraphics[height=0.35\linewidth]{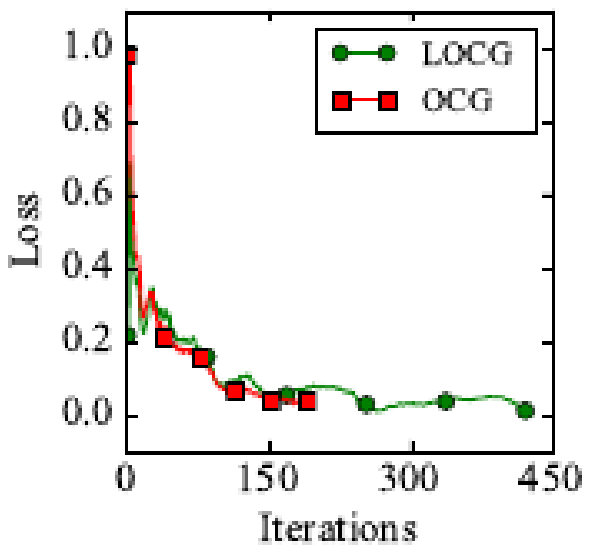}
  &
  \includegraphics[height=0.35\linewidth]{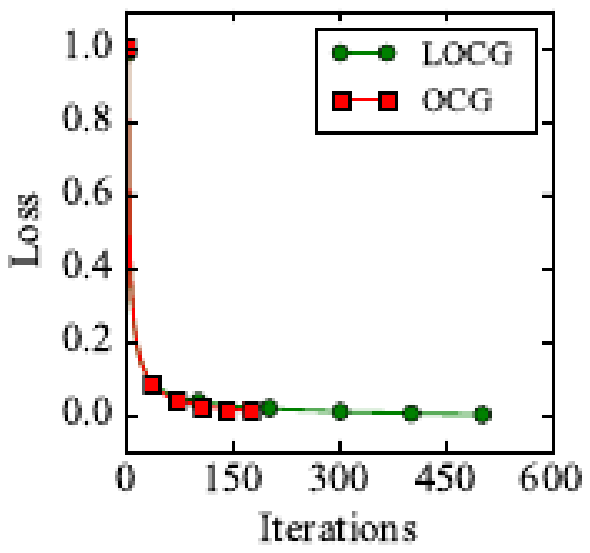}
  \\
  \includegraphics[height=0.35\linewidth]{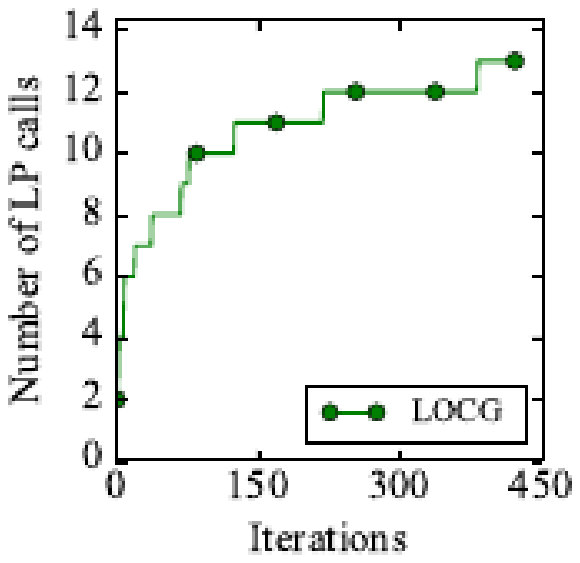}
  &
  \includegraphics[height=0.35\linewidth]{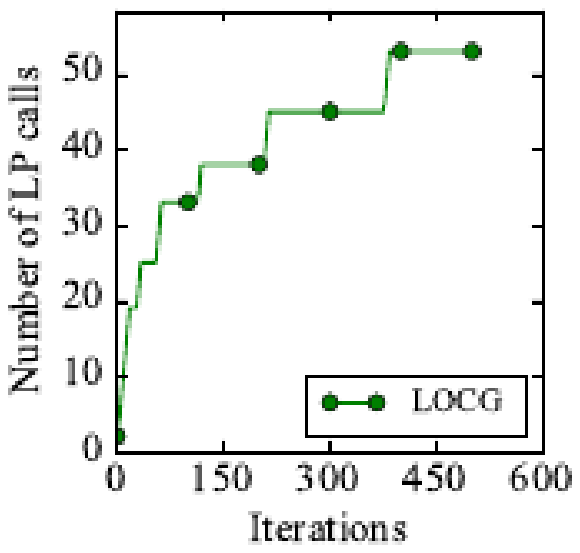}
  \\
  cache hit rate: \(97.0\%\)
  &
  cache hit rate: \(89.6\%\)
  \end{tabular}
  \caption{\label{fig:path} LOCG vs. OCG over a path polytope.
    Similar convergence rate in the number of iterations,
    but significant difference in terms of wall-clock time.
  }
\end{figure*}
We tested on the structured regression problems with
the MIPLIB \citep{AchterbergKochMartin2006,koch2011miplib}) instances
\texttt{eil33-2} (Figure~\ref{fig:eil33-2})
and \texttt{air04} (Figure~\ref{fig:air04}) as feasible regions.
\begin{figure*}
  \centering
  \begin{tabular}{cc}
  \includegraphics[height=0.35\linewidth]{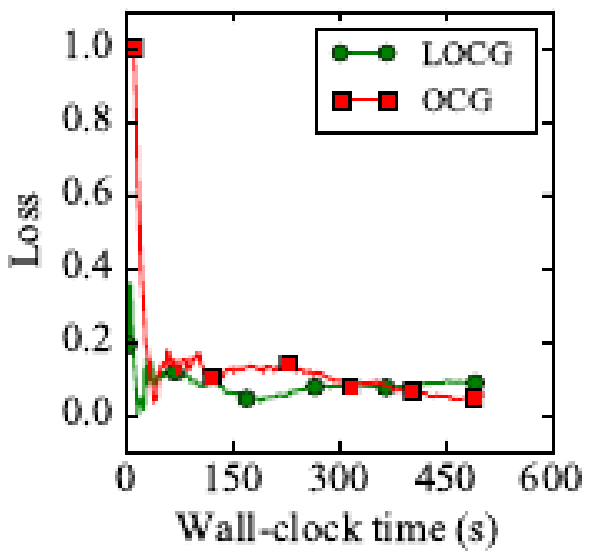}
  &
  \includegraphics[height=0.35\linewidth]{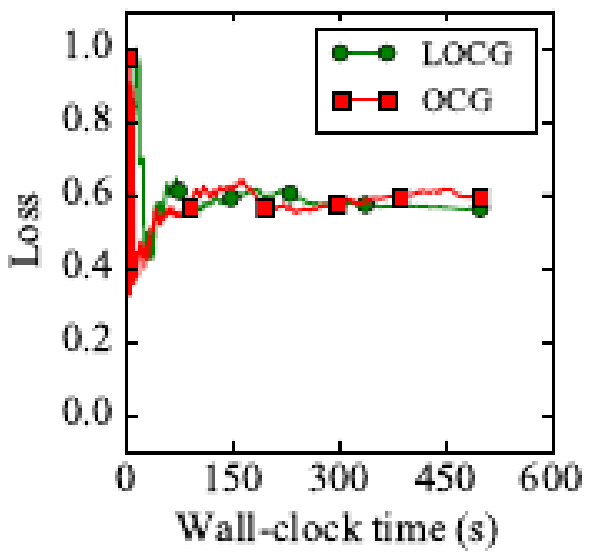}
  \\
  \includegraphics[height=0.35\linewidth]{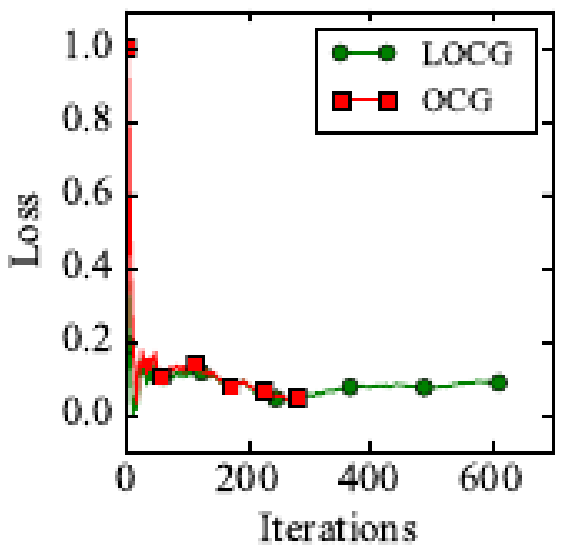}
  &
  \includegraphics[height=0.35\linewidth]{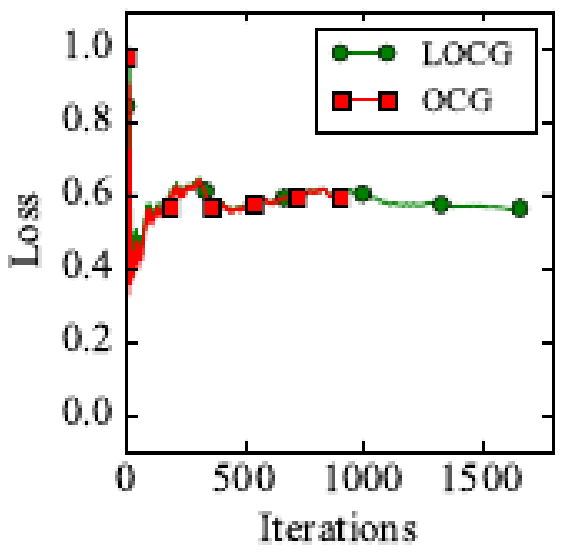}
  \\
  \includegraphics[height=0.35\linewidth]{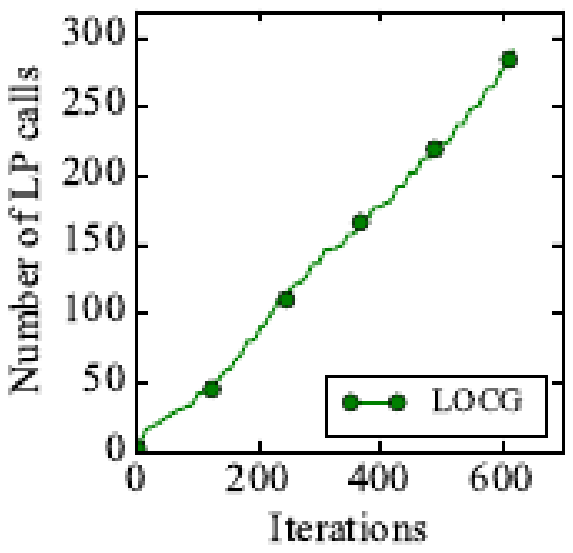}
  &
  \includegraphics[height=0.35\linewidth]{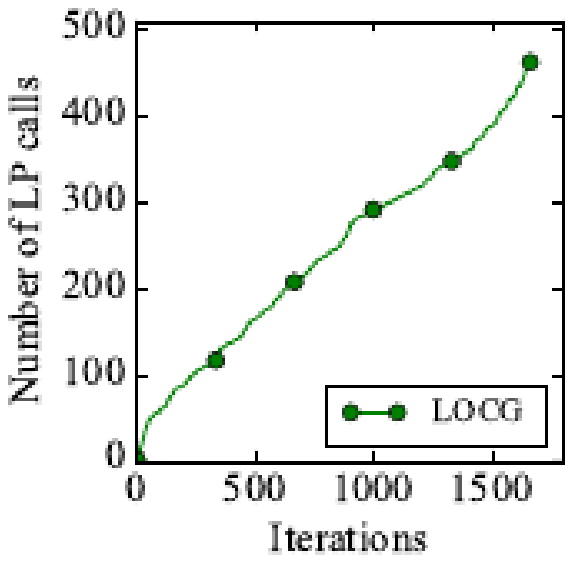}
  \\
  cache hit rate: \(53.1\%\)
  &
  cache hit rate: \(72.2\%\)
  \end{tabular}
  \caption{\label{fig:eil33-2} LOCG vs. OCG on the MIPLIB instance \texttt{eil33-2}.
    All algorithms performed comparably,
    due to fast convergence in this case.
  }
\end{figure*}
\begin{figure*}
  \centering
  \begin{tabular}{cc}
  \includegraphics[height=0.35\linewidth]{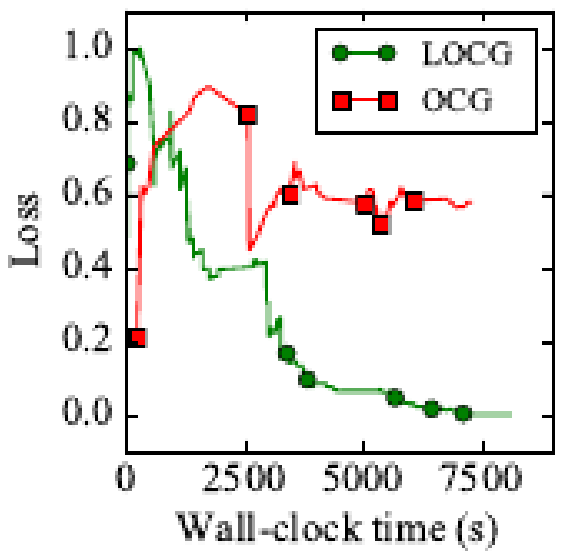}
  &
  \includegraphics[height=0.35\linewidth]{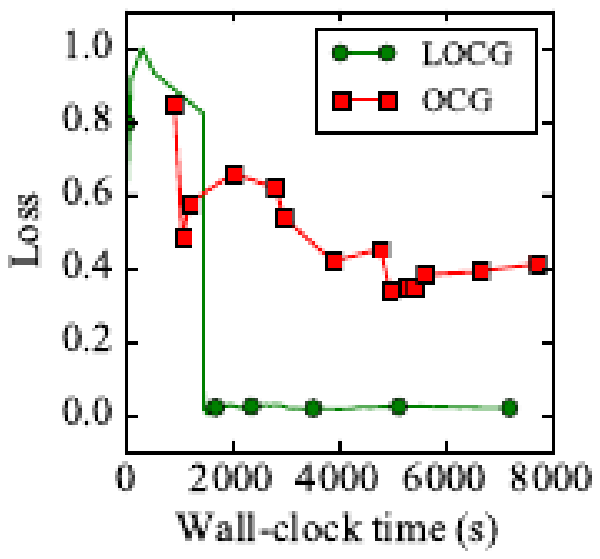}
  \\
  \includegraphics[height=0.35\linewidth]{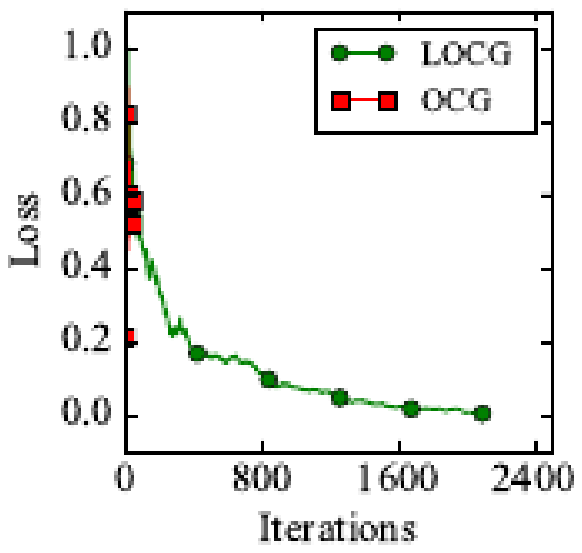}
  &
  \includegraphics[height=0.35\linewidth]{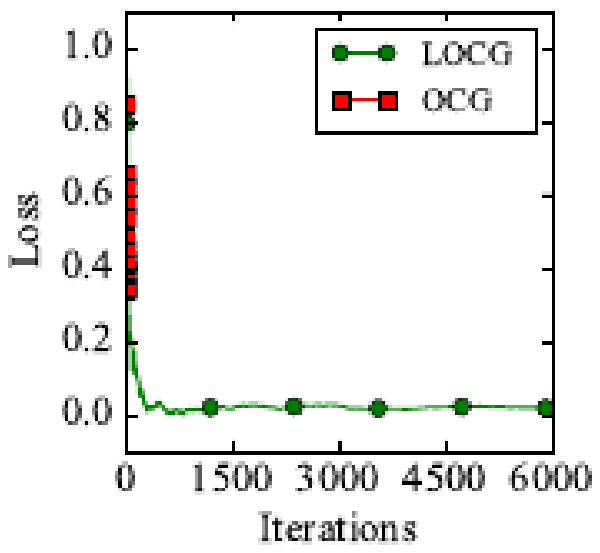}
  \\
  \includegraphics[height=0.35\linewidth]{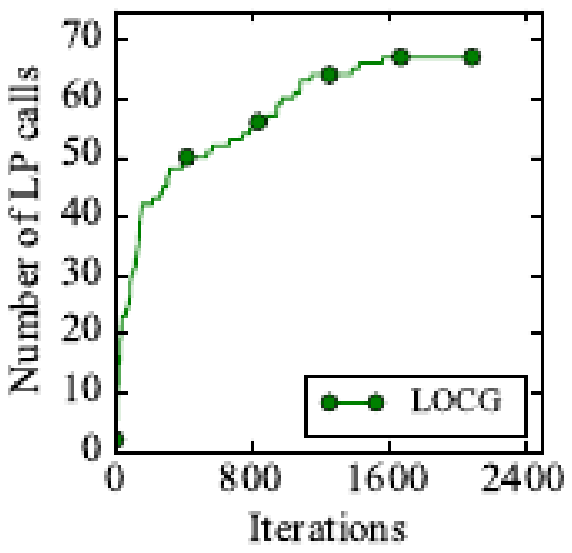}
  &
  \includegraphics[height=0.35\linewidth]{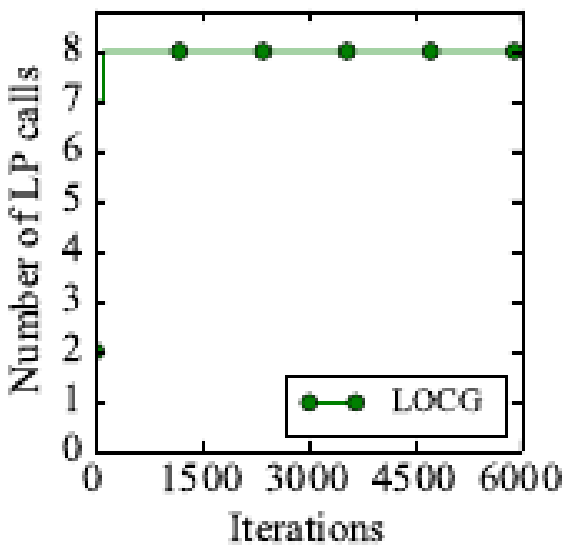}
  \\
  cache hit rate: \(96.8\%\)
  &
  cache hit rate: \(99.9\%\)
  \end{tabular}
  \caption{\label{fig:air04} LOCG vs. OCG on the MIPLIB instance \texttt{air04}.
    LOCG clearly outperforms OCG
    as the provided time was not enough for OCG to complete the
    necessary number of iterations for entering reasonable convergence.
  }
\end{figure*}
Finally, for the spanning tree problem,
we used the well-known extended formulation with \(O(n^{3})\)
inequalities for an \(n\)-node graph.
We considered graphs with 10 and 25 nodes
(Figures~\ref{fig:spt10} and~\ref{fig:spt25}).
\begin{figure*}
  \centering
  \begin{tabular}{cc}
  \includegraphics[height=0.35\linewidth]{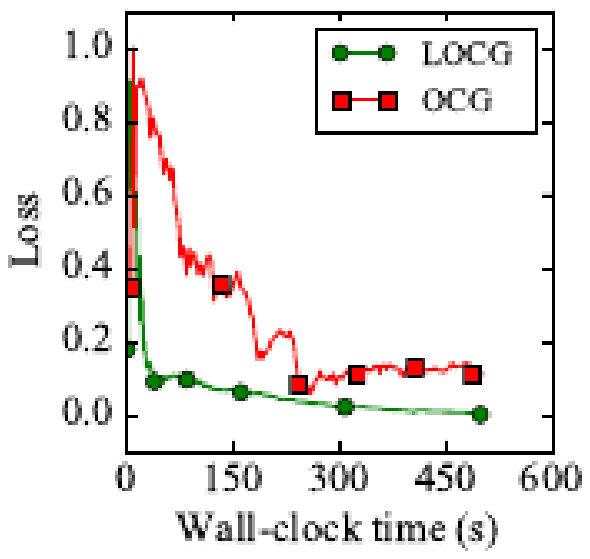}
  &
  \includegraphics[height=0.35\linewidth]{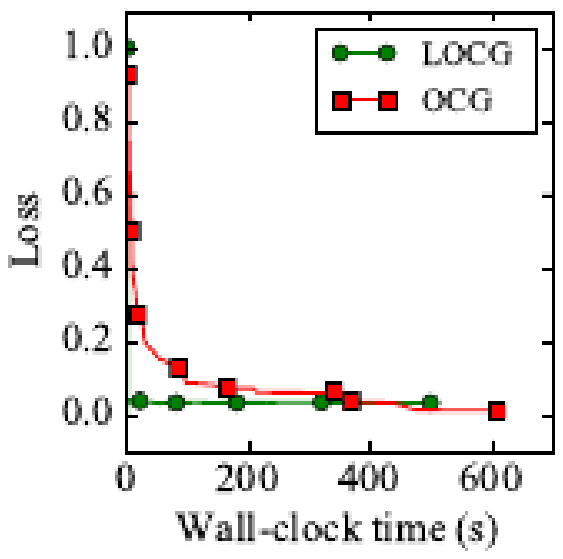}
  \\
  \includegraphics[height=0.35\linewidth]{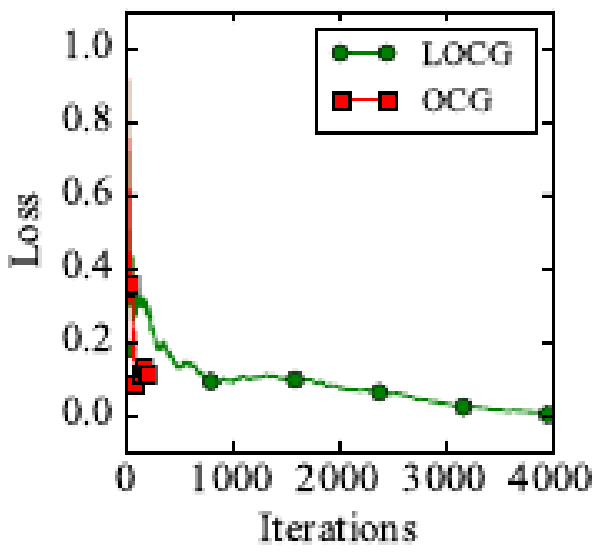}
  &
  \includegraphics[height=0.35\linewidth]{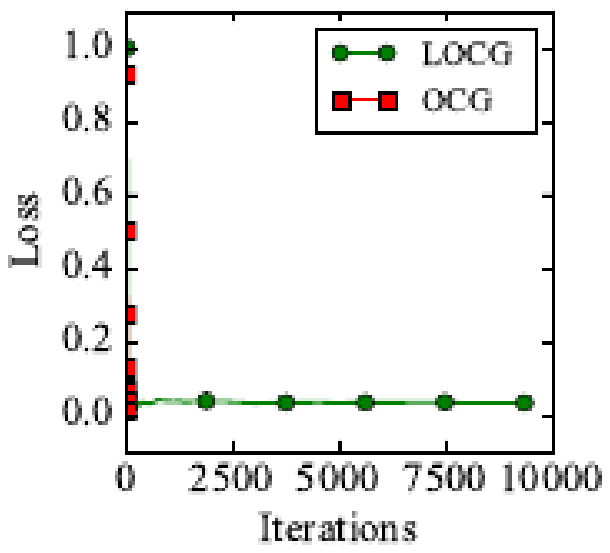}
  \\
  \includegraphics[height=0.35\linewidth]{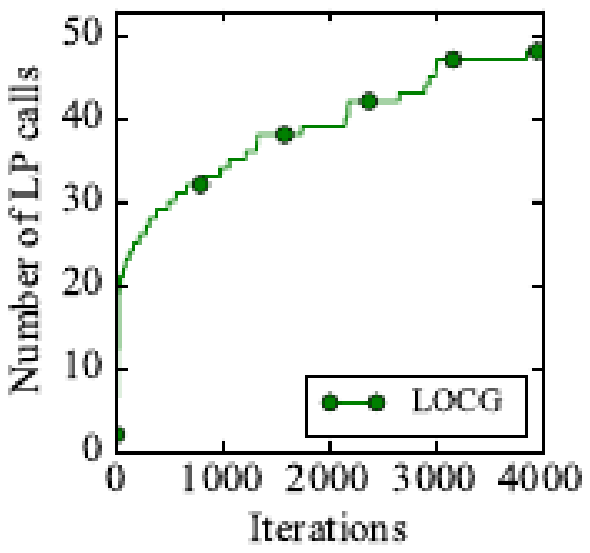}
  &
  \includegraphics[height=0.35\linewidth]{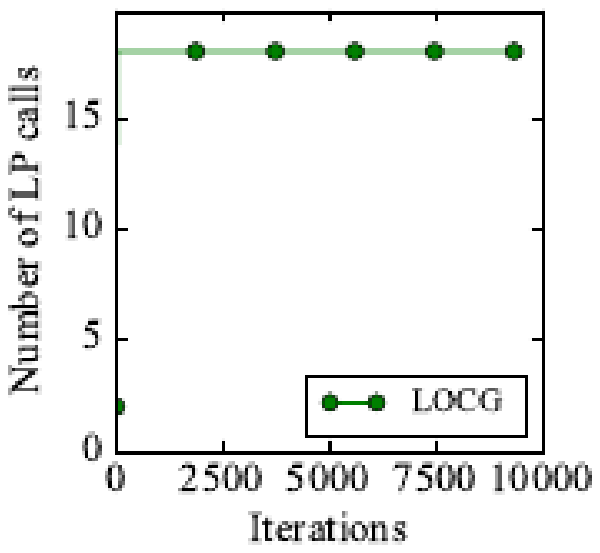}
  \\
  cache hit rate: \(98.8\%\)
  &
  cache hit rate: \(99.8\%\)
  \end{tabular}
  \caption{\label{fig:spt10} LOCG vs. OCG on a spanning tree instance for a 10-node graph.
    LOCG makes significantly more iterations,
    few oracle calls,
    and converges faster in wall-clock time.
  }
\end{figure*}
\begin{figure*}
  \centering
  \begin{tabular}{cc}
  \includegraphics[height=0.35\linewidth]{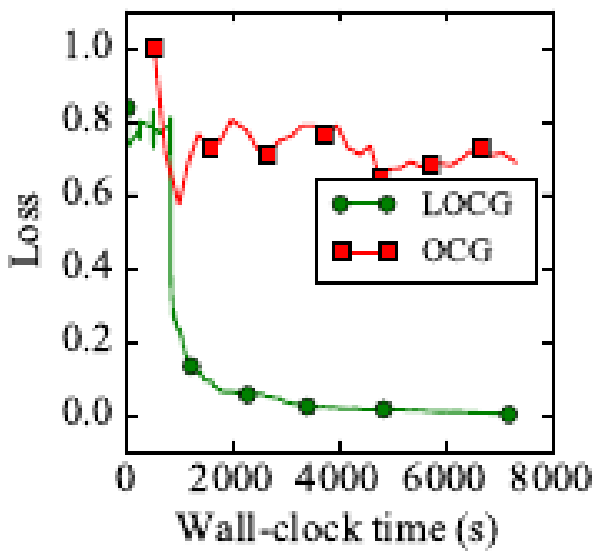}
  &
  \includegraphics[height=0.35\linewidth]{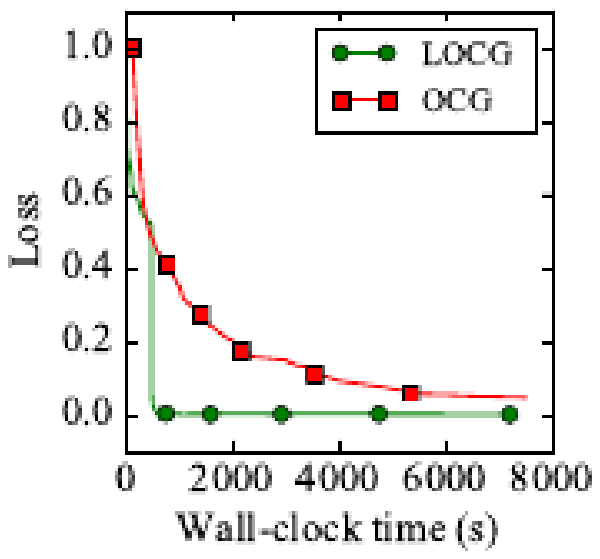}
  \\
  \includegraphics[height=0.35\linewidth]{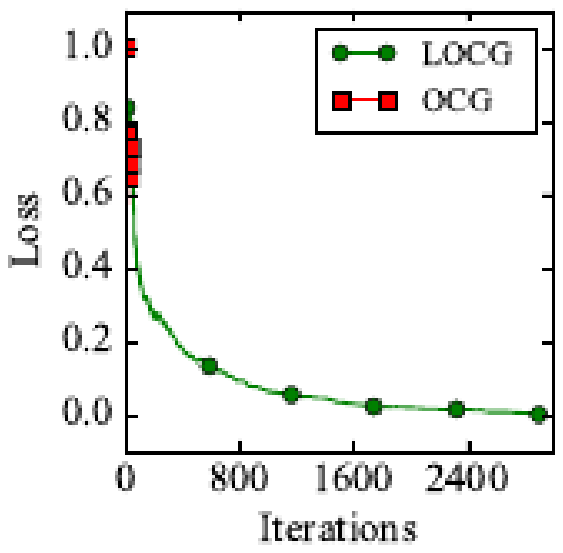}
  &
  \includegraphics[height=0.35\linewidth]{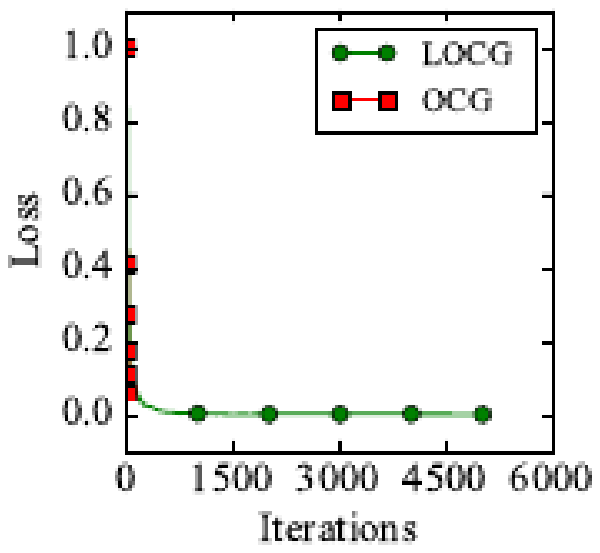}
  \\
  \includegraphics[height=0.35\linewidth]{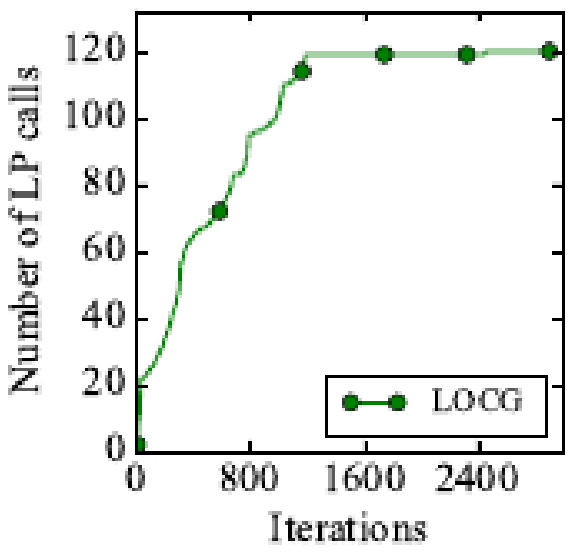}
  &
  \includegraphics[height=0.35\linewidth]{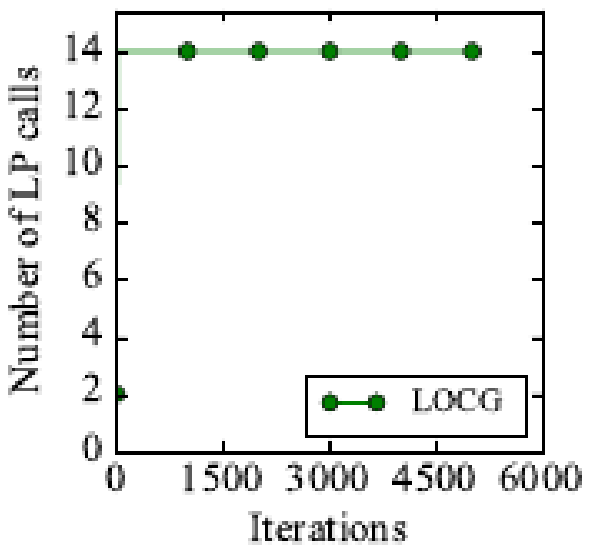}
  \\
  cache hit rate: \(95.9\%\)
  &
  cache hit rate: \(99.7\%\)
  \end{tabular}
  \caption{\label{fig:spt25} LOCG vs. OCG over a spanning tree instance for a 25-node graph.
    On the left, early fluctuation can be observed,
    bearing no consequence for later convergence rate.
    OCG did not get past this early stage.
    In both cases LOCG converges significantly faster.
  }
\end{figure*}

We observed that similarly to the offline case
while OCG and LOCG converge comparably in the number
of iterations, the lazy LOCG performed significantly more iterations;
for hard problems, where linear optimization is costly and convergence
requires a large number of iterations, this led LOCG converging much
faster in wall-clock time.  In extreme cases OCG could not complete
even a single iteration.  This is due to LOCG only requiring
\emph{some} good enough solution, whereas OCG requires a stronger
guarantee. This is reflected in faster oracle calls for LOCG.

\subsection{Performance improvements, parameter sensitivity, and tuning}
\label{sec:perf-impr-tuning}

\subsubsection{Effect of caching}
\label{sec:effect-cache}

As mentioned before, lazy algorithms have two improvements: caching
and early termination.  Here we depict the effect of caching in
Figure~\ref{fig:cacheEffect}, comparing OCG (no caching, no early
termination), LOCG (caching and early termination) and LOCG (only
early termination) (see Algorithm~\ref{alg:online-FW-WSep}).  We did
not include a caching-only OCG variant, because caching without early
termination does not make much sense: in each iteration a new linear
optimization problem has to be solved; previous solutions can hardly
be reused as they are unlikely to be optimal for the new linear
optimization problem.

\begin{figure*}
  \centering
  \includegraphics[width=0.35\linewidth]{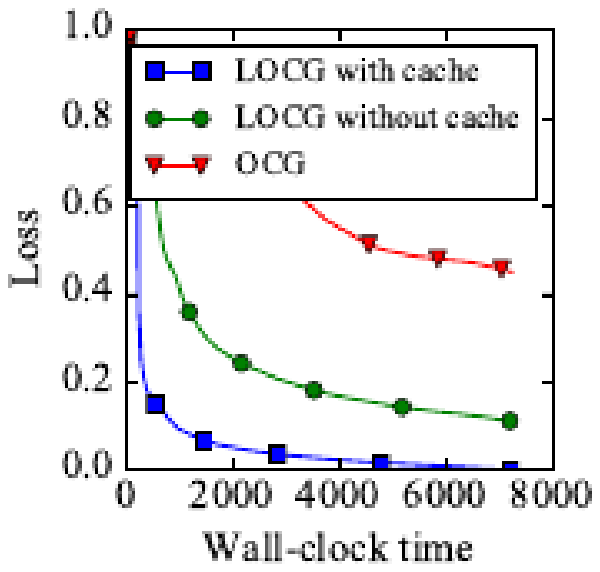}
  \includegraphics[width=0.35\linewidth]{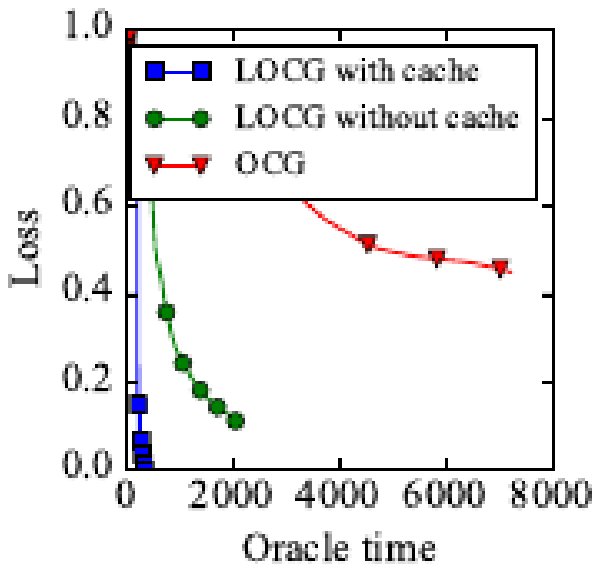}
  \caption{\label{fig:cacheEffect} Performance gain due to caching and
    early termination for online optimization over a maximum cut
    problem with linear losses. The red line is the OCG baseline, the
    green one is the lazy variant using only early termination, and
    the blue one uses caching and early termination.  Left: loss
    vs. wall-clock time. Right: loss vs. total time spent in oracle
    calls.  Time limit was \(7200\)
    seconds.  Caching allows for a significant improvement in loss
    reduction in wall-clock time. The effect is even more obvious in
    oracle time as caching cuts out a large number of oracle calls. }
\end{figure*}


\begin{figure}[htbp]
  \centering
  \begin{tabular}{cc}
  \includegraphics[height=0.45\linewidth]{results/vanilla_fw/netgen_12b_lasso_cache-100_det_dual_bound_wall_clock}
  &
  \includegraphics[height=0.45\linewidth]{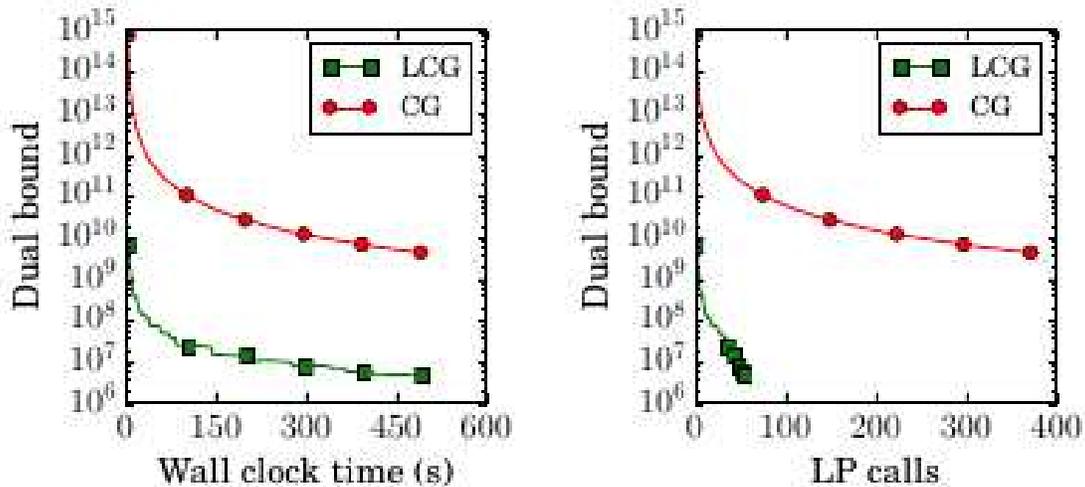}
  \end{tabular}
  \caption{\label{fig:video-colocalization}
  Performance on an instance of the video colocalization problem.
  We solve quadratic minimization over a flow polytope
  and report the achieved dual bound
  (or Wolfe-gap) over wall-clock time in seconds in logscale
  on the left and over the number of actual LP calls on the right.
  We used the parameter-free variant of
  the Lazy CG algorithm, which performs in both measures significantly better than
  the non-lazy counterpart. The performance difference is more prominent
  in the number of LP calls.}
\end{figure}

\begin{figure}[htbp]
  \centering
  \begin{tabular}{cc}
  \includegraphics[height=0.45\linewidth]{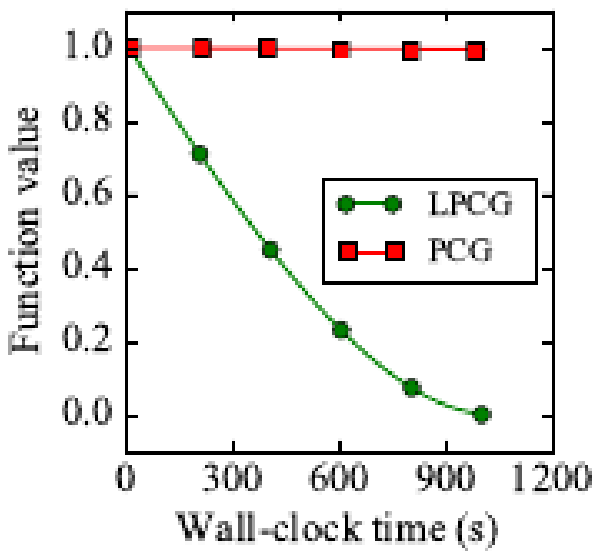}
  &
  \includegraphics[height=0.45\linewidth]{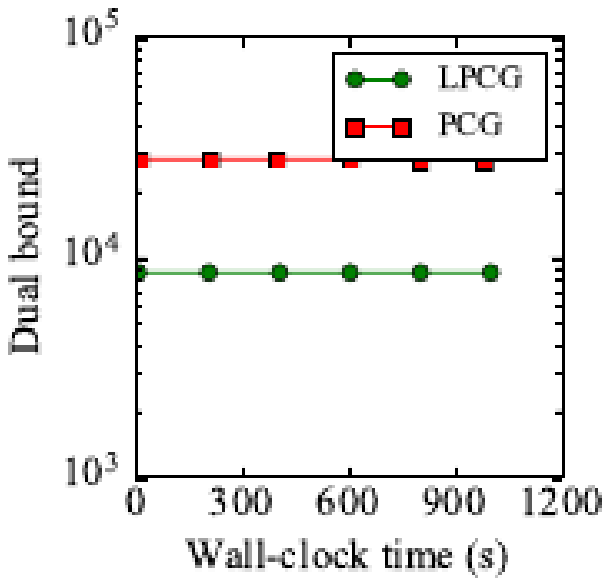}
  \end{tabular}
  \caption{\label{fig:video-colocalization-2}
  Performance on a large instance of the video colocalization problem using
PCG and its lazy variant. We observe that lazy PCG is significantly
better both in terms of function value and dual bound. Recall that the
function value is normalized between \([0,1]\).
}
\end{figure}

\begin{figure}[htbp]
  \centering
  \begin{tabular}{cc}
  \includegraphics[height=0.45\linewidth]{results/vanilla_fw/random_matrix_completion_10000_100_10_10000_totalWallClock}
  &
  \includegraphics[height=0.45\linewidth]{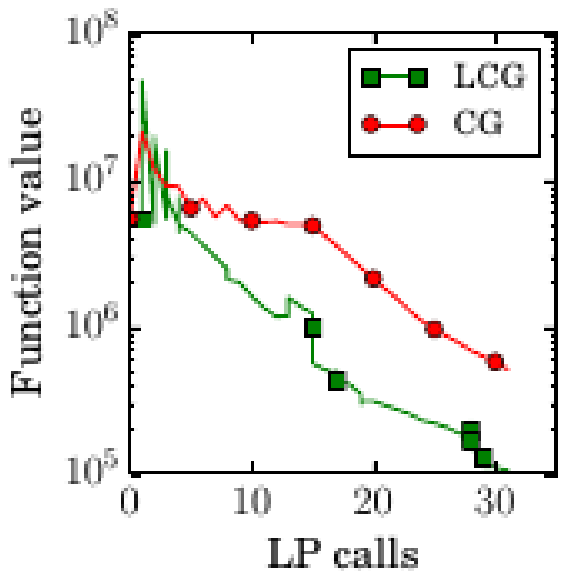}
  \end{tabular}
  \caption{\label{fig:matrix-completion}
  Performance on a matrix completion instance.
  More information about this problem can be found in the supplemental material
  (Section~\ref{sec:experiments}).
  The performance is reported as the objective
  function value over wall-clock time in seconds on the left and over LP calls on the right.
  In both measures after an initial phase the function value using LCG is much lower than with the
  non-lazy algorithm.}
\end{figure}

\begin{figure}[htbp]
  \centering
  \begin{tabular}{cc}
  \includegraphics[height=0.45\linewidth]{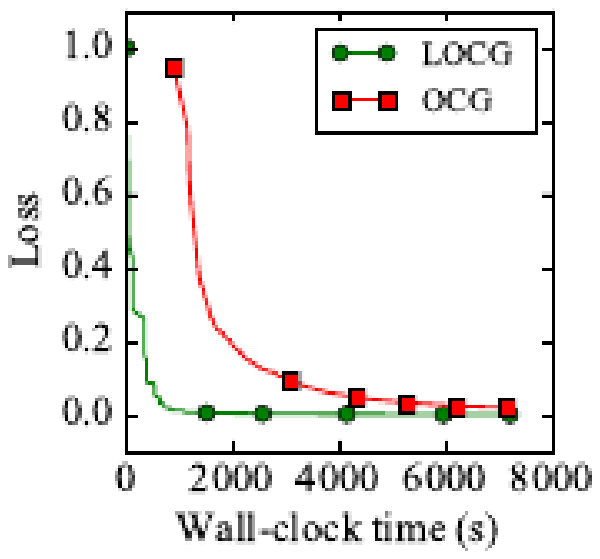}
  &
  \includegraphics[height=0.45\linewidth]{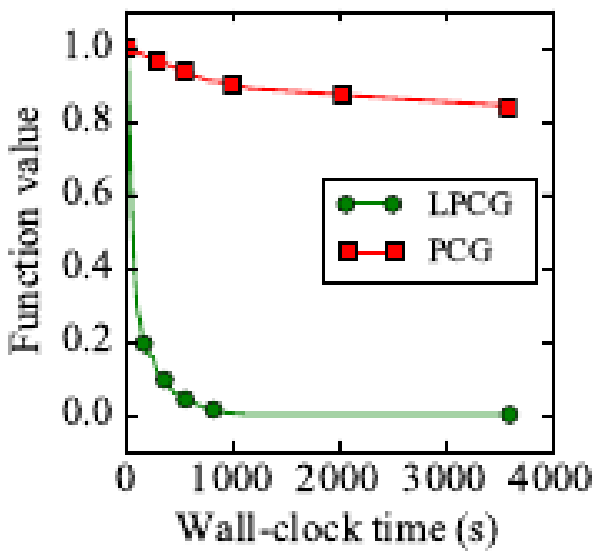}
  \\
  \includegraphics[height=0.45\linewidth]{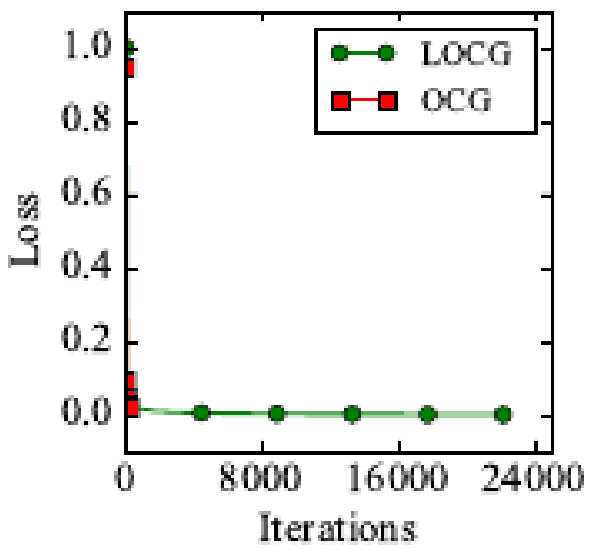}
  &
  \includegraphics[height=0.45\linewidth]{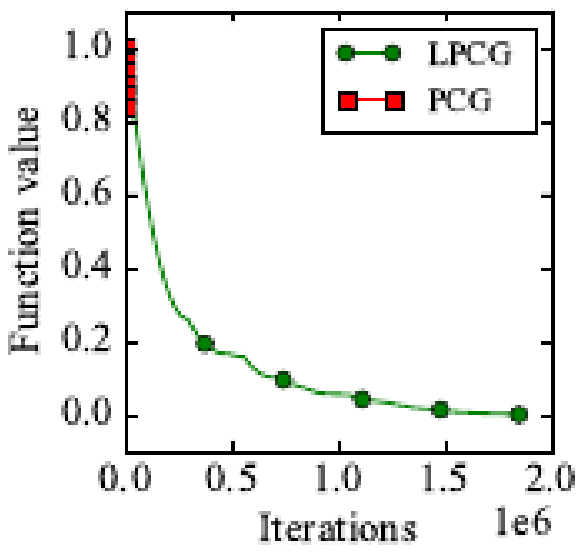}
  \\
  \includegraphics[height=0.45\linewidth]{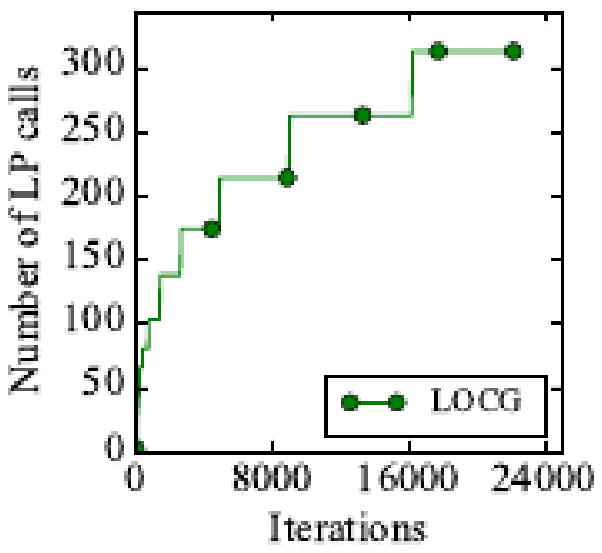}
  &
  \includegraphics[height=0.45\linewidth]{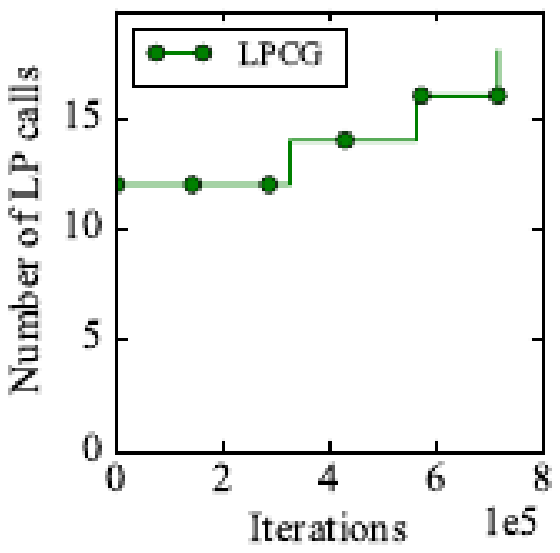}
  \\
  \end{tabular}
  \caption{\label{fig:ocg-pcg-results}
  Performance of the two lazified variants LOCG (left column) and
  LPCG (right column). The feasible regions are a cut polytope on the left and
  the MIPLIB instance \texttt{air04} on the right. The objective functions are in both
  cases quadratic, on the left randomly chosen in every step.
  We show the performance over wall clock time in seconds (first row) and over
  iterations (second row). The last row shows the number of call to the linear
  optimization oracle.
  The lazified versions perform significantly
  better in wall clock time compared to the non-lazy counterparts.
  }
\end{figure}

\subsubsection{Effect of $K$}
\label{sec:effect-k}

If the parameter \(K\)
of the oracle can be chosen, which depends on the actual oracle
implementation, then we can increase \(K\)
to bias the algorithm towards performing more positive calls. At the
same time the steps get shorter. As such there is a natural trade-off
between the cost of many positive calls vs.~a negative call. We depict
the impact of the parameter choice for \(K\) in Figure~\ref{fig:K-values}.

\begin{figure*}
  \centering
  \includegraphics[width=0.35\linewidth]{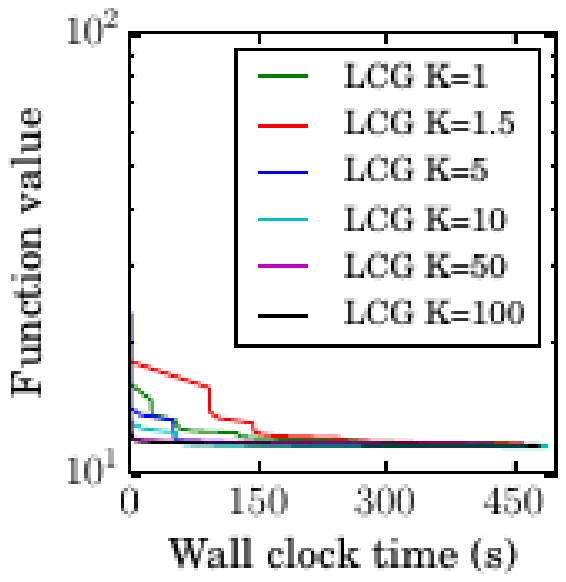}
  \includegraphics[width=0.35\linewidth]{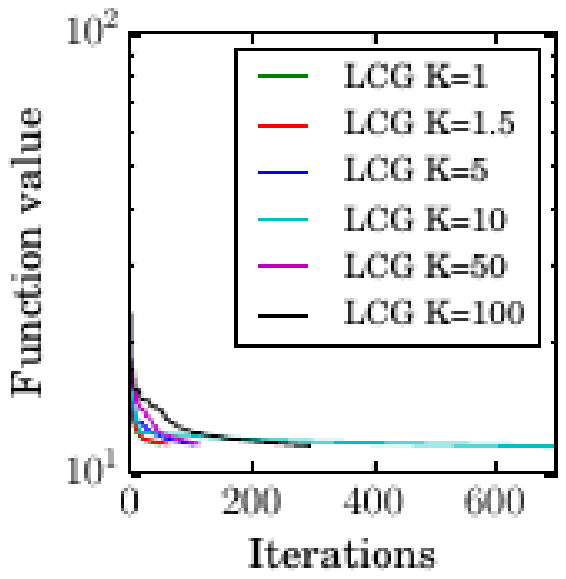}
  \caption{\label{fig:K-values} Impact of the oracle approximation parameter
    \(K\) depicted for the Lazy CG algorithm. We can see that
    increasing \(K\) leads to a deterioration of progress in
    iterations but improves performance in wall-clock time. The behavior is similar
    for other algorithms.}
\end{figure*}

\subsubsection{Paramter-free vs. textbook variant}
\label{sec:paramfreeVsTextbook}

For illustrative purposes, we compare the textbook variant of the lazy
conditional gradient (Algorithm~\ref{alg:FW-WSep}) with its
parameter-free counterpart (Algorithm~\ref{alg:CG-imp-imp}) in
Figure~\ref{fig:parameter_free_vs_textbook}. The parameter-free
variant outperforms the textbook variant due to the active management
of \(\Phi\) combined with line search.

Similar parameter-free variants can be derived for the other
algorithms; see discussion in Section~\ref{sec:frank-wolfe-with}.

\begin{figure*}
  \centering
  \begin{tabular}{cc}
  \includegraphics[height=0.35\linewidth]{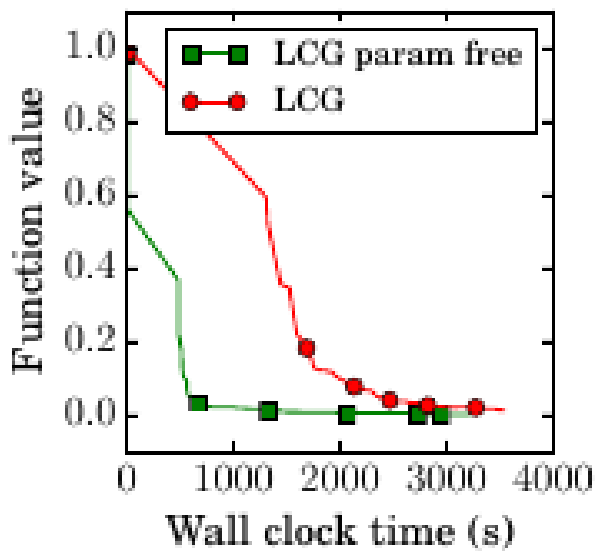}
  &
  \includegraphics[height=0.35\linewidth]{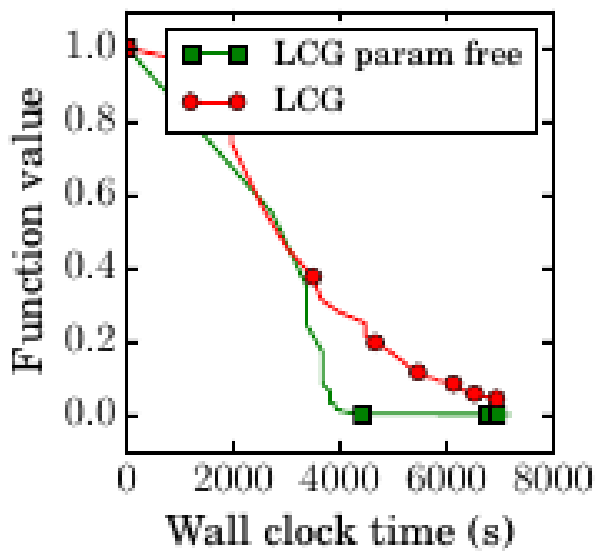}
  \\
  \includegraphics[height=0.35\linewidth]{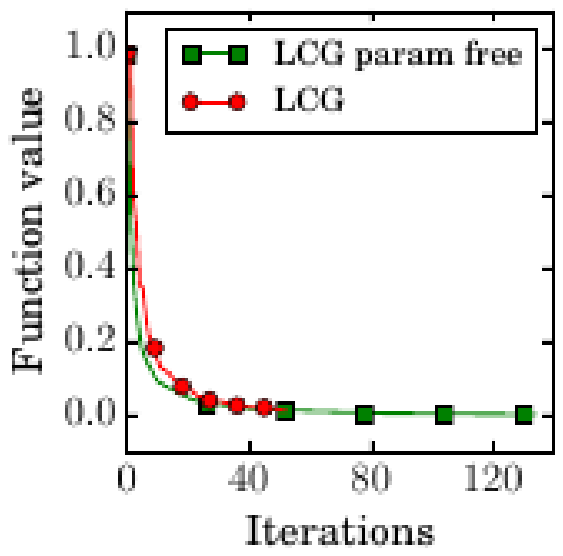}
  &
  \includegraphics[height=0.35\linewidth]{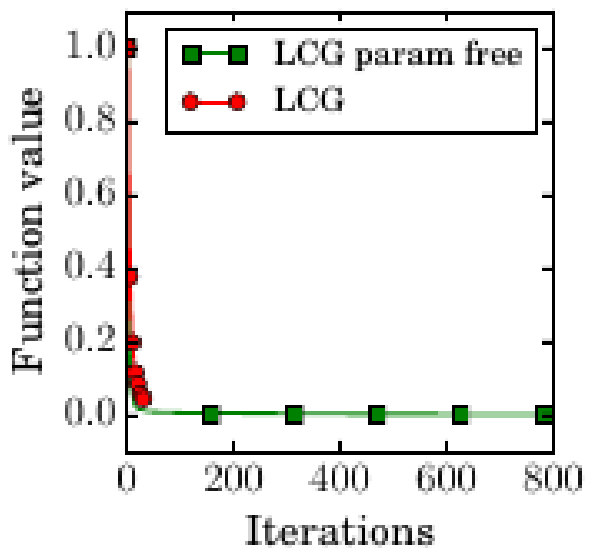}
  \\
  \includegraphics[height=0.35\linewidth]{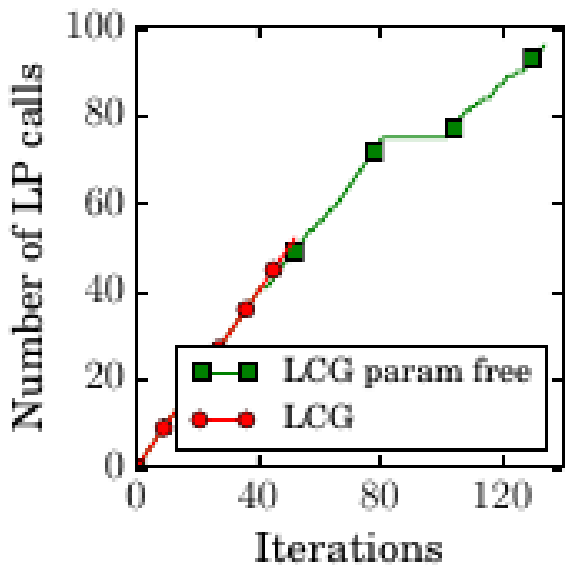}
  &
  \includegraphics[height=0.35\linewidth]{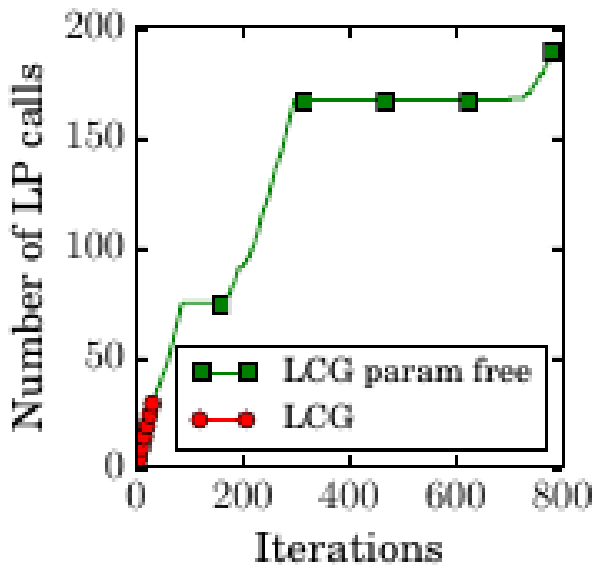}
  \end{tabular}
  \caption{\label{fig:parameter_free_vs_textbook} Comparison of the
    \lq{}textbook\rq{} variant of the Lazy CG algorithm
    (Algorithm~\ref{alg:FW-WSep}) vs. the Parameter-free Lazy CG
    (Algorithm~\ref{alg:CG-imp-imp}) depicted for two sample instances
    to demonstrate behavior. The parameter-free variant usually has a
    slighlty improved behavior in terms of iterations and a
    significantly improved behavior in terms of wall-clock
    performance. In particular, the parameter-free variant can execute
    significantly more oracle calls, due to the \(\Phi\)-halving
    strategy and the associated bounded number of negative calls (see
    Theorem~\ref{cor:negCalls}).}
\end{figure*}

\section{Final Remarks}
\label{sec:final-remarks}

If a given baseline
algorithm works over general compact convex sets \(P\),
then so does the lazified version. In fact, as the lazified algorithm
runs, it produces a polyhedral approximation of the set \(P\)
with very few vertices (subject to optimality vs.~sparsity tradeoffs; see
\cite[Appendix C]{jaggi2013revisiting}).

Moreover, the weak separation oracle does not need to return extreme
points. All algorithms also work with maximal solutions that are not
necessarily extremal (e.g., lying in a higher-dimensional
face). However, in that case we lose the desirable property that the
final solution is a sparse convex combination of extreme points
(typically vertices in the polyhedral setup).

We would also like to briefly address potential downsides of our
approach. In fact, we believe the right perspective is the following:
when using the lazy oracle over the LP oracle, we obtain potentially
\emph{weaker} approximations \(v_t - x_t\) of the true gradient
\(\nabla f(x_t)\) compared to solving the actual LP, but the
computation might be much faster. This is the tradeoff that one has to
consider: working with weaker approximations (which implies
potentially less progress per iteration) vs.~potentially significantly
faster computation of the approximations. If solving the LP is
expensive than lazification will be usually very beneficial, if the LP
is very cheap as in the case of \(P = [0,1]^n\) or \(P= \Delta_n\)
being the probability simplex, then lazification might be slower. 

A related remark in this context is that once the lazified algorithm
has obtained vertices \(x_1, \dots, x_m\) of \(P\), so that the
minimizer \(x^*\) of \(f\) satisfies
\(x^* \in \text{conv}\{x_1, \dots, x_m\}\), then from that point
onwards no actual calls to the true LP oracle have to be performed
anymore for primal progress and the algorithm will only use cache
calls; the only remaining true LP calls are at most a
logarithmic number for dual progress updates of the \(\Phi_t\).

\section*{Acknowledgements}
\label{sec:acknowledgements}

We are indebted to Alexandre D'Aspremont, Simon Lacoste-Julien, and
George Lan for the helpful discussions and for providing us with
relevant references. Research reported in this paper was partially
supported by NSF CAREER award CMMI-1452463.

\bibliographystyle{abbrvnat}
\bibliography{bibliography}
\end{document}